\documentclass[11pt,a4paper]{amsart}

\pdfoutput=1

\usepackage{booktabs}
\usepackage{lscape}

\usepackage{bigints}
\usepackage{framed}
\usepackage{mdframed}
\usepackage{enumerate}
\usepackage[inline]{enumitem}
\usepackage[T1]{fontenc}

\usepackage[nodisplayskipstretch]{setspace}
\usepackage{color}
\usepackage[dvipsnames]{xcolor}
\usepackage{cancel}
\usepackage{soul}

\usepackage{siunitx}
\usepackage{graphicx}
\usepackage{float}
\usepackage{rotating}
\usepackage{subcaption}
\usepackage{overpic}
\usepackage[all]{xy}
\usepackage{tikz}
\usetikzlibrary{arrows,matrix,positioning,calc,automata,patterns}
\usepackage{booktabs}
\usepackage{dcolumn}
\usepackage{multirow}
\usepackage{diagbox}
\usepackage{tabularx}
\usepackage{verbatim}
\usepackage{listings}
\usepackage[ruled,vlined]{algorithm2e}
\usepackage{fancyvrb}
\usepackage{hyperref}
\usepackage{natbib}
\usepackage[normalem]{ulem}

\usepackage{algorithmic}
\newtheorem*{rep@theorem}{\rep@title}
\newcommand{\newreptheorem}[2]{%
\newenvironment{rep#1}[1]{%
 \def\rep@title{#2 \ref{##1}}%
 \begin{rep@theorem}}%
 {\end{rep@theorem}}}
\makeatother

\newtheorem{theorem}{Theorem}[section]

\newtheorem{assumption}{Assumption}[section]
\newtheorem{corollary}{Corollary}[section]

\newreptheorem{theorem}{Theorem}
\newreptheorem{lemma}{Lemma}
\newreptheorem{corollary}{Corollary}
\newreptheorem{example}{Example}
\newreptheorem{remark}{Remark}

\begin{document}

\newcommand{\bfzero} {\boldsymbol{0}}
\newcommand{\bfalpha} {\boldsymbol{\alpha}}
\newcommand{\bfbeta} {\boldsymbol{\beta}}
\newcommand{\bfgamma} {\boldsymbol{\gamma}}
\newcommand{\bfdelta} {\boldsymbol{\delta}}
\newcommand{\bfXi} {\boldsymbol{\xi}}
\newcommand{\bfmu} {\boldsymbol{\mu}}
\newcommand{\bfsigma} {\boldsymbol{\sigma}}
\newcommand{\bfeta} {\boldsymbol{\eta}}
\newcommand{\bfzeta} {\boldsymbol{\zeta}}
\newcommand{\bfvarphi} {\boldsymbol{\varphi}}
\newcommand{\bftau} {\boldsymbol{\tau}}
\newcommand{\bflambda} {\boldsymbol{\lambda}}
\newcommand{\bfLambda} {\boldsymbol{\Lambda}}
\newcommand{\bfpsi} {\boldsymbol{\psi}}
\newcommand{\bfphi} {\boldsymbol{\phi}}
\newcommand{\bfpi} {\boldsymbol{\pi}}
\newcommand{\bfvarpsi} {\boldsymbol{\varpsi}}
\newcommand{\bfnu} {\boldsymbol{\nu}}
\newcommand{\bftheta} {\boldsymbol{\theta}}
\newcommand{\bfTheta} {\boldsymbol{\Theta}}
\newcommand{\bfSigma} {\boldsymbol{\Sigma}}
\newcommand{\bfDelta} {\boldsymbol{\Delta}}
\newcommand{\bfOmega} {\boldsymbol{\Omega}}
\newcommand{\bfUpsilon} {\boldsymbol{\Upsilon}}

\newcommand{\bfR} {\mathbf{R}}
\newcommand{\bfGamma} {\mathbf{\Gamma}}
\newcommand{\bfC} {\mathbf{C}}
\newcommand{\bfP} {\mathbf{P}}
\newcommand{\bfT} {\mathbf{T}}
\newcommand{\bfI} {\mathbf{I}}
\newcommand{\bfZ} {\mathbf{Z}}
\newcommand{\bfY} {\mathbf{Y}}
\newcommand{\bfX} {\mathbf{X}}
\newcommand{\bfV} {\mathbf{V}}
\newcommand{\bfE} {\mathbf{E}}
\newcommand{\bfW} {\mathbf{W}}
\newcommand{\bfD} {\mathbf{D}}
\newcommand{\bfA} {\mathbf{A}}
\newcommand{\bfB} {\mathbf{B}}
\newcommand{\bfK} {\mathbf{K}}
\newcommand{\bfQ} {\mathbf{Q}}
\newcommand{\bfU} {\mathbf{U}}
\newcommand{\bfS} {\mathbf{S}}
\newcommand{\bfPsi} {\mathbf{\Psi}}
\newcommand{\bfPhi} {\mathbf{\Phi}}
\newcommand{\bft} {\mathbf{t}}
\newcommand{\bfs} {\mathbf{s}}
\newcommand{\bfx} {\mathbf{x}}
\newcommand{\bfy} {\mathbf{y}}
\newcommand{\bfz} {\mathbf{z}}
\newcommand{\bfb} {\mathbf{b}}
\newcommand{\bff} {\mathbf{f}}
\newcommand{\bfd} {\mathbf{d}}
\newcommand{\bfepsilon} {\mathbf{\epsilon}}
\newcommand{\B}{\mathsf{B}}

\renewcommand{\Pr}{\mathsf{P}}
\newcommand{\p}{\mathsf{P}}
\newcommand{\E}{\mathsf{E}}
\newcommand{\M}{\mathsf{M}}
\newcommand{\Cov}{\mathsf{Cov}}
\newcommand{\Cor}{\mathsf{Cor}}
\newcommand{\reals}{\mathbb{R}}
\newcommand{\naturals}{\mathbb{N}}
\newcommand{\dd}{\mbox{d}}

\newcommand{\ground}{\mathcal{G}}

\newcommand{\kl}{\mathsf{KL}}
\newcommand{\nbd}{\mathsf{nbd}}
\newcommand{\normal}{\mathsf{N}}
\newcommand{\tst}{\mathsf{T}}
\newcommand{\DP}{\mathsf{DP}}
\newcommand{\Dir}{\mathsf{Dir}}
\newcommand{\Gam}{\mathsf{G}}
\newcommand{\IGam}{\mathsf{IG}}
\newcommand{\bin}{\mathsf{Bin}}
\newcommand{\geo}{\mathsf{Geo}}
\newcommand{\Exp}{\mathsf{Exp}}
\newcommand{\Wis}{\mathsf{Wis}}
\newcommand{\IWis}{\mathsf{IW}}
\newcommand{\Poi}{\mathsf{Poi}}
\newcommand{\bet}{\mathsf{Beta}}
\newcommand{\Uni}{\mathsf{Uni}}
\newcommand{\GWis}{\mathsf{Wis}}
\newcommand{\Mult}{\mathsf{Multinom}}
\newcommand{\SB}{\mathsf{SB}}
\newcommand{\bern}{\mathsf{Bernoulli}}
\newcommand{\CAR}{\mathsf{CAR}}
\newcommand{\MCAR}{\mathsf{MCAR}}
\newcommand{\cone}{\mathsf{P}}
\newcommand{\InvGauss}{\mathsf{InvGauss}}
\newcommand{\indep}{\perp \!\!\! \perp}

\title[Human Activity Space Estimation]{An Object-Oriented Spatial Statistics Approach for Human Activity Space Estimation}
\thanks{CONTACT H. Wu. Email: wuhy0902@uw.edu}

	\author[Wu, Chen and Dobra]{Haoyang Wu\textsuperscript{a},
			Yen-Chi Chen\textsuperscript{a} and Adrian Dobra\textsuperscript{a}}
    \address{\textsuperscript{a}Department of Statistics, University of Washington, Seattle, WA, USA}

\begin{abstract}
Human activity spaces are shaped by individual mobility and the built environment, motivating statistical methods that integrate GPS observations with GIS representations of places and routes. We propose a novel methodology to estimate activity spaces in built environments from GPS data within the Object Oriented Spatial Statistics framework. We characterize daily mobility through the distribution of time across spatial polygons and road segments, aiming to capture entity-specific time-use fractions and level-$\gamma$ activity spaces. We develop a time-weighted estimator to handle irregularly sampled GPS observations. We derive an error bound that quantifies the effects of measurement error, nearest-entity misclassification, temporal gaps, boundary crossings, and day-to-day variability. We also develop a map-augmented representation of daily activity patterns, a dwell-time-weighted distance for clustering daily trajectories, and polygon- and road-based stability summaries. Simulation studies and a real-data application demonstrate that the proposed framework recovers concentrated stationary anchors, interpretable travel corridors, and distinct stabilization behavior for dwelling and movement components, supporting the benefits of weighting under irregular sampling.
\end{abstract}
            
\maketitle

\date{\today} 

\tableofcontents

\section{Introduction} 

This article focuses on the statistical estimation of human activity spaces, which represent areas of direct contact during daily activities \citep{entwisle2007putting,golledge-stimson-1997}. Activity spaces measure individual spatial behavior and capture experiences through observed location choices \citep{golledge-1999}. People do not move randomly in space \citep{hagerstrand-1970}. Due to preferences, needs, knowledge, and movement constraints, individuals concentrate their visits around one or more anchor locations that serve as origin and destination hubs. Anchor locations hold material or symbolic significance, including home, work, a child's school, favorite markets, grocery stores, entertainment venues, and airports. The likelihood of visiting a specific location decreases with distance from anchor locations and depends on its position relative to common travel routes. The shape, structure, and extent of activity spaces result from the spatial arrangement of anchor locations and the routes traveled between them \citep{schonfelder-axhausen-2003}. 

Recent developments in global positioning systems (GPS) for wearable technology, such as smartphones, have garnered significant interest from scientists studying environmental influences on various population groups \citep{RN143,RN144,RN146,RN149,RN145,zenk2011activity,RN147,wiehe2013adolescent,entwisle2007putting}. \citet{mazimpaka2016trajectory} documents over 100 studies across 20 disciplines that collect and analyze time-stamped GPS location data from humans. This data is essential for understanding where people spend their time during daily activities and for establishing relationships with socioeconomic outcomes, crime victimization, and both physical and mental well-being. Advances in GPS tracking enable precise analysis of human mobility patterns \citep{wang_urban_2018}. Early mobility studies show that individuals often visit locations beyond their neighborhoods. The neighborhoods they engage with differ from their homes \citep{jones_redefining_2014}. People venture outside their residential areas for various reasons, including school, work, and recreation. This behavior fosters meaningful social connections between neighborhoods \citep{phillips_social_2021}.

Several human mobility studies have utilized travel log diaries from the Los Angeles Family and Neighborhoods Study (L.A.FANS) \citep{lafans-2006}. Travel log diaries enable researchers to map individual movements throughout daily routines; however, they have limitations \citep{schonfelder-axhausen-2003, schonfelder-axhausen-2004}.  First, large-scale samples are costly to collect. The costs of large-scale surveys \citep{olson_survey_2021} increase due to the burden of collecting open-ended data, which necessitates complex GIS processing \citep{wong_measuring_2011}. Second, travel logs are susceptible to various biases. Respondents may not document all visited locations \citep{jones_redefining_2014, bricka_comparative_2006}. Recall-based methods can be biased, because participants may forget details of events or experiences. These omissions can distort a respondent's activity space. 

Key advances in GPS technology have made location data collection more affordable and efficient. With the widespread adoption of smartphones, a significant portion of the population now tracks their mobility through digital trace data \citep{hughes_inferring_2016}. Digital trace data includes methods for obtaining location information from mobile devices, such as call detail records, location-enabled apps, social media activity, WiFi connections, and web searches \citep{hughes_inferring_2016}. However, GPS data from smartphones is limited to device owners, raising concerns about its representativeness and the insights it provides into behavioral patterns \citep{zagheni_demographic_2015}. The public availability and scale of this data make it a valuable resource for academic research. Researchers can utilize public data from social media platforms like Twitter to gather information on activity spaces at a larger scale than traditional travel diaries. This enables studies to be conducted at the neighborhood level rather than solely on an individual basis, facilitating assessments at both neighborhood and metropolitan levels. 

Early attempts to estimate human activity spaces from GPS data involved elliptical regions connecting two anchor locations \citep{sherman2005suite, newsome1998urban}. Models defined activity spaces as convex hulls that encompass an individual's recorded locations \citep[e.g.][]{buliung2006urban, fan2008urban}. More recent methods, such as kernel density estimation in a two-dimensional Cartesian space \citep[e.g.,][]{schonfelder-axhausen-2003,schonfelder-axhausen-2004}, assume uniform distributions of activity within grid cells. This implies that estimates of activity space will be influenced by the chosen grid cell size— a clear example of bias induced by the modifiable areal unit problem \citep{maup-2008}.

An important step forward has been made by estimation methods of activity spaces based on the shortest-path spanning trees of a road network. An individual's travel routes are estimated by projecting their locations onto a relevant road network and connecting each consecutive pair of origin-destination locations associated with trips via the shortest road path \citep{schonfelder-axhausen-2003,schonfelder-axhausen-2004}. \citet{golledge-1999} asserts that road networks shape people's perceptions and knowledge of places, indicating that activity spaces should reflect the paths taken by travelers through the network. Thus, an individual's activity space is represented by the spanning tree that encompasses the area defined by the union of the shortest paths connecting consecutively visited locations. Anchor locations and frequently used segments of the road network can be identified based on visitation frequencies.

In this paper, we propose a novel approach to estimating activity spaces from GPS data, integrating key features of existing methods into a coherent framework. We assume individuals spend their time within a set of spatial polygons or along a road network connecting these polygons. Our goal is to estimate the time spent in each polygon and road network component. The spatial polygons can vary in size and shape, representing bounded structures (e.g., home, office, gym), public spaces, or elements of the built environment. Our methodology falls within Object Oriented Spatial Statistics (O2S2) \citep{cressie-et-2006,NKDEBook,menafoglio-secchi-2017}, where data are points in an appropriate functional embedding space. By including spatial polygons, we create a more accurate framework for human activity spaces compared to methods (e.g., shortest-path spanning trees) that limit activities to road networks. We refer to this O2S2 version of the human mobility polygon-network (PN, henceforth) as activity spaces.

Our proposed methodology uses readily available GIS information about the built environment, which constrains daily activities to specific patterns that cannot be accurately recovered if individuals are assumed to move freely in two-dimensional space, such as the activity space KDE methods of \citet{schonfelder-axhausen-2003,schonfelder-axhausen-2004,chen2020measuring}. We note that most current methods for analyzing point patterns of a linear network \citep[e.g.,][]{NKDEBook,baddeley-et-2021} cannot directly estimate human activity spaces, as they assume GPS locations represent points on a linear network from which a KDE can be derived. This is related to the frequency with which GPS locations are generated; a greater number of GPS locations observed on a linear network segment does not necessarily indicate more time spent by an individual in that segment. Mobile devices that record GPS locations generate more data when traveling at higher speeds. Conversely, when stationary for longer intervals, mobile devices tend to generate fewer GPS records. Our work creates interpretable activity distributions of \citet{Adrianpaper} that reflect the time individuals spend in grid cells within a rectangular spatial window. We ensure that the support of these distributions is embedded in selected GIS features, spatial polygons, and road networks.

The paper is organized as follows: Section \ref{sec:data} describes the GPS data that motivates and illustrates our developments. Section~\ref{sec: mobility pattern} introduces the polygon-network framework, defines latent trajectory and activity-space objects, and presents the time-weighted estimator with its theoretical error bound. Section~\ref{sec:SMM-map} develops a map-augmented movement model along with methods for clustering daily activity trajectories and assessing the temporal stability of polygon-based and road-based activity spaces. Section~\ref{sec:real-data} illustrates our methodological framework with the GPS data described in
Section \ref{sec:data}. Section~\ref{sec:conclusion} concludes with a discussion of the main findings and future research directions. Supplementary material provides details on preprocessing GIS data, simulation experiments, and the proof of the main theoretical result.

\section{GPS data}\label{sec:data}

We analyze GPS data from a longitudinal study tracking the daily movement patterns of over 150 individuals using GPS-enabled smartphones over one month. Study participants were selected through convenience sampling from multiple urban areas within the same country. The positional information recorded by smartphones was securely transmitted to the study database on a secure server via cellular or wireless networks using state of the art encryption techniques. Positional data includes timestamps, smartphone IDs, latitude and longitude coordinates, and accuracy information regarding the reported coordinates (e.g., satellite connectivity). This GPS study design was approved by the relevant biomedical research ethics committees. Participants signed written consent allowing the study team to record GPS data from their smartphones. During GPS tracking, participants maintained their regular daily routines. 

For this analysis, we use data from a single randomly selected participant with 7,248 GPS records collected over $n=25$ days, averaging about 290 observations daily. GIS data for the geographic region where this participant was active were processed as detailed in the Supplementary Material. All identifying information about specific locations visited by this participant has been removed from the following maps and summaries. We added random distance-based noise to the raw GPS data translated from a neutral reference point. We report relative times instead of exact timestamps.

The GPS observation timestamps are unevenly spaced in time, as shown in the left panel of Figure~\ref{fig:example_summary}, which presents the empirical distribution of timestamps over the observation window. 
The right panel of Figure~\ref{fig:example_summary} depicts the distribution of consecutive time gaps between observations, truncated at the 90th percentile to emphasize the predominant patterns. 
This irregular sampling pattern highlights the necessity of employing statistical methods that explicitly accommodate nonuniform temporal spacing, rather than depending on simple averaging procedures. These descriptive summaries demonstrate that, even for a single participant’s path, there is considerable variation in both the temporal granularity and the spacing of observations. This, in turn, supports the use of methods capable of accommodating irregularly spaced spatiotemporal data, especially when combining them with GIS sources such as polygons and road networks to describe individual activity patterns.

\begin{figure}[ht]
	\centering
	\includegraphics[width=0.45\linewidth]{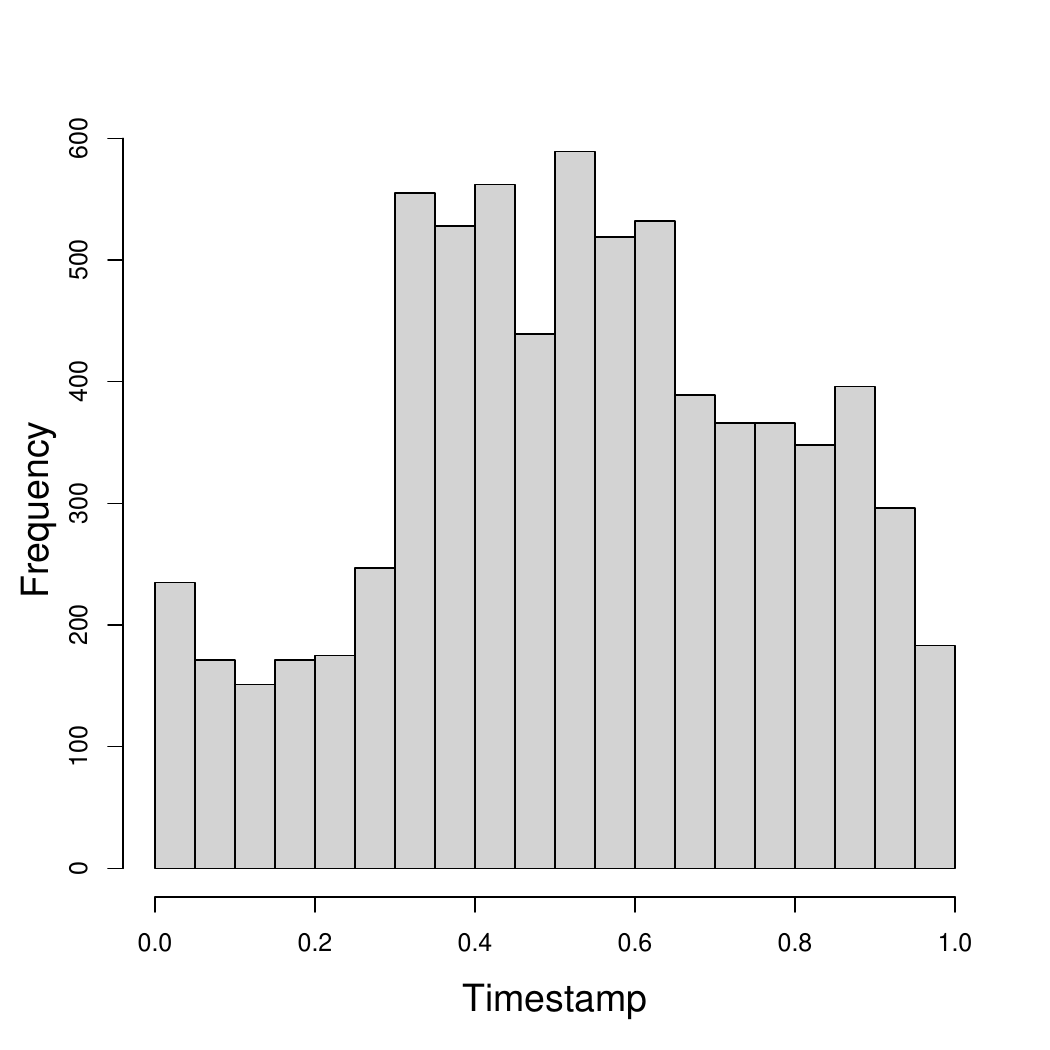}	\includegraphics[width=0.45\linewidth]{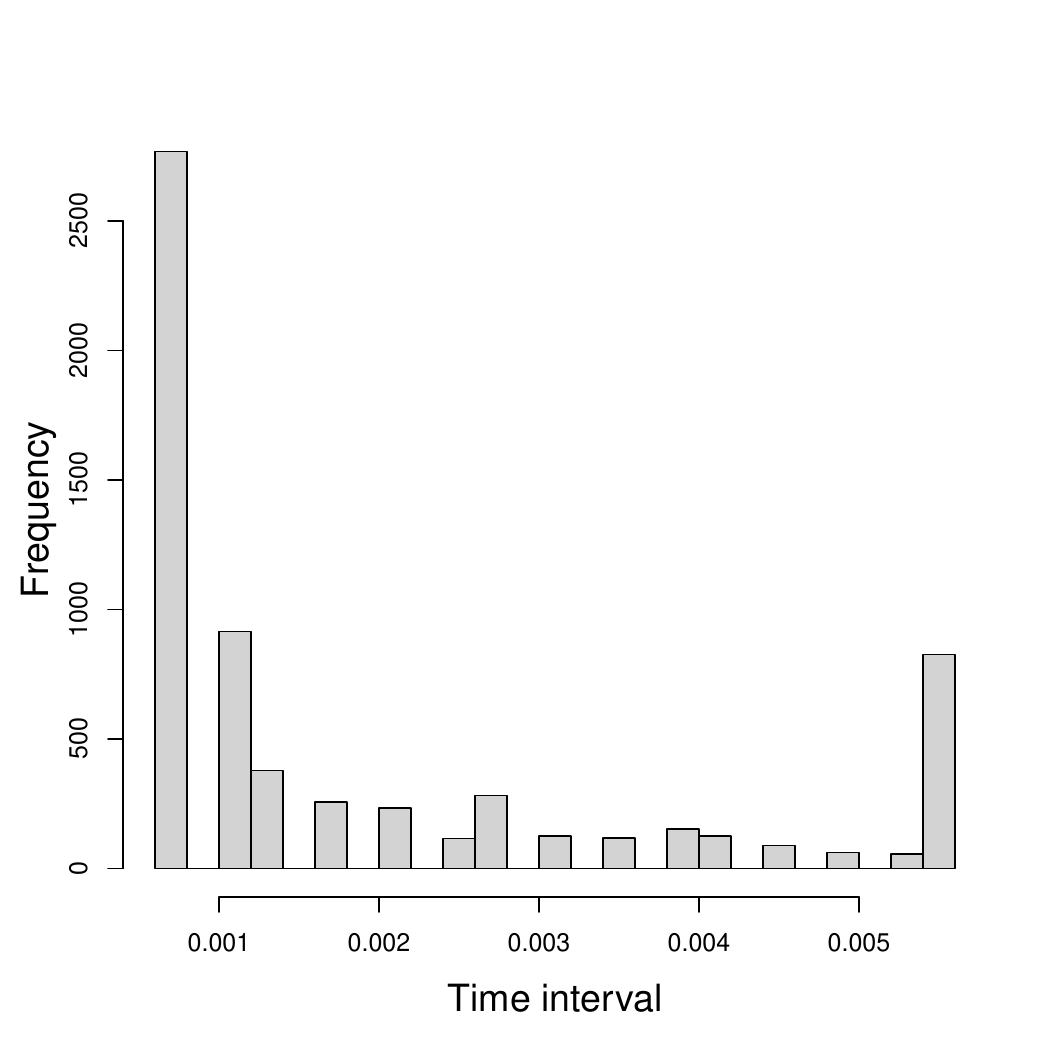}
	 
	\caption{Summary information for the GPS data recorded for the selected study participant. The left panel displays the histogram of recorded timestamps, showing irregular sampling across the observation window. The right panel shows the histogram of time intervals between successive observations, truncated at the 90th percentile of all intervals to suppress the influence of extreme gaps.}
	\label{fig:example_summary}
\end{figure}  

We based our analysis on accurate GIS spatial data, including unmodified polygon entities and road segments. Thus, all nearest-entity assignments, time-use proportions, activity-space estimates, clustering summaries, and stability analyzes are derived from real-world polygon-network representations. However, the map figures were produced using privacy-preserving graphical layers. For visualization, road segments without nearby GPS observations were randomly thinned, and polygonal entities were displayed as equal-sized squares centered at their centroids. These modifications affect only the rendered figures, not the underlying estimates or conclusions. Details are provided in the Supplementary Material.

\section{Our modeling framework}\label{sec: mobility pattern}

We consider GPS data collected from a single individual during \(n\) non-overlapping time periods of equal length.  We treat one period as a day; however, our methodology applies unchanged to coarser intervals (e.g., weeks, weekdays, weekends) or finer intervals (e.g., 07:00-17:00). We represent the individual's GPS data as random spatial locations recorded at fixed timestamps $(X_{i,j},\, t_{i,j}) \ \in\  \mathbb{R}^{2}\times [0,1]$, for $i = 1,\dots,n$ and $j = 1,\dots,m_{i}$, where \(X_{i,j}\) is the planar spatial coordinate of the \(j\)-th observation on day \(i\); \(m_{i}\) is the total number of observations on day \(i\); and \(t_{i,j}\) is its timestamp, rescaled so that the beginning and end of each day map to 0 and 1, respectively: $0 \leq t_{i,1} < t_{i,2} < \dots < t_{i, m_i} \leq 1$, for $i=1,\dots,n$.

The observed locations belong to a rectangular spatial window \(\mathcal{W}\subset\mathbb{R}^{2}\). This spatial window contains a finite number of polygons $\mathcal{A}=\{a_1,\ldots,a_{n_A}\}$ and a road network $\mathcal{L}$, where  each $a_i\subset \mathcal{W}$ and $\ell_j\subset \mathcal{W}$ — see Figure~\ref{fig:map}. We model $\mathcal{L}$ as a planar network of undirected segments $\{l_1,...,l_{n_L}\}$ that intersect only at their endpoints \citep{NKDEBook}.  We define the set $\mathcal{E}=\mathcal{A}\,\cup\,\mathcal{L}$ of polygons and segments, referred to as entities.

\begin{figure}[H]
	\centering
    \includegraphics[width=0.8\linewidth]{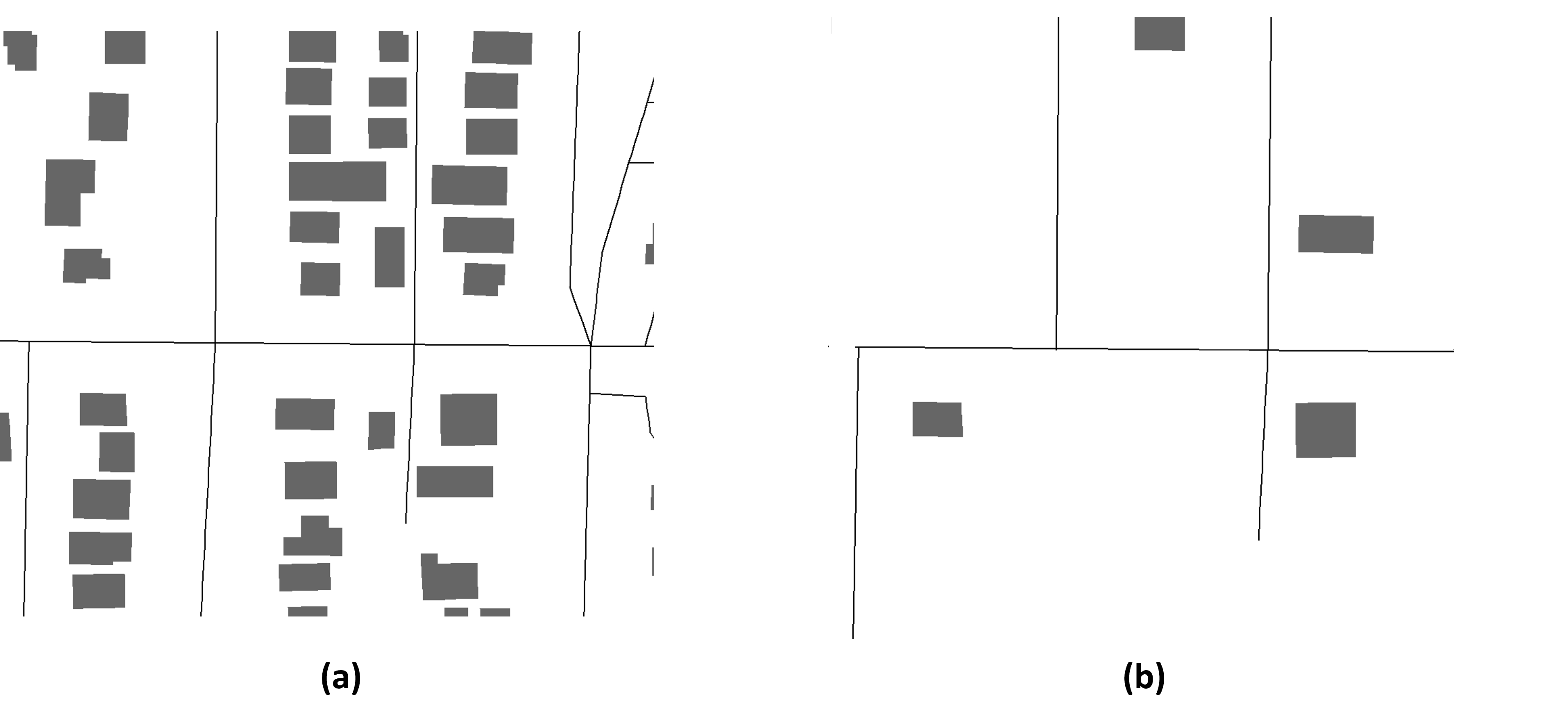}
	\caption{(a) An example spatial window containing spatial polygons and a road network; (b) An example polygon-network activity space (subset of all spatial polygons and road network with high activities) in this spatial window.}
	\label{fig:map}
\end{figure}

The individual performs daily living activities within the spatial polygons $\mathcal{A}$ and can move between them on the road network $\mathcal{L}$. We assume that this individual spends negligible time outside the polygon-network (PN, henceforth) space $E=A\cup L\subset \mathbb{R}^2$, where $A = \bigcup_{i=1}^{n_A} a_i$ and $L=\bigcup_{i=1}^{n_L} \ell_i$ denote the 2D regions covered by the polygons and road segments, respectively. This assumption is valid if all GIS configurations of polygons and segments available in a spatial window (e.g., bounded structures, parks, parking lots, roads, and alleys) are included in the PN space $E$. While including all available GIS polygon-network features might be feasible in rural areas, it is impractical for urban environments with complex built structures that involve many layers and numerous features. Dense GIS layers with too many spatial polygons and road segments inflate storage demands and distance-to-entity calculations, leading to a surplus of entities where the example individual likely spends little time. Consequently, this results in significant increases in computation time and more challenging results to interpret. Thus, the PN space should only include polygons and segments relevant to the individual’s daily living patterns while remaining sufficiently dense to ensure that the time spent outside \(E\) is negligible. Figure~\ref{fig:dense_and_sparse} contrasts two extremes: the left panel shows
an overly dense PN space containing too many entities, while the
right panel displays a PN space with fewer entities than necessary to represent the locations visited by the selected study participant. The middle panel illustrates a PN space that includes sufficient entities to encompass the GPS data of the selected study participant. For privacy-preserving visualization, the displayed road-network layers have been thinned; this modification is for rendering the figure only and does not affect the PN-space construction or subsequent analysis.

\begin{figure}[H]
	\centering
  \includegraphics[width=0.32\linewidth]{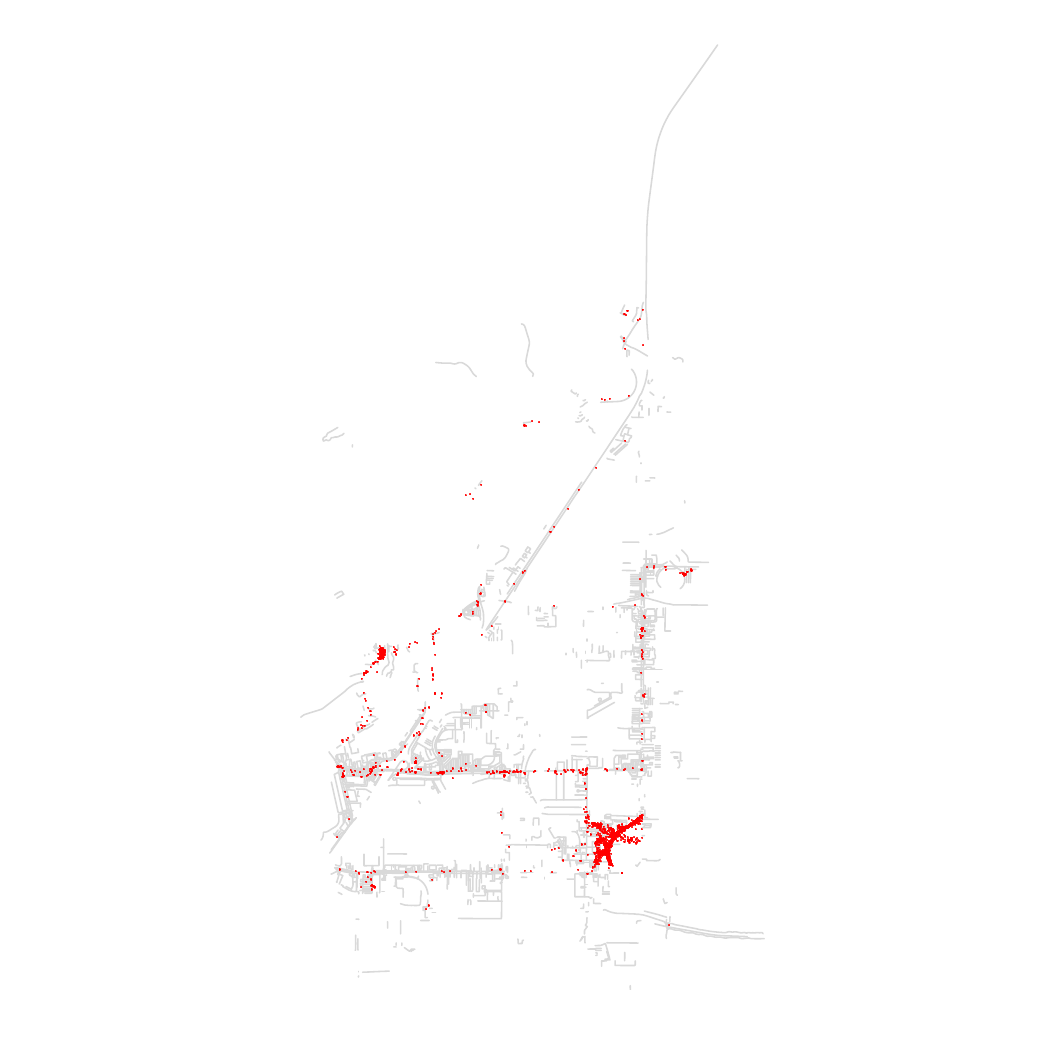}
	\includegraphics[width=0.32\linewidth]{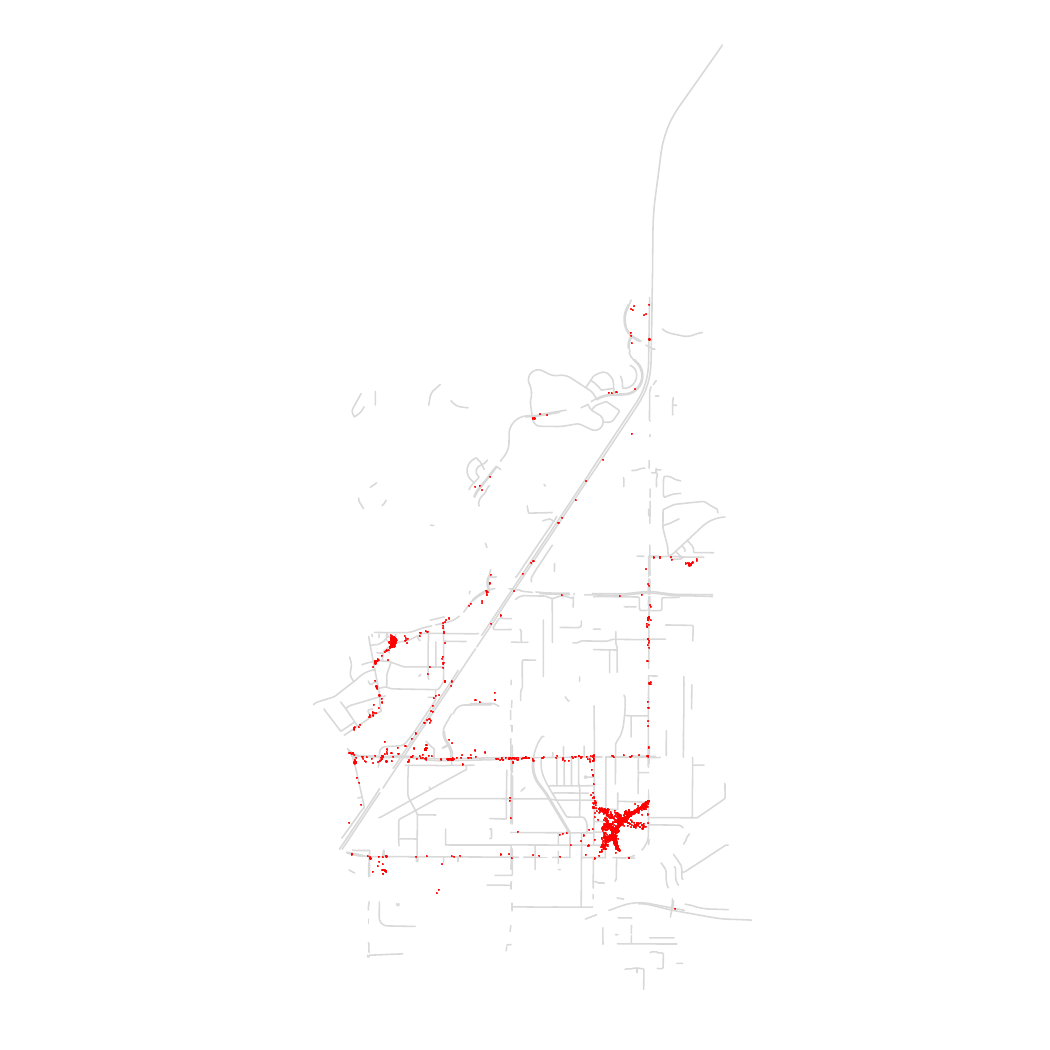}	
   	\includegraphics[width=0.32\linewidth]{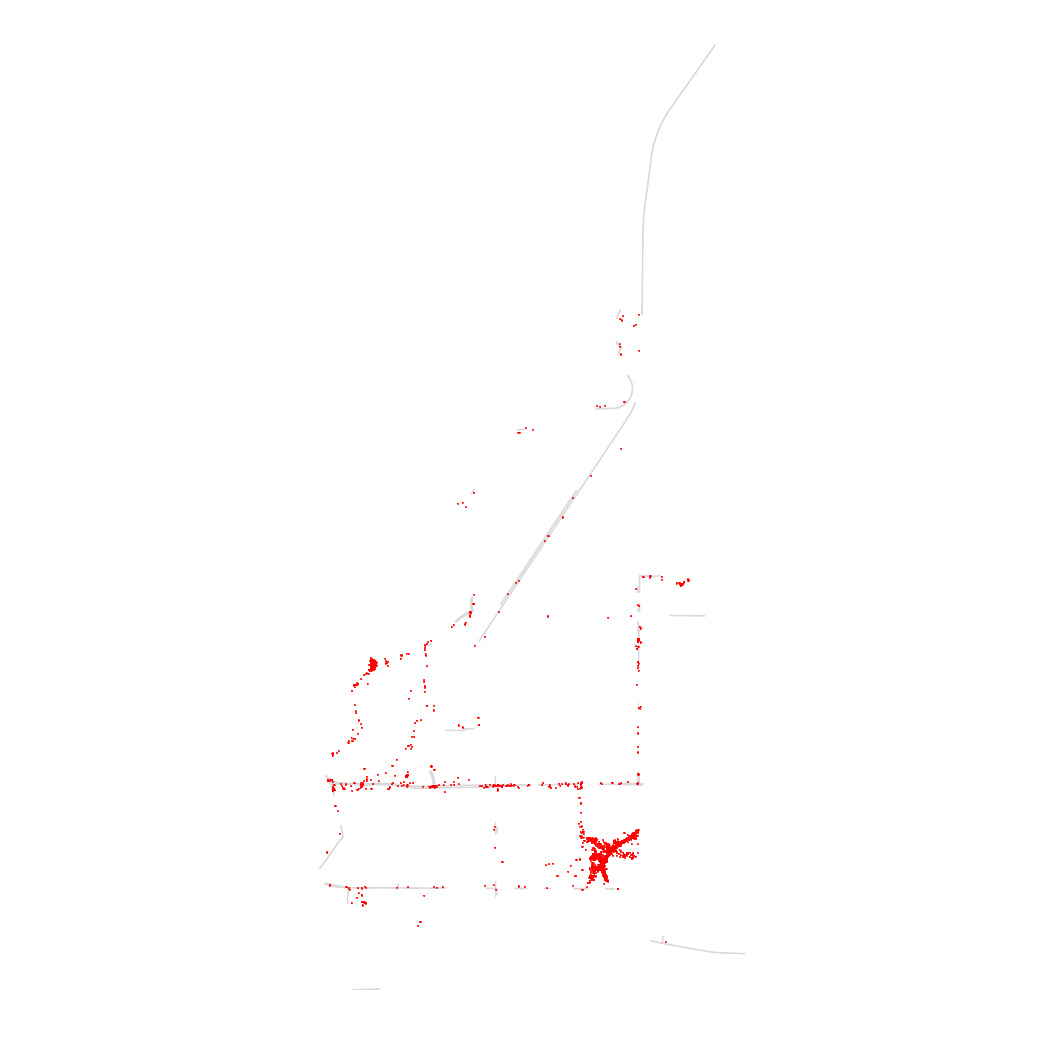}
    \caption{Illustration of the selection of an appropriate PN space shown in grey with respect to a set of GPS records shown in red. Left panel: an example of a PN space that is too dense. Right panel: an example of a PN space that is too sparse. Middle panel: an example of a PN space that is judiciously chosen to match the GPS data. Note that the road-network layers in each panel are privacy-preserving renderings used for visualization only.}
	\label{fig:dense_and_sparse}
\end{figure}

We follow two key rules for identifying appropriate PN spaces in simulated and real data analysis. We retain polygons and segments covering the majority of GPS-recorded locations. Locations far from the retained entities are discarded. Second, we remove road segments entirely within polygons classified as entities. These rules balance computational efficiency and ease of interpretation while capturing meaningful PN activity spaces.

\subsection{Polygon-network latent trajectory}
\label{subsec:trajectory}

In our framework, an individual can only be in locations belonging to the PN space $E$. We model the individual's movement on day \(i\in \{1,2,\ldots,n\}\) using a map $S_i : [0,1]\longrightarrow E$. We refer to \(S_i\) as the (latent) trajectory for day \(i\). Daily trajectories \( S_1,\dots,S_n \) are i.i.d. realizations of a distribution \(P_{\mathcal{S}}\) in the space $\mathcal{S}$ of continuous functions mapping $[0,1]$ into $E$. Each day, the individual follows a random trajectory consistent with $E$, independent of spatial trajectories on other days. We assume GPS data are recorded with measurement error \citep{ranacheret-2016}: $X_{i,j} \;=\; S_i(t_{i,j}) + \varepsilon_{i,j}$, for $j = 1,\dots,m_i$ and $i=1,\ldots,n$. The errors \(\epsilon_{i,j}\) follow a bivariate Normal distribution with mean zero and location-dependent variance:  $\varepsilon_{i,j}\mid\{X^{\mathrm{true}}_{i,j}\in e\}\ \sim\ \normal_2(\bfzero,\sigma_e^2 I_2)$, with \(\sigma_e^2\) constant within each entity. This distribution captures spatial heterogeneity (e.g., open sky versus urban canyon) while maintaining local isotropy. For tractability, we treat \(\{\varepsilon_{i,j}\}\) as i.i.d. Measurement error in data from Global Navigation Satellite Systems (GNSS) may exhibit temporal correlation over seconds to tens of seconds due to satellite geometry and multipath, but this dependence is generally weak at typical mobility sampling scales and is ignored here \citep{park2013dgps, lapadat2021sensors}. The true location without error at time \(t_{i,j}\) is \(X^{\mathrm{true}}_{i,j} = S_i(t_{i,j})\). We note that the true locations visited by the example individual are assumed to be unknown.

\begin{figure}[H]
	\centering
	\includegraphics[width=.7\linewidth]{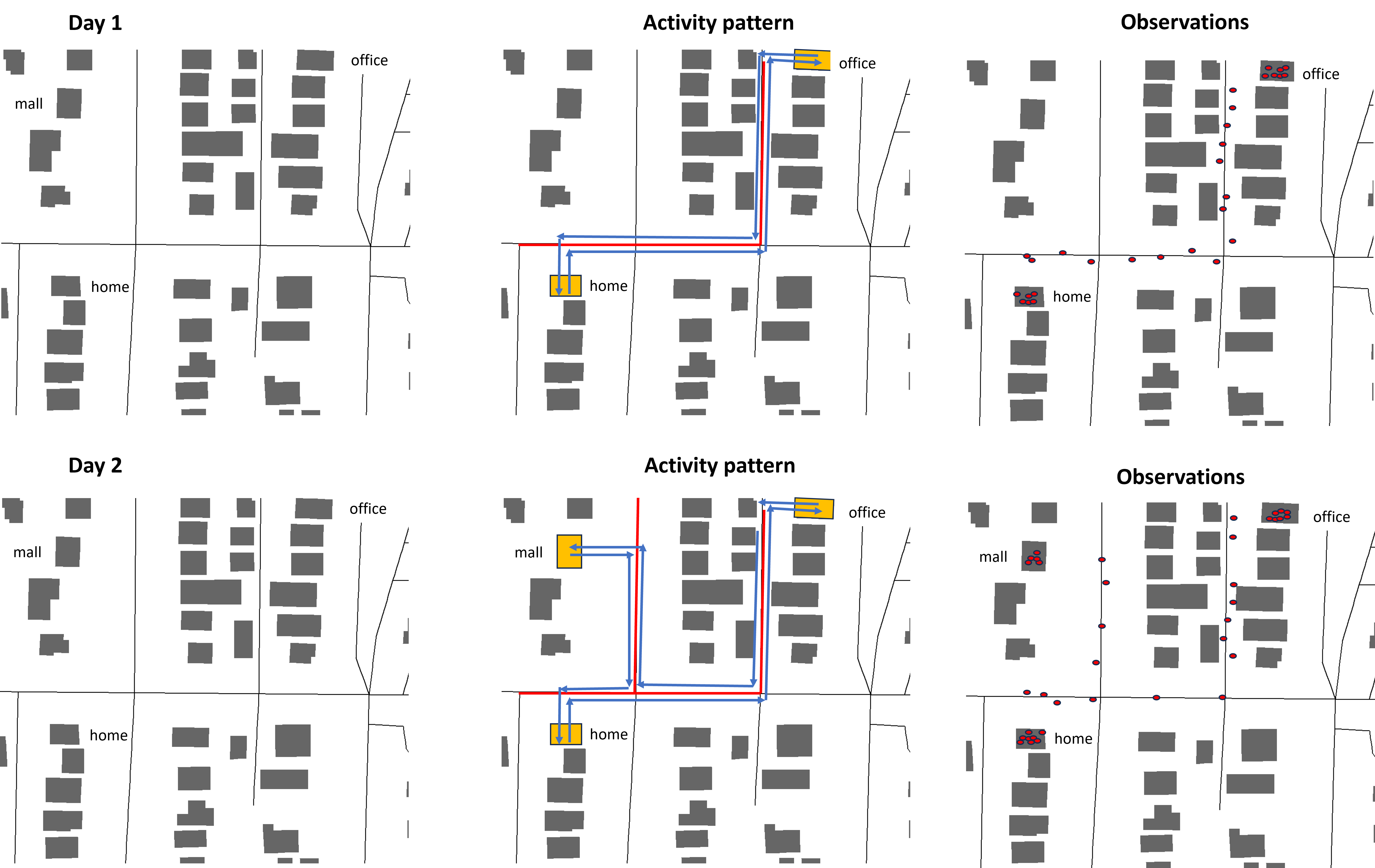}
	\caption{Illustration of the proposed data generating mechanism. Left column: study window with relevant spatial polygons and road network in two different days. Middle column: two sample latent trajectories comprising polygon dwell times (yellow) and path segments (blue, with a highlighted red sub-segment). The line segments the individual uses when moving between polygons are marked in red. Right panel: noisy GPS observations (red dots) recorded for these two daily trajectories.}
	\label{fig:DGP}
\end{figure}

Figure~\ref{fig:DGP} shows an example of our proposed data generating mechanism. The middle column displays the latent trajectories of an individual over two days, during which they remain within two polygons and move between them along the red-marked road segments. The right column illustrates the recorded GPS observations, which trace the true paths but slightly deviate due to measurement errors.

\subsection{Polygon-network activity space}\label{sec:activity-space}

Let $S_d\sim P_{\mathcal{S}}$ such that $S_d\in \mathcal{S}$ be an example individual's latent trajectory on a particular day $d$. Consider a measurable subset $B\subseteq E$. The individual might never visit $B$, or they might visit it once, twice, or several times. Each visit to $B$ might vary in duration. The probability that the individual is present in $B$ at time $t$ during day $d$ is \(\Pr\bigl\{S_d(t)\in B\bigr\}\). 
Let \(U\sim  {\sf Uniform}[0,1]\) be independent of the latent trajectory \(S_d\). Given the trajectory $s$ on day $d$, the time proportion the individual spends in $B$ during day $d$ is 
\begin{eqnarray*}
    \Pr\bigl\{S_d(U)\in B\mid S_d=s\bigr\} &=& \int_0^1 I\bigl\{ s(t)\in B \bigr\}\, dt.
\end{eqnarray*}

The mean proportion of time spent in $B$ across all periods is
\begin{eqnarray}
\Pr\bigl\{S_d(U)\in B\bigr\} = \mathbb{E}_{S_d}\bigl(\Pr\bigl\{S_d(U)\in B \mid S_d\bigr\}\bigr) &=&\int_{s\in \mathcal{S}}\int_0^1 I\bigl\{s(t)\in B\bigr\}\, dtdP_{\mathcal{S}}(s). 
\label{eq:avepnasdistribution}
\end{eqnarray}
The individual's activity space relative to the PN space $E$ is called the polygon-network (PN) activity space, fully determined by the distribution \eqref{eq:avepnasdistribution}. It includes all entities visited by the individual. Considering the collection of entities $\mathcal{E}=\mathcal{A}\cup\mathcal{L}=\{ e_1,e_2,\ldots,e_{n_A+n_L}\}$ that includes all road segments and polygons, let 
\begin{eqnarray}
    T_e&=\Pr\bigl\{S_d(U)\in e\bigr\}=&\mathbb{E}_{S_d}\bigl(\Pr\bigl\{S_d(U)\in e\mid S_d\bigr\}\bigr)
    \label{eq:meantimespententity}
\end{eqnarray}
\noindent  denote the mean proportion of time an individual spends in entity $e\in \mathcal{E}$. The PN activity space is $\bigcup_{\{e\in \mathcal{E}:T_e>0\}}\{e\}$. The level $\gamma$ PN activity space, denoted by $\mathcal{AS}_{\gamma}$, is defined as the subset of entities with the smallest number of elements where an individual spends at least $100\gamma \%$ of their time. Denote $\mathcal{Q}_{\gamma}=\left\{ Q\subseteq \mathcal{E}:\sum\limits_{e\in Q}T_e\ge \gamma\right\}$. Then $\mathcal{AS}_{\gamma}$ is the element of $\mathcal{Q}_{\gamma}$ such that $\mathrm{card}(\mathcal{AS}_{\gamma})\le \mathrm{card}(Q)$ for all $Q\in \mathcal{Q}_{\gamma}$. If multiple subsets of $\mathcal{E}$ in \(\mathcal{Q}_{\gamma}\) have the same minimal cardinality, we select the one with the maximum total time-spent proportion \(\sum_{\{e: e\in Q\}} T_e\). 

Analytically, individuals typically spend longer periods within spatial polygons than on road segments. Therefore, if all polygon and road entities are ranked together, a level $\gamma$ PN activity space may be dominated by polygons and have few or no road segments, making it difficult to study individual movement patterns on the road network. By normalizing time-spent proportions within polygons and road segments, we obtain activity spaces that summarize key stationary locations and major movement routes on comparable scales. Substantively, this distinction aligns with the time-geographic perspective of \citet{hagerstrand-1970}, where human mobility encompasses both activity locations and the routes connecting them. Therefore, in the empirical and simulation analyzes below, we primarily use polygon-based and road-based activity spaces to separately characterize these two components of mobility. Let
\[
T_{\mathcal{A}}=\sum_{a\in\mathcal{A}}T_a,
\qquad
T_{\mathcal{L}}=\sum_{\ell\in\mathcal{L}}T_\ell
\]
be the total mean proportions of time spent in polygons and road segments, respectively. For the polygon-based activity space, given that $T_{\mathcal{A}}>0$, we rescale the polygon time-spent proportions by
\[
T^{\mathcal{A}}_a=\frac{T_a}{T_{\mathcal{A}}},
\qquad a\in\mathcal{A},
\]
so that $\sum_{a\in\mathcal{A}}T^{\mathcal{A}}_a=1$. The level $\gamma$ polygon-based activity space, denoted by $\mathcal{AS}^{\mathcal{A}}_{\gamma}$, is the smallest subset of polygons where the individual spends at least $100\gamma\%$ of their polygon time. Similarly, provided that $T_{\mathcal{L}}>0$, we define
\[
T^{\mathcal{L}}_\ell=\frac{T_\ell}{T_{\mathcal{L}}},
\qquad \ell\in\mathcal{L},
\]
so that $\sum_{\ell\in\mathcal{L}}T^{\mathcal{L}}_\ell=1$. The level $\gamma$ road-based activity space, denoted by $\mathcal{AS}^{\mathcal{L}}_{\gamma}$, is the smallest subset of road segments where the individual spends at least $100\gamma\%$ of their road-segment time.

To cover the activity space of both polygons and roads, we use the \emph{composed activity space} 
$\mathcal{AS}^{\mathcal{A}}_{\gamma}\cup\mathcal{AS}^{\mathcal{L}}_{\gamma}$.  
In both cases, ties are resolved for the PN activity space by selecting the subset with the largest total normalized time-spent proportion.
The composed activity space encompasses both polygons and high-activity roads without being dominated by polygons. 
In our analysis, we primarily use the composed activity space. This definition of the PN activity space can be made more flexible in several ways. One extension is to introduce weights $w_e$ associated with the entities $e$. The weights could be proportional to the areas of the spatial polygons and the lengths of the network segments. 
The level $\gamma$ PN activity space will include entities with a minimum total weight where individuals spend at least $100\gamma$\% of their time, i.e., $\mathcal{AS}_{\gamma}\in \mathcal{Q}_{\gamma}$ such that
\[
\sum_{e\in \mathcal{AS}_{\gamma}}w_e
\le
\sum_{e\in Q}w_e
\quad
\text{for all } Q\in \mathcal{Q}_{\gamma}.
\]

\subsection{Estimation of time spent in entities}\label{sec:estimation}

Since GPS data consist of locations at discrete timestamps, we use a weighted approach to estimate the mean proportion of time $T_e$ an individual spends in an entity $e\in \mathcal{E}$ \ -- see Eq.~\eqref{eq:avepnasdistribution}. For each day $i=1,\ldots,n$ and each $j=1,\ldots,m_i$, we assign a mark $W_{i,j}$ to the observation $(X_{i,j},t_{i,j})$ to indicate the time spent at location $X_{i,j}$. We assume this time is half the elapsed time between $(X_{i,j},t_{i,j})$ and the previous location, plus half the time between $(X_{i,j},t_{i,j})$ and the next location on the same day. Consequently, we set \(W_{i,j}=(t_{i,j+1}-t_{i,j-1})/2\) for \(1<j<m_i\). For the first and last GPS locations recorded on the $i-$th day, we set $W_{i,1} = (t_{i,1}+t_{i,2})/2$ and $W_{i,m_i} = 1-(t_{i,m_i-1}+t_{i,m_i})/2$ to ensure that the time spent at all GPS locations recorded on day $i$ encompasses the full interval of the day $[0,1]$: \(\sum\limits_{j=1}^{m_i} W_{i,j} = 1\).

We assign each GPS location $(X_{i,j},t_{i,j})$ to its nearest entity $E_{i,j}\in \mathcal{E}$: \( d\left(X_{ij},E_{ij}\right) \;\le\; d\left( X_{i,j},e\right)\) for all \(e\in \mathcal{E}\). Here, $d(x,y)$ represents the minimum geodesic distance from any point in the first geometry $x$ to any point in the second geometry $y$ \citep{pebesma-bivand-2023}.

We estimate the time an individual spends in entity $e\in \mathcal{E}$ on a given day as the sum of GPS marks recorded that day, which are closer to $e$ than to any other entity. An estimate $\widehat{T}_e$ of the mean proportion of time $T_e$ an individual spends in $e$ is the average proportion of time spent in $e$ across all days:
\begin{align}
	\widehat{T}_{e}
	\;=\;
	\frac{1}{n}\sum_{i=1}^{n}\sum_{j=1}^{m_i}
	W_{i,j}\,    
	1\!\left( E_{i,j}=e
	\right).\label{esti:Tk}
\end{align}
 We denote $E^{\mathrm{true}}_{i,j}$ as the entity with the true location $X^{\mathrm{true}}_{i,j}=S_i(t_{i,j})\in E$. The trajectory $S_i$ crosses an entity $e\in \mathcal{E}$ boundary between two consecutive GPS locations $(X_{i,j-1},t_{i,j-1})$ and $(X_{i,j},t_{i,j})$ if the individual's true location is in $e$ at $t_{i,j-1}$ or $t_{i,j}$, but not at the other time point. We collect all indices associated with boundary crossings of entity $e$ in a set
\[
\mathcal{I}_{e}
\;=\;
\bigl\{(i,j)\,:\,
\mathbf{1}\{E^{\mathrm{true}}_{i,j-1} = e\}
\;\neq\;
\mathbf{1}\{E^{\mathrm{true}}_{i,j}=e\}
\bigr\}.
\]

Let \( n_{e}=\mathrm{card}(\mathcal{I}_{e})\). We say the trajectory $S_i$ visits the entity $e$ if there exists a non-empty time interval \([\tau_{i,j_1},\tau_{i,j_2}]\subset[0,1]\) and $\xi>0$ such that 
\(S_i(\tau)= e\) holds for all \(\tau_{i,j_1}\le \tau\le \tau_{i,j_2}\), along with \(S_i(\tau_{i,j_1}-\beta)\ne e\) and \(S_i(\tau_{i,j_2}+\beta)\ne e\) for any $\beta\in[0,\xi]$. Let $V_{i,e}$ represent the number of visits to $e$ on day $i$. Denote $V_e=\sum\limits_{i=1}^n V_{i,e}$. Each visit results in two boundary crossings—an entry and an exit from $e$. Thus \(\mathrm{card}\bigl(\mathcal{I}_{e}\bigr)
	\;=\;
	2\,V_e\).
    
To study the asymptotic behavior of the estimator  in Eq.~\eqref{esti:Tk}, we assume that the individual spends a finite time in each entity and visits each a finite number of times.
\begin{assumption}\label{ass:visits}
	\begin{enumerate}[label=\textnormal{(\arabic*)}]
		\item \textbf{Detectability.} Every visit of the example individual to any entity can be recovered from the recorded GPS data. That is, at each day $i$, for any visit to the entity $e$ that occurs during the time interval $[\tau_{i,j_1},\tau_{i,j_2}]$, there is at least one GPS observation $(X_{i,j},t_{i,j})$ recorded in $[\tau_{i,j_1},\tau_{i,j_2}]$, i.e. $\tau_{i,j_1}\le t_{i,j} \le \tau_{i,j_2}$. 
		\item \textbf{Finite crossing rate.}  The total number of boundary crossings during the entire time window is bounded by average number of GPS records: \(
		\displaystyle
		\frac{\sum_{e\in \mathcal{E}}\,V_{e}}{\sum_{i=1}^n m_i}= O\left(\frac{n}{\sum_{i=1}^n m_i}\right)
		\). 
	\end{enumerate}
\end{assumption}
For each GPS location $(X_{i,j},t_{i,j})$, we define the entity misclassification probability as the likelihood that the entity closest to $X_{i,j}$ differs from the one closest to \(X^{\mathrm{true}}_{i,j}\):
\( \pi_{i,j}
\;=\;
\Pr\bigl\{
E_{ij}\neq E^{\mathrm{true}}_{i,j}
\bigr\}\).
This probability depends on GPS measurement error, the spatial configuration of the entities $\mathcal{E}$, and the proximity of the true locations \(X^{\mathrm{true}}_{i,j}\) to the entities' boundaries.

Understanding the misclassification probabilities \(\pi_{i,j}\) is important. For each entity \(e\in\mathcal{E}\), let its Voronoi cell be
\(
V_e=\{x\in\mathbb{R}^2:\ d(x,e)\le d(x,e')\ \text{for all }e'\in\mathcal{E}\}
\).
Let \(d_{i,j}\) denote the distance from the true location \(X^{\mathrm{true}}_{i,j}\) to the nearest boundary where the nearest-entity assignment can change. Equivalently, \(d_{i,j}\) is the distance from \(X^{\mathrm{true}}_{i,j}\) to the boundary of the Voronoi cell of its true nearest entity. Misclassification can only occur if the measurement error is large enough to push the observed location across this boundary. Therefore, we have 
\(
\pi_{i,j}
\;\le\;
\Pr\!\bigl(\|\varepsilon_{i,j}\|\ge d_{i,j}\bigr)
\).
Let \(r_{0.95}\) denote a 95\% horizontal accuracy radius, i.e.,
\(
\Pr\!\bigl(\|\varepsilon_{i,j}\|\le r_{0.95}\bigr)=0.95
\).
Then \(d_{i,j}\ge r_{0.95}\) implies \( \pi_{i,j}\le 0.05\).
Figure~\ref{fig:piij-bound} illustrates when misclassification events occur. The solid dot marks the true location \(X^{\mathrm{true}}_{i,j}\), the open square marks a possible observed location \(X_{i,j}\), and the vertical line represents the Voronoi boundary where the nearest-entity assignment changes. The quantity \(d_{i,j}\) is the distance from the true location to this boundary, and the dashed circle indicates a representative GPS uncertainty scale \(r_{0.95}\). When \(d_{i,j}\ge r_{0.95}\), most measurement errors stay within the same Voronoi cell, so nearest-entity assignment is unaffected by measurement error with at least \(0.95\) probability.

The GPS SPS Performance Standard \citep{gps-sps-ps-2020} reports worst-site position accuracy of \(\leq 15\,\mathrm{m}\) at the 95\% horizontal error level under typical user conditions. If the true location is more than \(15\) meters from the nearest assignment boundary, the probability of incorrectly assigning the nearest entity is at most \(5\%\) under conservative GPS accuracy. Since \(15\) meters is usually small compared to human mobility, nearest-entity misclassification rarely occurs in most GPS observations.

Under the Gaussian error model in Section~\ref{subsec:trajectory}, a sharper bound is available:
\(
\pi_{i,j}\le \exp\!\left(-\frac{d_{i,j}^2}{2\sigma^2}\right)
\).
Once \(d_{i,j}\) moderately exceeds the GPS error scale, the misclassification probability becomes negligible.
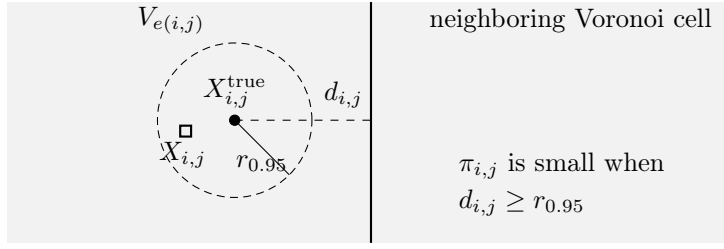
\begin{figure}[ht]
\centering
\begin{tikzpicture}[scale=1.2, font=\small]

\fill[gray!10]   (-4,-1.2) rectangle (0,1.45);
\fill[gray!10] (0,-1.2) rectangle (4,1.45);

\draw[thick] (0,-1.2) -- (0,1.45);

\node at (-2.2,1.25) {\(V_{e(i,j)}\)};
\node at (2.2,1.25) {neighboring Voronoi cell};

\draw[thick] (-2.1,-0.03) rectangle ++(0.12,0.12);
\node[below left] at (-1.7,0.02) {\(X_{i,j}\)};

\coordinate (C) at (-1.5,0.15);

\draw[densely dashed] (C) circle (0.85);

\draw[thin] (C) -- ++(-45:0.85);
\node at (-1.2,-0.32) {\(r_{0.95}\)};

\filldraw[black] (C) circle (1.6pt);
\node[above] at (-1.5,0.18) {\(X^{\mathrm{true}}_{i,j}\)};

\draw[dashed] (C) -- (0,0.15);
\node[above] at (-0.3,0.18) {\(d_{i,j}\)};

\node[align=left] at (2.1,-0.55)
{\(\pi_{i,j}\) is small when\\[2pt] \(d_{i,j}\ge r_{0.95}\)};

\end{tikzpicture}
\caption{Misclassification can occur only if the measurement error moves the observed location across the assignment boundary. Thus, if the true location is farther than \(r_{0.95}\) from the boundary, then the misclassification probability is at most \(0.05\).}
\label{fig:piij-bound}
\end{figure}

We are ready to present the main result on the convergence rate of the estimation error.
\begin{theorem}\label{thm:1}
	Under Assumption \ref{ass:visits}, for each entity $e\in\mathcal{E}$ the asymptotic absolute error \( |\widehat{T}_{e}-T_e|\) between the mean proportion of time spent by the individual in $e$ and its estimator \eqref{esti:Tk} is
	\begin{align}
		&O_p\left(\frac{1}{n}\sqrt{\sum_{i=1}^{n} \sum_{j=1}^{m_i} W_{i,j}^2 \pi_{i,j}}\right)+O\left(\frac{1}{n} \sum_{i=1}^n\sum_{j=1}^{m_i}W_{i,j} \pi_{i,j}\right)\nonumber\\
        & +O\left(\frac{1}{n}\sum_{i=1}^n \sum_{j:(i,j)\in  \mathcal{I}_{e}} t_{i,j}-t_{i,j-1} \right)+
		O_p\left(\sqrt{\frac{1}{n}}\right),\label{eq:thm1 res}
	\end{align}
	as $\max_{(i,j)}W_{i,j}\rightarrow 0$ and $n\rightarrow \infty$.
\end{theorem}
This decomposition highlights four key sources of estimation error. Component (a), GPS misclassification (\(\pi_{i,j}\)), reflects the likelihood of misassigning a GPS observation to a different entity due to displacement across a decision boundary. This probability is higher near boundaries and lower deep within an entity. In practice, it can be reduced by using higher-accuracy devices, removing low-quality observations, employing map-matching/smoothing, or coarsening entities with small buffers to avoid sharp boundaries. Component (b), irregular sampling (large gaps \(t_{i,j}-t_{i,j-1}\)), occurs from long gaps between observations, causing the piecewise-constant dwell-time approximation to over-allocate time to single observations; increasing the sampling rate (or using adaptive sampling that records more frequently during movement) mitigates this error. Component (c), frequency of boundary crossings, arises from boundary crossings between observations: the longer the trajectory spends near boundaries relative to the sampling interval, the larger the crossing-induced bias. Finer sampling during transitions, speed constraints, or map-matching can reduce it; merging very small adjacent entities can also help. Component (d), day-to-day variance (the term \(O_{p}(n^{-1/2})\)), is intrinsic day-to-day variability in behavior and decreases at the usual \(\sqrt{n}\) rate as more days are observed.

Let \(\bar\pi_W=\frac{1}{n}\sum_{i=1}^n\sum_{j=1}^{m_i}W_{i,j}\pi_{i,j}\) and \(W_{\max}=\max_{i,j}W_{i,j}\). Then
\[
\frac{1}{n}\sqrt{\sum_{i,j}W_{i,j}^2\pi_{i,j}}
\;\le\;
\sqrt{\frac{W_{\max}}{n}}\sqrt{\bar\pi_W},
\]
since \(\sum_{i,j}W_{i,j}^2\pi_{i,j}\le W_{\max}\sum_{i,j}W_{i,j}\pi_{i,j}=nW_{\max}\bar\pi_W\). In \eqref{eq:thm1 res}, the first term is a stochastic (variance) term that shrinks with \(n\) and denser sampling (smaller \(W_{\max}\)); the second is a bias term of order \(\bar\pi_W\). Bias dominates when misclassification does not decrease with more data (e.g., persistent GPS error near boundaries or dense entity definitions); the variance term vanishes while the bias remains. Variance dominates only when misclassification is extremely small and decreasing (e.g., fixes well inside entities, high-accuracy sensors, or aggressive filtering), making the bias term negligible. Under even spacing with $m$ observations per day, i.e., \(W_{\max}\asymp 1/m\), the variance term is \(O_p(\sqrt{\bar\pi_W/(nm)})\), and it can exceed the bias only if \(\bar\pi_W\ll 1/(nm)\); otherwise, the bias sets the error floor. In many real-world settings with non-negligible GPS error, the bias-dominated regime is the conservative and practically relevant case, which we apply in the even-spacing analysis below.

Next, we discuss error regimes with evenly spaced recordings. For simplicity, we assume there are $m$ observations recorded evenly each day with a finite crossing time. Thus, the boundary-crossing term satisfies
\(\frac{1}{n}\sum_{(i,j)\in\mathcal{I}_e}(t_{i,j}-t_{i,j-1})=O(1/m)\) \--- see Assumption~\ref{ass:visits}.
\begin{corollary}\label{cor:even}
Given \(m\) daily, evenly spaced observations 
\((t_{i,j}=j/(m+1))\), we have
\begin{align*}
	|\hat T_{e}-T_e|  &= 
	O_p\left(\frac{1}{n}\sqrt{\frac{1}{m}\sum_{i=1}^{n}\sum_{j=1}^{m}  \pi_{i,j}}\right)+O\left(\frac{1}{nm} \sum_{i=1}^n\sum_{j=1}^{m} \pi_{i,j}\right)\\
    & +O\left(\frac{1}{m} \right)+
	O_p\left(\sqrt{\frac{1}{n}}\right),
\end{align*}
	as $\max_{(i,j)}W_{i,j}\rightarrow 0, n\rightarrow \infty$.
\end{corollary}

Let \(S_\pi=\sum_{i=1}^n\sum_{j=1}^m \pi_{i,j}\). The leading terms in Corollary~\ref{cor:even} are\\
\(O_p\!\bigl(\frac{1}{n}\sqrt{S_\pi/m}\bigr)\), \(O(S_\pi/(nm))\), \(O(1/m)\), and \(O_p(n^{-1/2})\). The overall rate is the maximum of the misclassification bias \(S_\pi/(nm)\), the sampling/crossing term \(1/m\), and the day-to-day variance \(n^{-1/2}\). The stochastic misclassification term is smaller than its bias counterpart when \(S_\pi\) does not vanish too quickly. Thus, the overall error rate is: (a) bias-dominated: if \(S_\pi/(nm)\) is the largest term (e.g., persistent GPS error or many observations near boundaries), error is influenced by misclassification and does not significantly decrease by increasing \(n\) or \(m\) alone; improving sensor accuracy or coarsening entities is necessary for a meaningful reduction; (b) sampling/crossing-dominated: if \(1/m\) is largest, error is driven by temporal discretization and boundary crossings; increasing the sampling rate or using adaptive sampling reduces error; (c) day-to-day variance-dominated: if \(n^{-1/2}\) is largest, the overall rate can be reduced by collecting GPS observations over more days. When measurement error is negligible (i.e., \(S_\pi/(nm)\) is small), the dominant terms simplify to \(O(1/m)\) and \(O_p(n^{-1/2})\).

If GPS error is negligible relative to entity size and crossings rarely occur near boundaries, the dominant bias and variance reduce to the recording time interval during crossings and day-to-day variance:
	\begin{align*}
		\max_k|\hat T_{e}-T_e|  = & 	O\Bigl(
		\tfrac{1}{n}\sum_{(i,j)\in\mathcal{I}_{e}}
		(t_{i,j}-t_{i,j-1})
		\Bigr)+O_p\left(\sqrt{\frac{1}{n}}\right).
	\end{align*}
Under the even-spacing setup, this simplifies to \(O(1/m)+O_p(n^{-1/2})\).
Finer sampling and more recording days improve accuracy.

\section{The map-augmented simple movement model} \label{sec:SMM-map}

This section introduces a model that captures the spatial distribution of an individual's daily activities using GIS contextual information. Following the space-time geography framework of \citet{hagerstrand-1970}, we decompose a daily space–time path into an alternating sequence of stations and routes. Stations represent conceptual pockets of local order—bounded areas where individuals engage in activities such as sleeping, working, or strolling. The boundaries of stations do not have a specific geometric shape. Routes link segments connecting consecutive stations while the individual is in motion. A station can range from a room to a city park. Within its boundaries, fine‐scale motion does not change the individual’s activity state.  Representing stations as polygons and routes as road segments embeds the conceptual units of \citet{hagerstrand-1970} into a GIS layer without altering their behavioral meaning. The GIS map layers allow us to assign operational geometry to the behavioral units detailed in Table \ref{tab:smm}.

\begin{table}
    \centering
	\begin{tabular}{@{}lll@{}}
		\toprule
		Time–geography term & Operational representation & SMM symbol \\ \midrule
		Station (conceptual) & Chosen spatial polygon & $a\in\mathcal{A}$ \\
		Route (conceptual)   & Ordered road segments     & $l\in\mathcal{L}$ \\ \bottomrule
	\end{tabular}
    \caption{Key components of the map-augmented simple movement model.}
    \label{tab:smm}
\end{table}

The GIS information, containing spatial polygons and road segments, serves as a proxy for translating 
H\"{a}gerstrand’s abstraction into a computable model without altering its meaning. We present the map–augmented simple movement
	model (map-SMM) that generates trajectories consistent with the latent
paths \(S(t)\) introduced in Section~\ref{subsec:trajectory}. In our previous work \citep{wu2025statisticalframeworkanalyzingactivity}, we proposed the simple movement model (SMM) that generates daily trajectories by alternating between anchor locations and route segments, capturing regular human activity patterns. Building on the SMM structure, the map-SMM embeds stations and routes into the polygon–network layer, operationalizing H\"{a}gerstrand’s space–time geography.

Recall that \(\mathcal{A}=\{a_{1},\dots ,a_{n_A}\}\) and \(\mathcal{L}=\{l_{1},\dots ,l_{n_L}\}\) denote the set of spatial polygons and road segments. In the map-SMM,
a one–day activity pattern is represented by an action vector
\(
B \;=\;
\bigl(
a_{1},
l_{1,1},\dots ,l_{1,n_{1}},
a_{2},\dots ,
a_{k}
\bigr).
\)
For $j=1,\ldots,k-1$, the ordered sequence
\(\ell_{j,1},\dots ,\ell_{j,n_{j}}\in\mathcal{L}\) represents the route involving road segments
 from polygon \(a_{j}\) to polygon \(a_{j+1}\).
Daily variability is represented by treating \(B\) as a random vector generated from a collection of vectors. Given the action vector \(B\), the map-SMM  includes dwelling and travel time vector 
\(
Z
=\bigl(
Z_{1},\;
Z_{1,1},\dots ,Z_{1,n_{1}},\;
Z_{2},\dots ,Z_{k}
\bigr),
\)
that records the proportion of time spent in each polygon and on each visited road segment. The sum of all elements in $Z$ is always 1. The pair \((B,Z)\) encapsulates the trajectory \(S(t)\) at the
entity level: it specifies the polygon or road segment visited and the duration of the visit for each \(t\in[0,1]\)
, capturing the essential structure of the daily activity pattern, even though the precise intra-entity spatial locations remain unspecified.

\paragraph{Example.}
The top-left panel in Figure~\ref{fig:map_pattern} displays a synthetic map from a simulation in the Supplementary Material. This map contains six polygons denoted by \(\mathrm{P}1\)--\(\mathrm{P}6\) and twelve road segments denoted by \(\mathrm{S}1\)--\(\mathrm{S}12\). This notation simplifies specifying movement patterns like \(\mathrm{S}3\!-\!\mathrm{S}10\!-\!\mathrm{S}11\), which denotes a path home (\(\mathrm{P}1\)) to the supermarket (\(\mathrm{P}4\)).

\begin{figure}[H]
	\centering
	\includegraphics[width=4.5in]{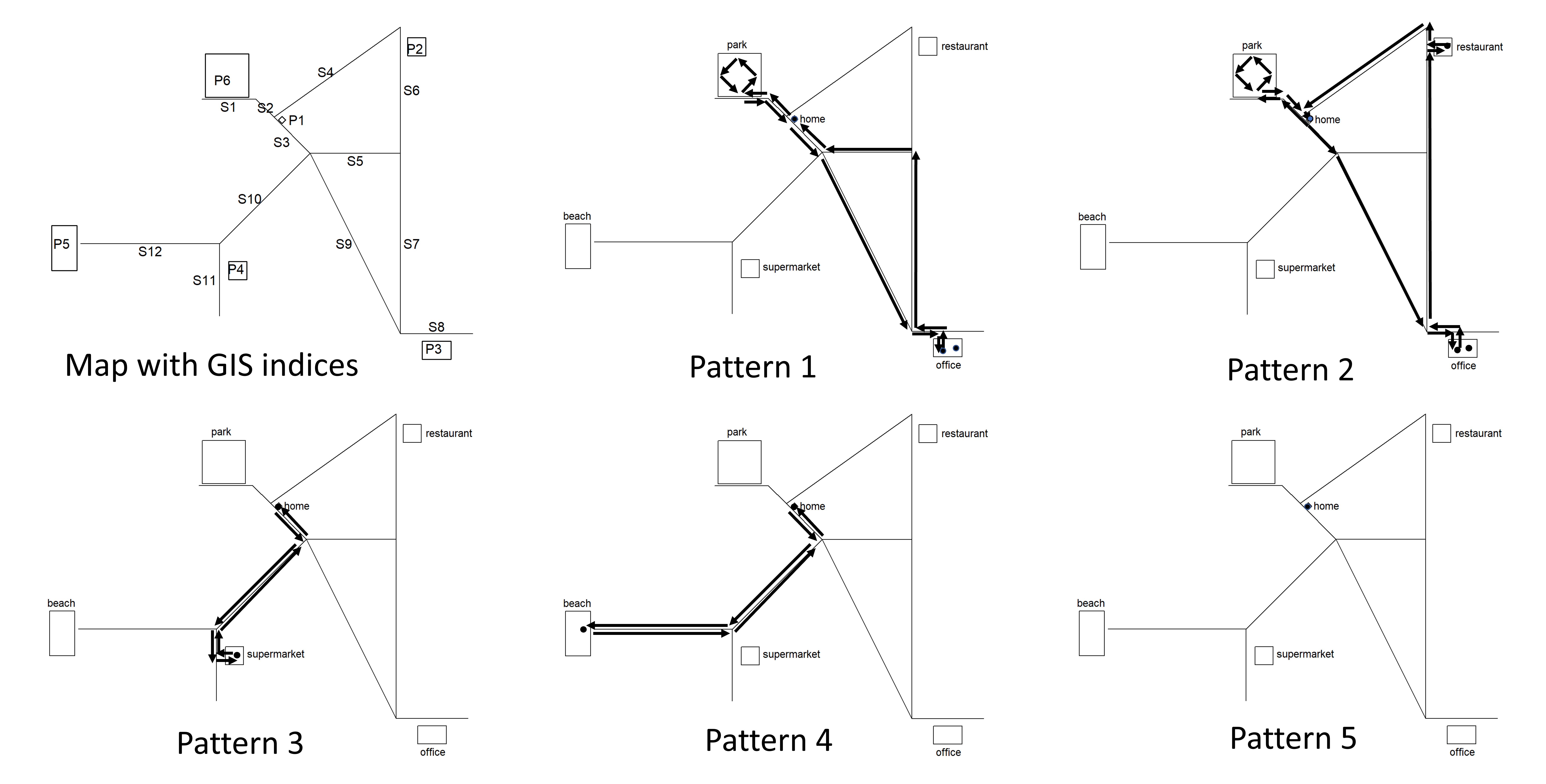}
	\caption{Synthetic map (top left panel) and five possible daily action
		vectors (other panels).  Full details appear in Supplementary Material.}
	\label{fig:map_pattern}
\end{figure}

The remaining panels in Figure~\ref{fig:map_pattern} illustrate five distinct daily movement patterns of a fictitious individual. For example, pattern~1 is represented by
\[
B_{1}=
\{\,
\mathrm{P}1,\;
\mathrm{S}3,\mathrm{S}9,\mathrm{S}8,\;
\mathrm{P}3,\;
\mathrm{S}8,\mathrm{S}7,\mathrm{S}5,\mathrm{S}3,\;
\mathrm{P}1,\;
\mathrm{S}3,\mathrm{S}2,\mathrm{S}1,\;
\mathrm{P}6,\;
\mathrm{S}1,\mathrm{S}2,\mathrm{S}3,\;
\mathrm{P}1
\}.
\]
Here $\mathrm{P1} = \textit{home}$, $\mathrm{P3} = \textit{office}$ and $\mathrm{P6} = \textit{park}$. The sequences
\(\mathrm{S3}\!-\!\mathrm{S9}\!-\!\mathrm{S8}\) and
\(\mathrm{S}8\!-\!\mathrm{S7}\!-\!\mathrm{S5}\!-\!\mathrm{S3}\) depict the individual's daily commute, while \(\mathrm{S3}\!-\!\mathrm{S2}\!-\!\mathrm{S1}\) describes an evening walk  from home to the park. 

\paragraph{Intra-polygon motion.}
The original time–geographic term station referred to 
being at a bounded locale. By representing a station as a polygon, we extend the notion of staying to include dwelling within it, allowing local movement without altering the activity state.
In Pattern~1, the \texttt{office} polygon \(\mathrm{P3}\)
might include both the worker’s desk area and a nearby commercial space; the individual can circulate freely between these sub-locations while the action vector records only the total office dwell time. Similarly, within the \texttt{park} polygon \(\mathrm{P6}\), the
individual may move continuously rather than remaining still; however, the model counts the entire interval as a single stay episode. For simplicity, the action vector encoding the daily activity pattern does not differentiate motion within a polygon (e.g., strolling  for 30 minutes vs. standing in the park for 30 minutes); only the total dwell time is considered and contributes to \(Z\).

Given a realization \((B,Z)\), the implicit time sub-interval from \(Z\) partitions \([0,1]\); the corresponding elements of \(B\) determine whether the individual is in a polygon (station) or traveling along segments (route) during each sub-interval. Hence, the map–based SMM encapsulates the daily trajectory \(S(t)\) as an ordered sequence of entities with dwell times. It preserves H\"{a}gerstrand’s station–route alternation while remaining clear and storage-efficient for describing daily activities. It also enables the analysis and simulation of daily human activity patterns.

\subsection{Clustering spatiotemporal trajectories}\label{subsec:clustering-trajectory}

Human mobility data exhibit significant day-to-day variability influenced by work schedules, social activities, and irregular events, even for the same individual. Mobility analysis aims to extract meaningful behavioral patterns from this variability by identifying representative daily patterns that summarize recurrent activity configurations. Clustering daily mobility patterns reduces high-frequency variation while preserving interpretable behavioral modes, enabling analyzes of routine stability, behavioral change, and exposure heterogeneity. In particular, day-level clustering allows researchers to distinguish structurally different activity patterns—such as workdays, weekends, or home-bound days—without imposing strong parametric assumptions on temporal ordering or spatial usage.

Let $\{(X_{d,j},t_{d,j})\}_{j=1}^{m_{d}}$ represent the GPS observations recorded on day $d$. Each observation is matched to its nearest entity, denoted by $\widehat{e}_{d,j}\in\mathcal{E}$ as in Section~\ref{sec:estimation}, and assigned weight $w_{d,j}= \tfrac{1}{2}(t_{d,j+1}-t_{d,j-1})$, with time adjustments for $j=1$ and $j=m_{d}$. We aggregate consecutive observations sharing the same matched entity by following these steps until no two adjacent entities are identical: (i) scan the sequence $(\widehat{e}_{d,1},\dots,\widehat{e}_{d,m_{d}})$ from left to right; and (ii) when $\widehat{e}_{d,j}=\widehat{e}_{d,j-1}$ occurs, merge the two records and add their weights. The result is the estimated action vector $\hat B_d$, along with the estimated proportion of aggregated time spent $\widehat{\mathbf{Z}}_{d}$ corresponding to each entity in the action vector
\[
\widehat{\mathbf{B}}_{d}
=\bigl(
e^{(d)}_{1},\dots ,e^{(d)}_{r_{d}}
\bigr),
\qquad
\widehat{\mathbf{Z}}_{d}
=\bigl(
z^{(d)}_{1},\dots ,z^{(d)}_{r_{d}}
\bigr).
\]
Here, $e^{(d)}_{j}$ is the entity label of the $j$-th block, $z^{(d)}_{j}$ is the sum of associated weights, and $\sum_{j=1}\limits^{r_{d}} z^{(d)}_{j}=1$. After determining the pair
$\bigl(\widehat{\mathbf{B}}_{d},\widehat{\mathbf{Z}}_{d}\bigr)$
 For each day $d\in\{1,\dots ,n\}$, the next task is to compare daily sequences to cluster similar activity patterns.

\subsubsection{Clustering daily activity}\label{subsec:clustering}

The Levenshtein distance \citep{Levenshtein1966} between two strings
is the minimum number of operations, such as insertions, deletions, and substitutions of a single character, required to transform one string into the other. By treating the compressed action vector $\widehat{\mathbf{B}}_{d}$ as a word with entity label letters
, we create a natural analogy: the edit operations represent changes in activity (specifically, entity visited) needed to transform one day’s pattern into another—inserting an absent activity, deleting one that did not occur, or substituting one activity for another at the same point in the sequence.  Each edit, therefore, represents a concrete behavioral adjustment made by the individual.

Sequence-based similarity measures compare mobility patterns as ordered activity sequences. For example, \citet{Yang2025Sequence} proposes a time-informed longest common subsequence (T-LCS) framework that incorporates temporal gaps and timestamps to align daily mobility sequences and derive similarity scores for clustering and pattern mining. Other research improves edit distance–based methods by incorporating spatial and temporal penalties, such as the augmented space–time–weighted edit distance of \citet{Zhai2019STWED}, where costs increase with spatial separation and temporal misalignment between activities.

Our approach differs from these methods in both construction and interpretation. Instead of aligning activities with clock time or penalizing spatio-temporal mismatches, we directly weight edit operations by dwell time. This adjustment ensures the resulting distance reflects the minimum time needed to transform the activity pattern on day~$d_{1}$ into that on day~$d_{2}$. This results in a dissimilarity metric with a clear behavioral interpretation in time units, suitable for clustering daily activity patterns within individuals, where variations in time allocation across activities are the main source of difference.

To implement this concept, we quantify how different the daily activity patterns are on two days $d_{1}$ and $d_{2}$ by calculating the edit distance between their respective compressed action vectors, $\widehat{\mathbf{B}}_{d_{1}}$ and $\widehat{\mathbf{B}}_{d_{2}}$. The resulting distance is symmetric and reflects the minimal behavioral changes needed to transform the sequence of visited entities on day $d_{1}$ into that on day $d_{2}$. We observe that not all changes in activities are equally important; for instance, substituting a brief stop is qualitatively distinct from redistributing several hours across multiple locations. To account for this heterogeneity, we define a dwell-time-weighted variant of the Levenshtein distance, in which the cost of each edit operation is proportional to the duration of the activity it modifies.

To reflect the importance of each modification, we assign a weight to every edit according to the dwell time of the corresponding entity: (i) inserting an entity $e$ from day~$d_{2}$ into the sequence of day~$d_{1}$ incurs a cost of $z_{e}^{(d_{2})}$, corresponding to its dwell time on $d_{2}$; (ii) deleting an entity $e$ from day~$d_{1}$ incurs a cost of $z_{e}^{(d_{1})}$; (iii) substituting $e\mapsto e'$ incurs a cost of $z_{e}^{(d_{1})}+z_{e'}^{(d_{2})}$, since the total dwell time needs to be reassigned. The time–weighted edit distance $d_{\text{TW}}(d_{1},d_{2})$ is the minimum total cost to convert $\widehat{B}_{d_{1}}$ into $\widehat{B}_{d_{2}}$. The metric helps answer the question: ``How much time must the individual reallocate to convert the pattern on day~$d_{1}$ into the activity pattern on day~$d_{2}$?'' This indicates that individuals' trajectories can be clustered into groups based on the action vector $B$. 

We construct the $n\times n$ distance matrix \(D_{ab}=d_{\text{TW}}(a,b)\) using the weighted Levenshtein metric. Since $D$ is a proper metric, any classical clustering method can be applied; we use agglomerative hierarchical clustering with single linkage, providing
 both a flat partition and a dendrogram for 
analysis. Days with unusually large average distances to others are
flagged as behavioral outliers.

\subsection{Temporal stability of human mobility patterns}

A central challenge in mobility research is to determine how many days of GPS tracking are necessary to reliably represent an individual’s spatial behavior. As emphasized by \citet{Adrianpaper}, identifying the minimal observation window is crucial for limiting participant burden and reducing both the monetary and privacy costs associated with extended data collection. In our framework, activity spaces are derived from GIS entities—polygons corresponding to fixed locations such as home, workplace, and discretionary venues, together with road segments that encode travel pathways. By conducting a stability analysis, we can evaluate how quickly the cumulative distribution of time spent across these entities converges to its full-period distribution and, in turn, infer how many days of data are required before the main features of a person’s mobility pattern (e.g., habitual destinations and frequently traveled routes) become evident.

To quantify the stability of an individual’s mobility pattern, \citet{Adrianpaper} divide the study region into a fine square grid \(\mathcal{G}=\{g_{1},\dots ,g_{N}\}\). For each day \(d\), they compute the proportion of time the individual spends in each grid cell, denoted \(\mathbf{p}^{(d)}=(p^{(d)}_{1},\dots ,p^{(d)}_{N})\). They then form the mean daily time-allocation vector over the first \(D\) days,
\[
\bar{\mathbf{p}}(D)=\tfrac{1}{D}\sum_{d=1}^{D}\mathbf{p}^{(d)}.
\]
For each \(D\in\{1,\dots,n\}\) and a fixed \(c\in[0,1]\), they sort the entries of \(\bar{\mathbf{p}}(D)\) and identify the smallest set of cells whose cumulative probability is at least \(c\). This yields the \emph{$c$–core} \(S_{c}(D)\), serving as the grid-based analog of an activity space at level \(c\).  The last–crossing time (LCT) is defined as
\[
\widehat{\mathrm{LCT}}_{\text{grid},c}(\gamma)
=\max_{D=1,\dots,n}
\Bigl\{
D:\,
\tfrac{|S_{c}(D)\triangle S_{c}(n)|}
{|S_{c}(n)|}
>\gamma
\Bigr\},
\]
that is, the latest day on which the set of high-use cells still differs from its final configuration by more than the tolerance \(\gamma\). Here, \(\triangle\) denotes the symmetric difference and \(|\cdot|\) the cardinality of a set. This measure characterizes how rapidly the core spatial region stabilizes as additional days of observation are incorporated.

In our object-oriented framework, grid cells are replaced by GIS entities
\(\mathcal{E}=\mathcal{A}\cup\mathcal{L}\), where \(\mathcal{A}\) denotes polygons and \(\mathcal{L}\) denotes road segments. For the first \(D\) days, we construct cumulative $c$–level composed activity spaces
\[
\underbrace{\mathcal{AS}^{\mathcal{A}}_{c}(D)}_{\text{polygon part}}
\;\cup\;
\underbrace{\mathcal{AS}^{\mathcal{L}}_{c}(D),}_{\text{road--network part}}
\]
where \(\mathcal{AS}^{\mathcal{A}}_{c}(D)\) is \(100c\%\) polygon--based activity space observed up to day \(D\),
and \(\mathcal{AS}^{\mathcal{L}}_{c}(D)\) is the  \(100c\%\) of road-based activity. These sets are an individual’s core entity collections in the built environment. Building on this concept, we evaluate stability using the GIS entities introduced in our framework: polygons representing stationary activity locations and road segments representing travel routes. Because these two layers describe different behavioral processes—stationary versus mobile activity—we assess their stability separately. This separation also maintains consistency with our definition of activity spaces, estimated independently for both polygon and network components.

We define two LCT measures to monitor the temporal stabilization of polygon and road-based activity spaces:
\begin{eqnarray}
    \widehat{\mathrm{LCT}}^{\mathcal{A}}_{c}(\xi)
	& = & \max_{D=1,\dots ,n}
	\Bigl\{
	D:
	\frac{
		\lvert
		\mathcal{AS}_{c}^{\mathcal{A}}(D)
		\triangle
		\mathcal{AS}_{c}^{\mathcal{A}}(n)
		\rvert}{
		\lvert
		\mathcal{AS}_{c}^{\mathcal{A}}(n)
		\rvert}
	>\xi
	\Bigr\},\\ \nonumber
    \widehat{\mathrm{LCT}}^{\mathcal{L}}_{c}(\xi)
	& = &\max_{D=1,\dots ,n}
	\Bigl\{
	D:
	\frac{
		\lvert
		\mathcal{AS}_{c}^{\mathcal{L}}(D)
		\triangle
		\mathcal{AS}_{c}^{\mathcal{L}}(n)
		\rvert}{
		\lvert
		\mathcal{AS}_{c}^{\mathcal{L}}(n)
		\rvert}
	>\xi
	\Bigr\}.
	\label{eq:LCT-level}
\end{eqnarray} 

Here, \(\widehat{\mathrm{LCT}}^{\mathcal{A}}_{c}(\xi)\) and \(\widehat{\mathrm{LCT}}^{\mathcal{L}}_{c}(\xi)\) denote the final days when the polygon-based and road-based activity spaces, respectively, still differ from their full-period versions by more than a fraction \(\xi\). Smaller values reflect quicker stabilization of frequently visited polygons or segments, while larger values indicate ongoing spatial evolution, such as the gradual integration of new buildings or travel routes into an individual’s routine. This two-tier formulation maintains LCT interpretability while accommodating varying spatial resolutions of built-environment features. It generalizes the grid-cell approach of \citet{Adrianpaper} to a spatial object-oriented setting, assessing stability separately for stationary and network-based activity spaces.

\subsubsection{Understanding LCT values concerning daily activity patterns}\label{subsec:LCT and pattern}
The polygon- and road-based LCT values can be understood by tracking how the cumulative distribution of time spent across the polygon or road segment sets approaches the full-period distribution as more days are included. Since each polygon or segment contributes a fixed share of daily time use, the activity space at level~$c$ is defined as the smallest collection of entities whose cumulative probability reaches~$c$. Entities with larger time shares are included in the activity space at lower coverage levels, while those with smaller positive time allocations are included only at higher coverage levels. Stability is evaluated by comparing, for each $D$, the interim activity space $\mathcal{AS}^{\mathcal{A}}_{c}(D)$ and $\mathcal{AS}^{\mathcal{L}}_{c}(D)$ to their corresponding full-period activity space $\mathcal{AS}^{\mathcal{A}}_{c}(n)$ and $\mathcal{AS}^{\mathcal{L}}_{c}(n)$.

The notion of PN space and associated mobility patterns helps explain the mechanisms that produce large LCT values. Such high values occur when a person’s day-to-day movement is governed by distinct behavioral regimes that shift over time. For example, an individual may follow a temporary routine for the first several weeks (such as during a vacation or an unusual work schedule) and only later settle into a stable home–work routine once regular employment begins. In this situation, the cumulative distribution based on the first $D$ days is heavily shaped by the initial, short-lived pattern and thus differs substantially from the distribution over the full observation window, which is mainly driven by the later, more persistent regime. Only after enough days under this new, stable pattern have been included does the cumulative distribution begin to converge toward the long-term distribution, which in turn yields large LCT values.

The LCT curves are not necessarily monotonic in their coverage. This phenomenon can also be observed in our simulation experiments; see Section \ref{subsec: real-data:LCT} and the Supplementary Material for examples. To explain this, we illustrate the mechanism using the polygon-based activity space.
The LCT can achieve large values at coverage levels where a new polygon first enters the activity space for target level $c$, which is close to the cumulative probability of the marginal polygon. For such $c$, even small discrepancies between the cumulative distributions computed from the first $D$ days and from all $n$ days can influence whether this polygon is included in $\mathcal{AS}^{\mathcal{A}}_{c}(D)$. As a result, polygon collections for mid-period days may differ from the full-period collection because the threshold passes through a sensitive area of the cumulative distribution. This produces LCT curves that may exhibit a high value at one coverage level, a lower value at the next, and then a renewed increase once a different polygon becomes marginal, even if the underlying behavior from day to day is relatively stable.

When daily mobility alternates among stable routines, time spent across polygons is consistent each day. The cumulative distribution then converges rapidly to its long-run shape, producing low LCT values. This interpretation aligns with \citet{schneider2013unravelling}, indicating that many people cycle through recurring mobility motifs, which creates regularity in aggregate spatial behavior.

The LCT characterizes how quickly an individual's core spatial behavior can be inferred from daily observations. High LCT values occur either when there are genuine changes in mobility regimes or when the measure is sensitive to thresholds as new polygons are incorporated into the activity space. Low LCT values, by contrast, signal stable, repetitive mobility patterns. Assessing this stability clarifies when an activity space can be treated as statistically reliable and how regularity or shifts in behavior manifest over time in its evolution.

\section{Analysis of GPS data}\label{sec:real-data}

We apply our methodology to the GPS data in Section \ref{sec:data}.

\subsection{Estimation of time spent proportion and activity space} 
\begin{figure}[htbp]
	\centering
	\includegraphics[width=0.48\linewidth]{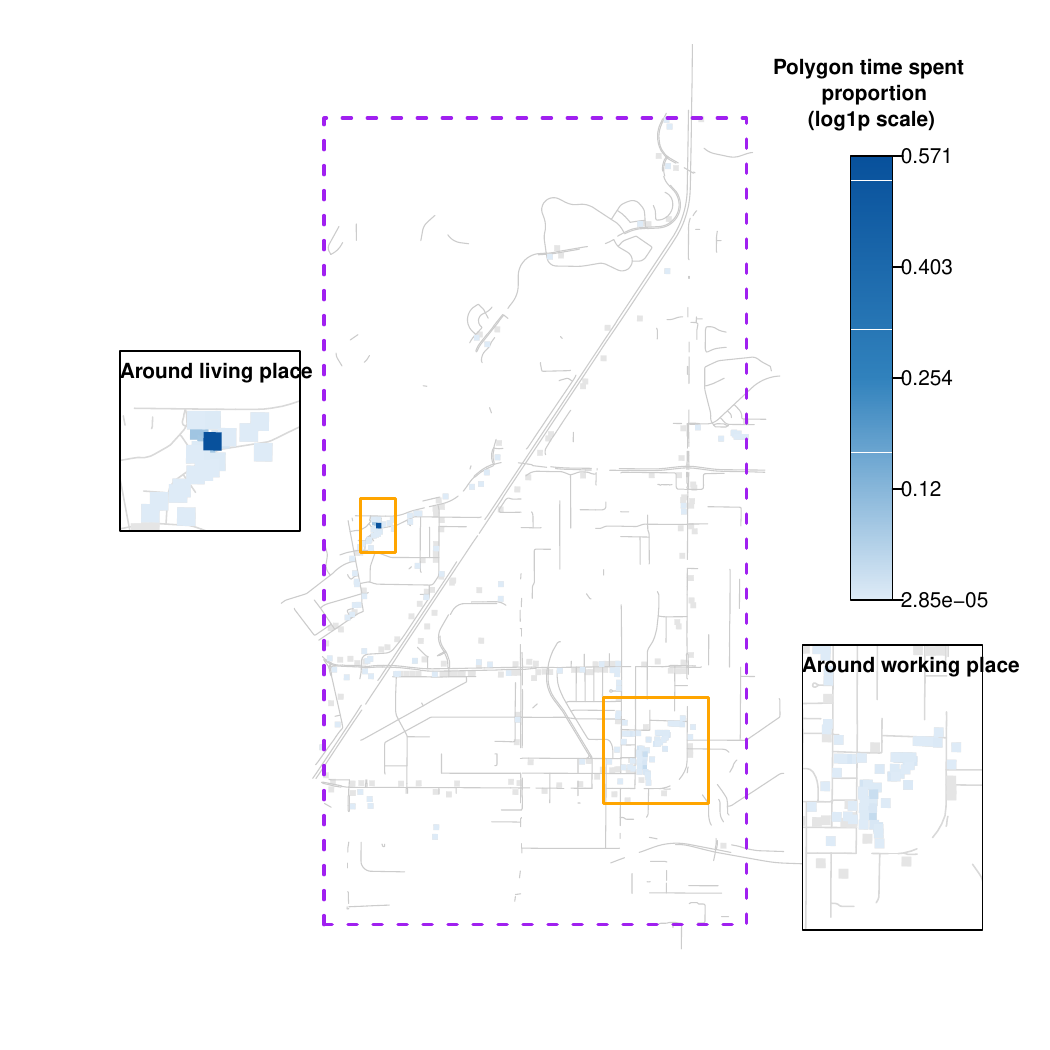}
	\includegraphics[width=0.48\linewidth]{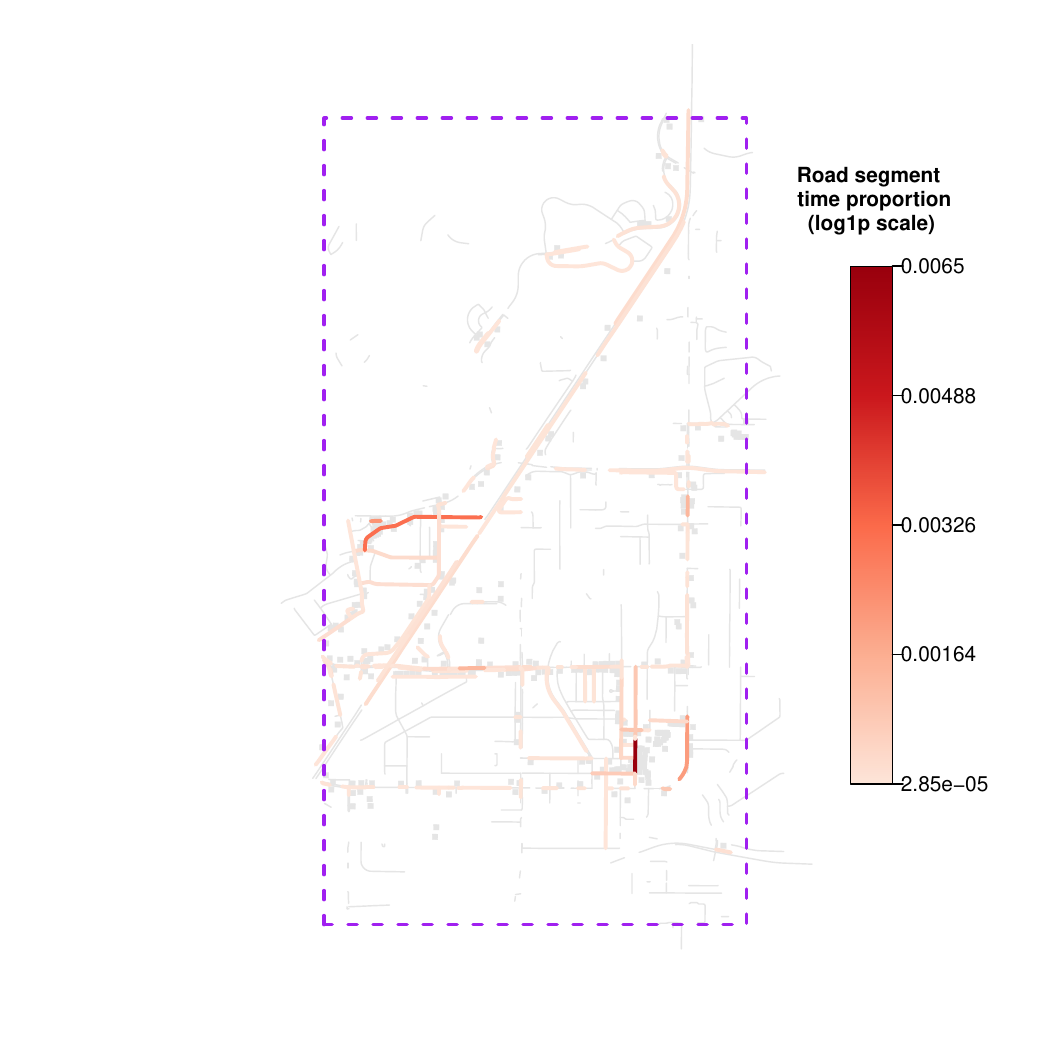}
	\includegraphics[width=0.32\linewidth]{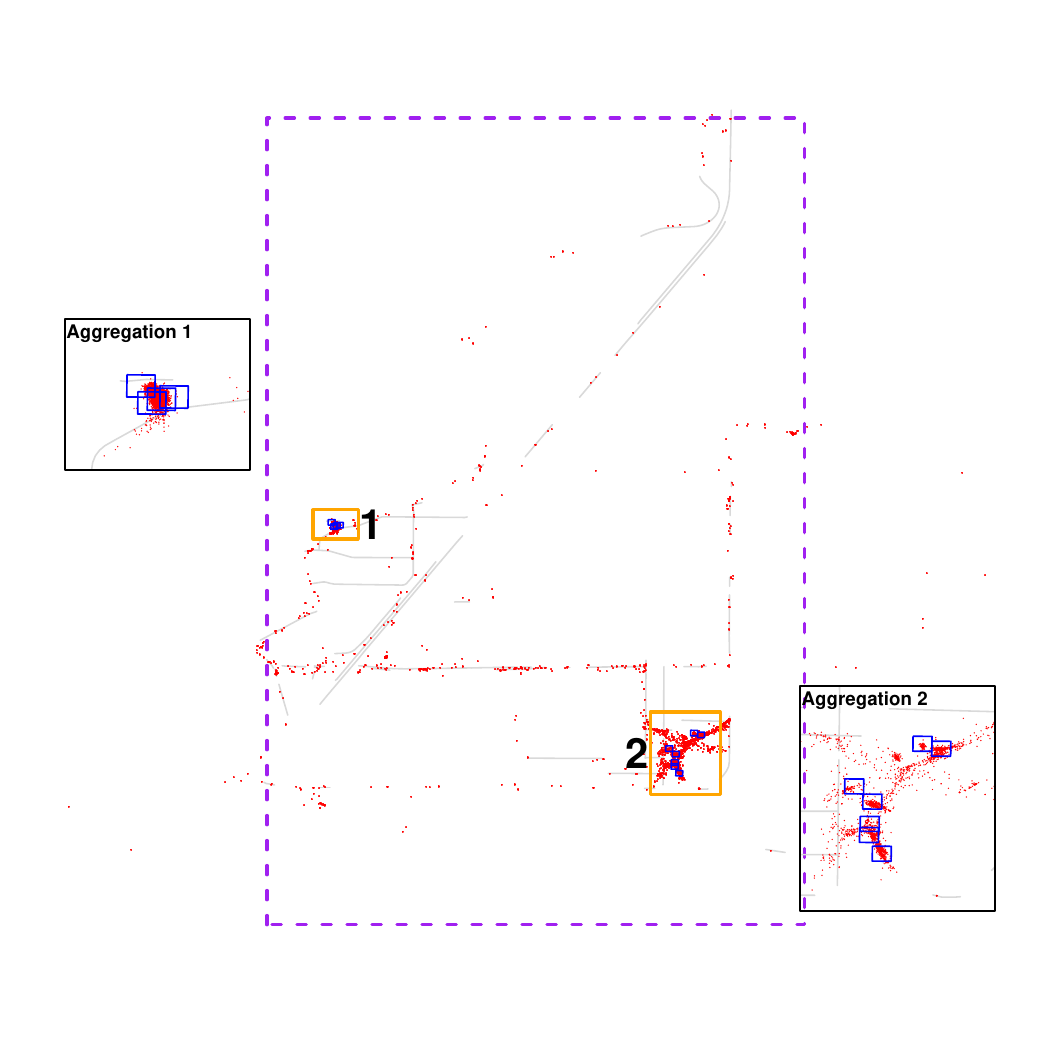}
	\includegraphics[width=0.32\linewidth]{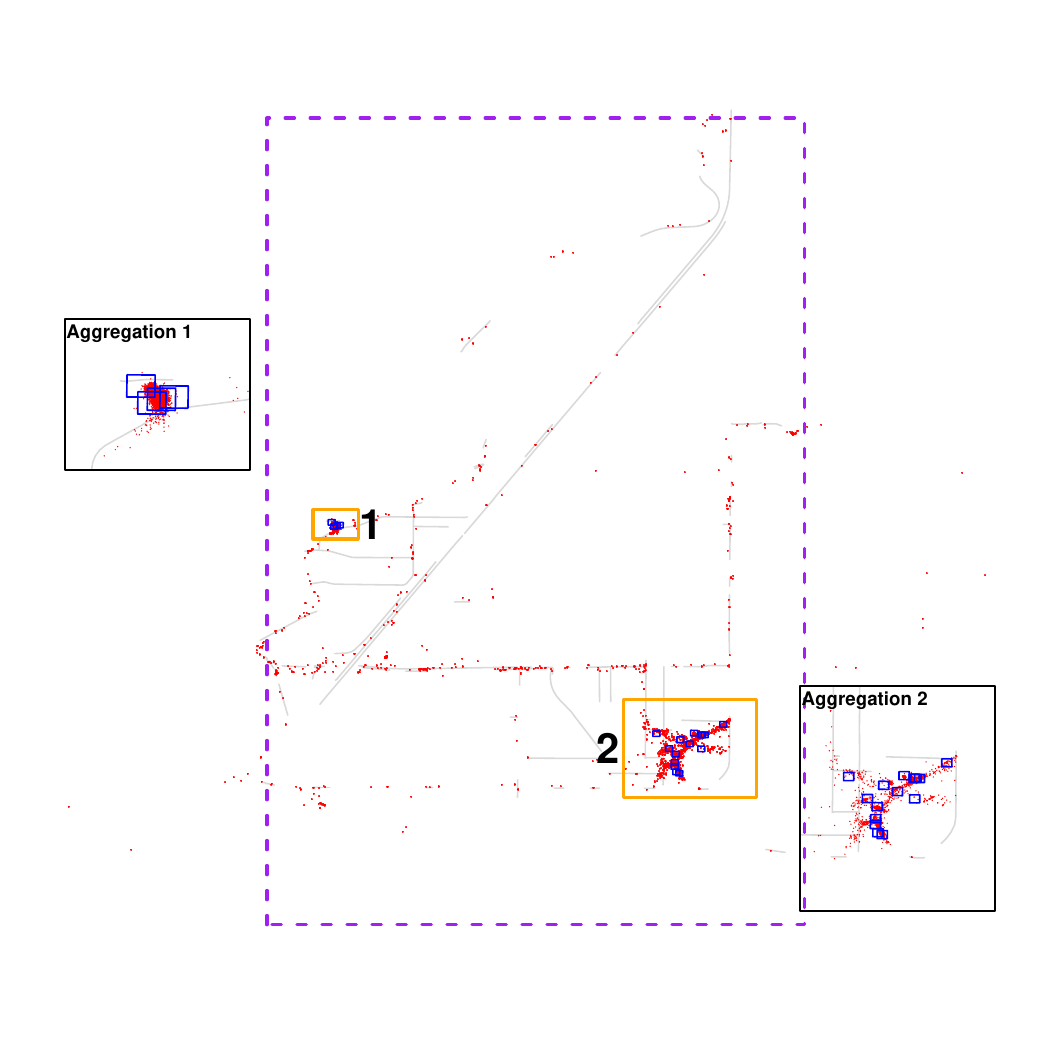}
	\includegraphics[width=0.32\linewidth]{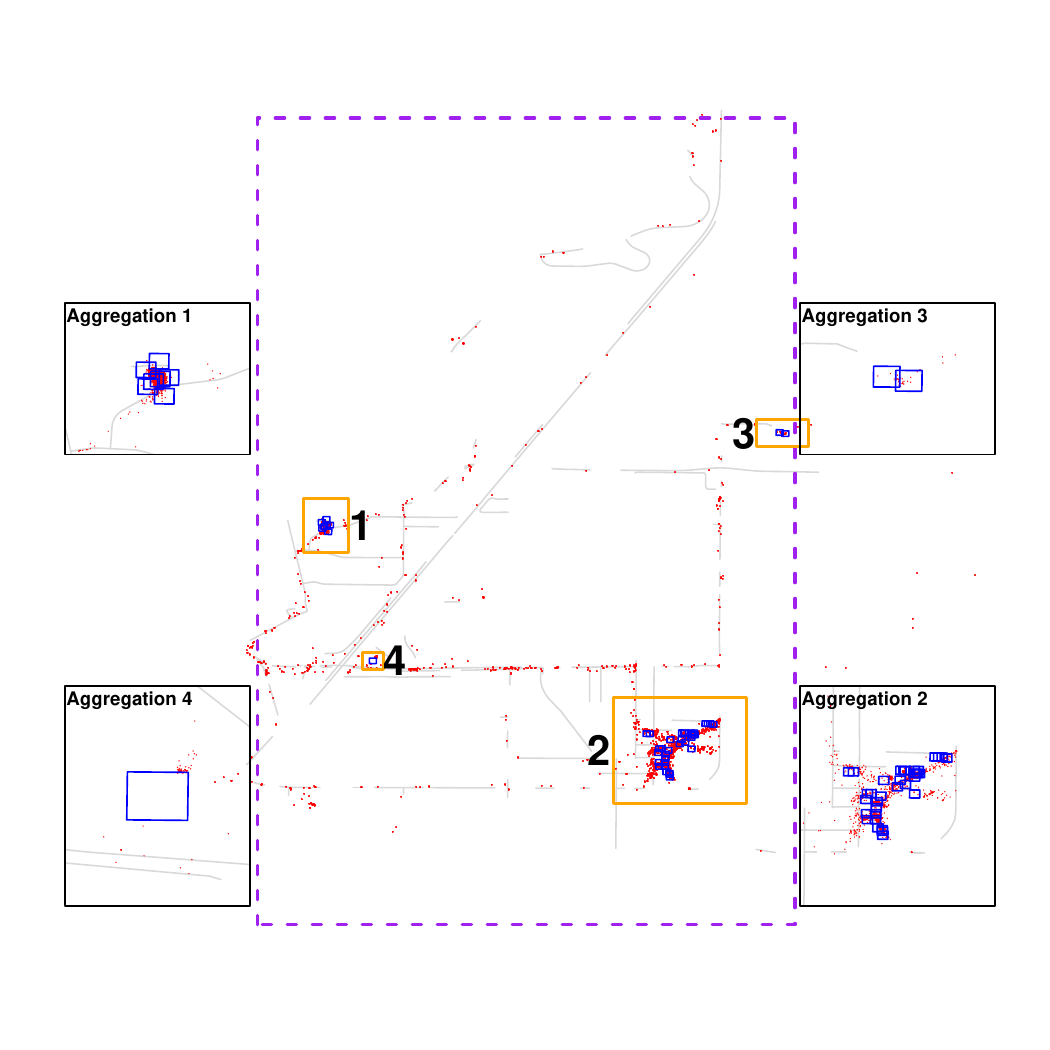}		
	\caption{First row: Spatial distribution of estimated time-spent proportions across
		 estimated time-spent proportion in building polygons (left panel) and on road segments (right panel). Darker blue indicates higher $\log(1+p_e)$; grey denotes zero. Second row: Estimated composed activity spaces at three coverage levels ($90\%$, $95\%$, and $99\%$). Blue polygons correspond to building-based spaces (stationary exposure) with insets highlighting aggregations of buildings; red dots represent GPS observations. Aggregation~1 is identified as the living place, aggregation~2 as the main daytime hub, aggregation~3 (at the 99\% level) as a functionally similar site to aggregation~2, and aggregation~4 (at the 99\% level) as a commercial space. Grey lines denote the road network, with highlighted red segments indicating road-based activity spaces (transitional exposure).}
	\label{fig:time_spent_maps}
\end{figure}

We begin with a pre-processing stage where raw GPS measurements and GIS data are combined to derive a compact set of spatial entities. We define an individual-specific study area based on GPS traces, extract the relevant road network and building polygons, and apply distance-based filtering and aggregation to reduce spatial fragmentation in dense environments. This procedure ensures that the final spatial entities represent meaningful activity locations and are relatively insensitive to GPS measurement error. Further information on GIS selection, bounding-box definition, road-network filtering, and polygon aggregation is provided in the Supplementary Material.

After pre-processing, we compute the time spent within each spatial entity. 
Let $\mathcal{A}$ denote the final collection of building-based activity polygons and $\mathcal{E}$ the chosen set of road segments. 
For every entity $e \in \mathcal{A} \cup \mathcal{E}$, we determine its proportion of time spent using Eq.~\eqref{esti:Tk}, which transforms the GPS observations located in that entity into an estimate of the share of time that the individual devotes to that place. These estimated time shares form the basis for the later construction of activity space analyses, daily activity profiles, and subsequent clustering and stability analyses. For visualization, we present the time spent proportion estimates for polygons and road segments in separate subpanels. A $\log(1+x)$ transformation is applied to mitigate heavy-tailed values and enhance the visibility of low-intensity areas; locations with zero time are depicted in gray. The first row of Figure~\ref{fig:time_spent_maps} displays the estimated time spent proportions for polygons and road segments.

The left panel in the top row of Figure~\ref{fig:time_spent_maps} reveals that stationary exposure is highly concentrated: a small number of polygons accounts for the majority of time spent, consistent with recurrent anchor locations (e.g., home or central hubs). The inset, showing polygons to the left of the study area, highlights an aggregation of polygons around the residence. A second aggregation appears in the bottom-right, corresponding to primary locations where the individual spends daytime. The right panel in the top row of Figure~\ref{fig:time_spent_maps} shows that transitional exposure spreads across the road network, complementing the polygon-based view of mobility by capturing where people dwell (polygons) and how they travel between locations (roads). Multiple potential routes are visible between home and other destinations, some of which are occasionally utilized. When traveling, the main roads are used more frequently than local road segments farther from home and daytime activity polygons. Moreover, the distribution of time spent on road segments is heavy-tailed: a small subset of segments near home and main daytime activity locations accounts for a disproportionately large fraction of total travel time. In addition to reflecting lower speeds near these locations (on local roads and during departures or arrivals), this pattern indicates frequent use of these segments for commuting and social activities.

After estimating the time proportions spent in each polygon and road segment, we construct polygon- and road-based activity spaces for the chosen coverage level $\alpha \in \{90\%, 95\%, 99\%\}$, i.e., the composed activity spaces. We depict them in the second row of Figure~\ref{fig:time_spent_maps}. The marked polygons in each panel indicate the dominant polygon aggregations for stationary exposure. In all three panels, two prominent polygon aggregations are clearly visible. Polygon aggregation 1 is located in a residential area, likely including the participant’s home, where most time is spent. Polygon aggregation 2 represents the primary daytime activity hub, comprising multiple buildings where the participant conducts daily activities. Insets highlight these polygon aggregations, while smaller clusters capture infrequent or peripheral visits within the $99\%$ activity space shown in the rightmost panel of the second row of Figure~\ref{fig:time_spent_maps}. Polygon aggregation 3 functions like aggregation 2, while polygon aggregation 4 corresponds to a commercial space visited by the participant.

The road segments that comprise the activity spaces show how transitional exposure is distributed across the network, with a small portion of these segments accounting for most travel. Raising the coverage threshold from $90\%$ to $99\%$ enlarges the spatial extent: the $90\%$ activity space for road segments contains the primary route frequently used between home and the main daytime hub, along with another major road connecting to aggregation 3 of polygons. The $95\%$ composed activity space is similar to the $90\%$ space, while the $99\%$ space expands into more peripheral areas, encompassing aggregations 3 and 4, in addition to less frequently traveled routes. This stepwise expansion highlights the heavy-tailed nature of human mobility, where a small number of anchors and routes account for most exposure, while numerous other locations are visited infrequently.

\subsection{Clustering daily activity trajectories}
\begin{figure}[htbp] 
	\centering
	\includegraphics[width=0.32\linewidth]{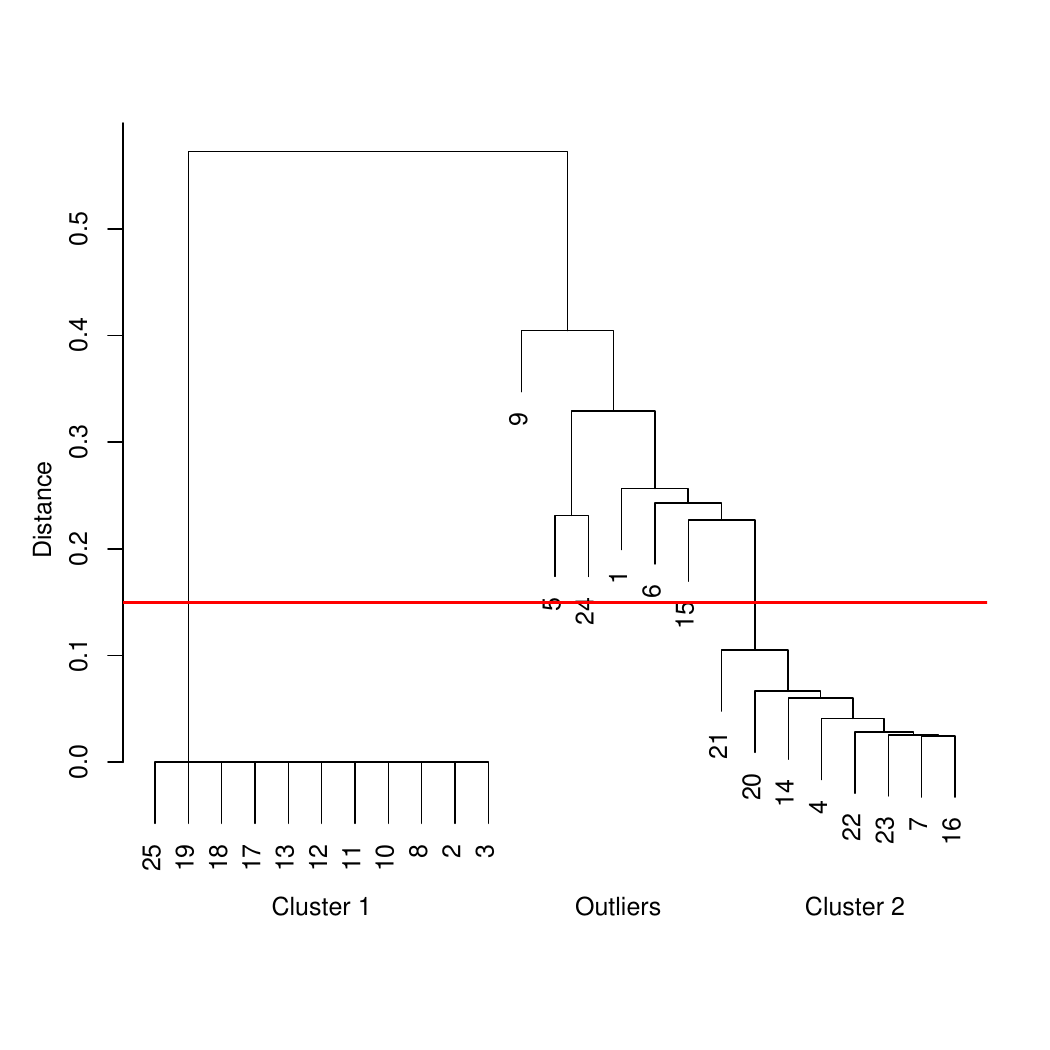}
	\includegraphics[width=0.32\linewidth]{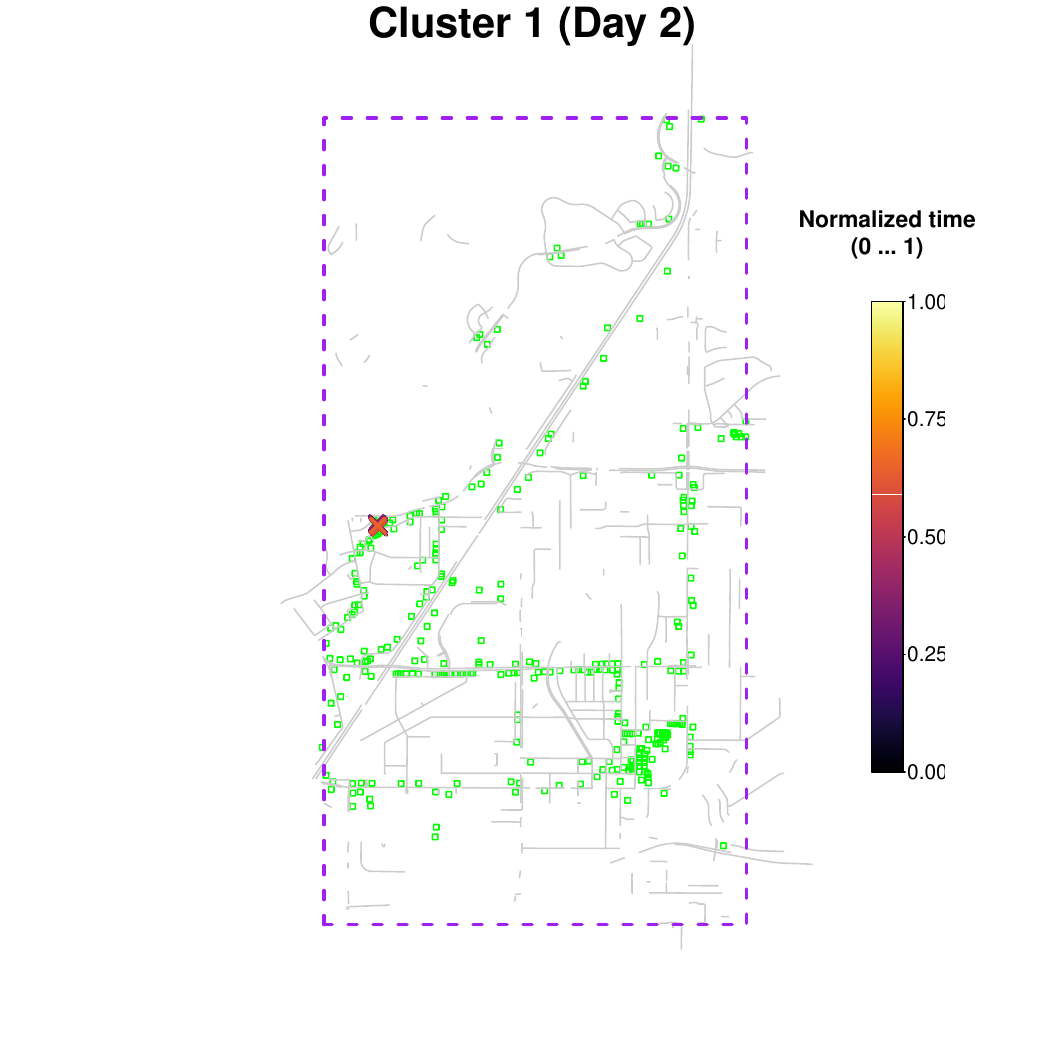}
	\includegraphics[width=0.32\linewidth]{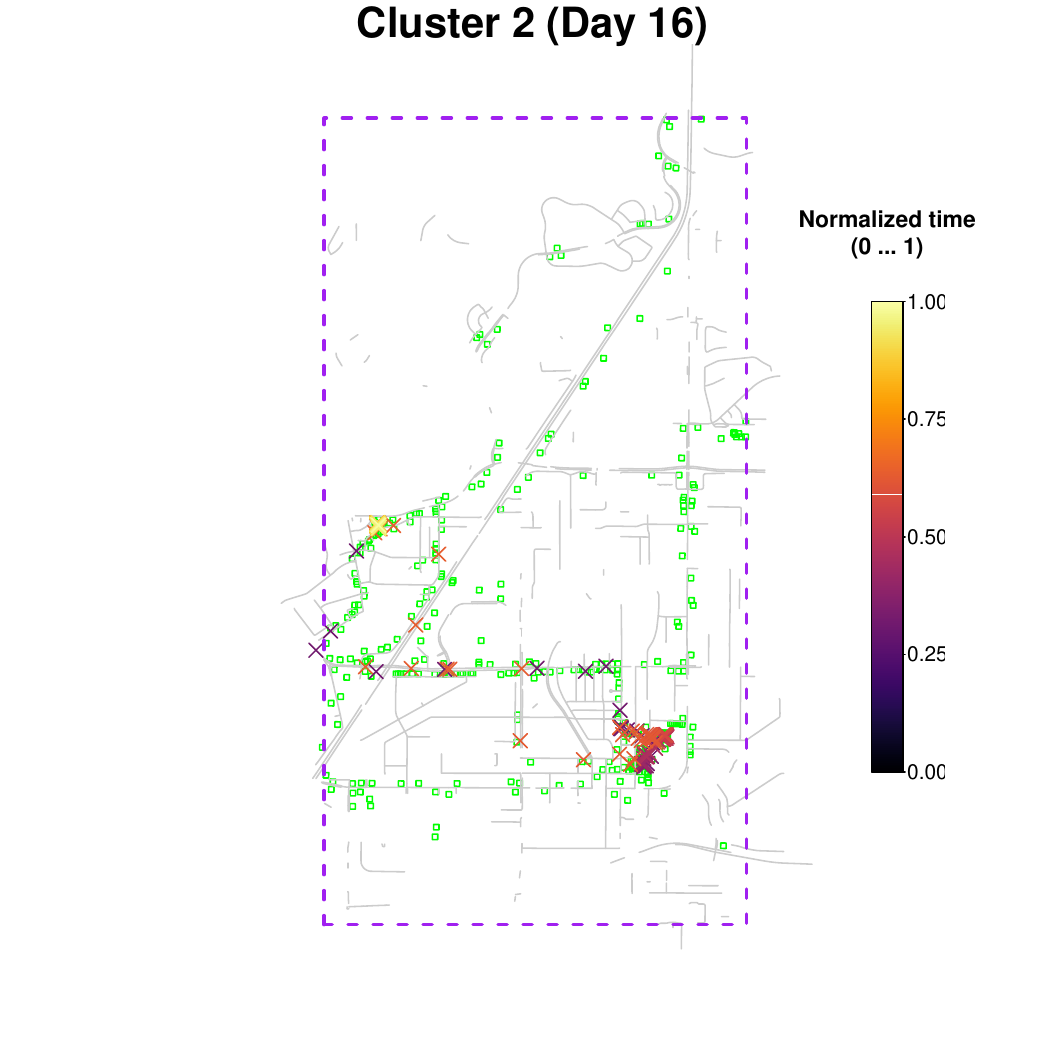}
	\caption{Left panel: hierarchical clustering dendrogram of daily trajectories (single linkage, Levenshtein distance weighted by dwell time). Two main clusters are identified, along with a few outlier days. Middle and right panels: representative daily mobility patterns randomly chosen from clusters 1 and 2.}
	\label{fig:clustering}
\end{figure}

To more precisely capture day-to-day heterogeneity, we clustered daily activity trajectories using their time-spent distributions. For each day, the ordered list of visited entities (buildings and road segments) was encoded as a symbolic trajectory, and pairwise dissimilarities were derived via a Levenshtein distance weighted by dwell time. This distance measure incorporates both the ordering of locations and the proportion of time spent at each. We then performed hierarchical clustering with single linkage on the resulting distance matrix.

The left panel of Figure~\ref{fig:clustering} shows the hierarchical clustering dendrogram, indicating that the days separate into two main clusters and a small group of outlier days. Cluster~1 includes days when the individual remained mostly at home with minimal movement. Cluster~2 consists of days characterized by regular travel between home and the main daytime hub, with occasional minor detours or activities. The outlier days exhibit distinct patterns that do not align with either of the two main clusters. Representative days from each cluster highlight internal similarities and minor variations. To illustrate our findings, we randomly select one day from cluster~1 and one from cluster~2. The middle panel of Figure~\ref{fig:clustering} displays all observations from the selected day in cluster~1, reflecting a largely stationary pattern at home. The right panel of Figure~\ref{fig:clustering} shows all observations from a randomly chosen day in cluster~2, reflecting a travel pattern between home and the primary daytime hub.

\subsection{Polygon- and road-based LCT results}\label{subsec: real-data:LCT}

\begin{figure}[htbp]
	\centering
    	\includegraphics[width=0.32\linewidth]{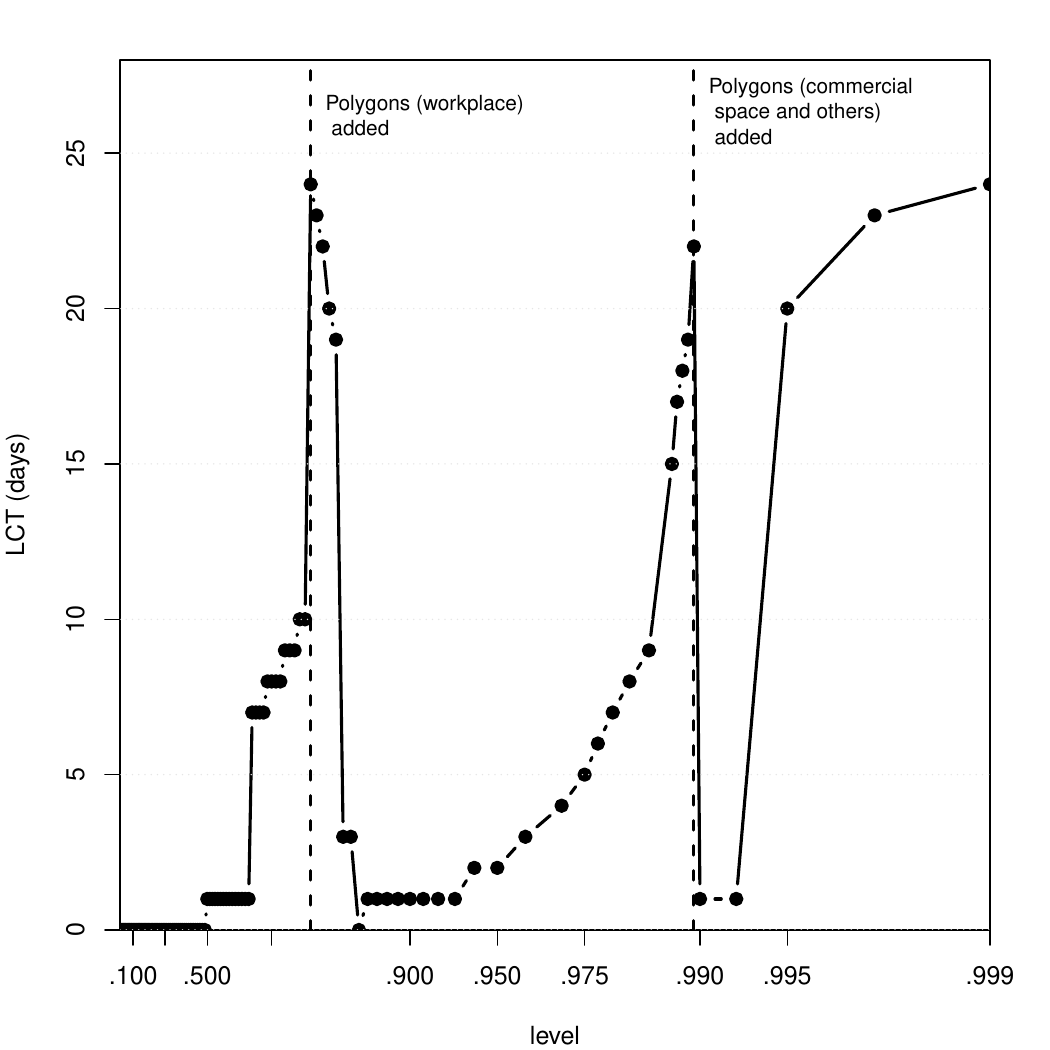}
	\includegraphics[width=0.32\linewidth]{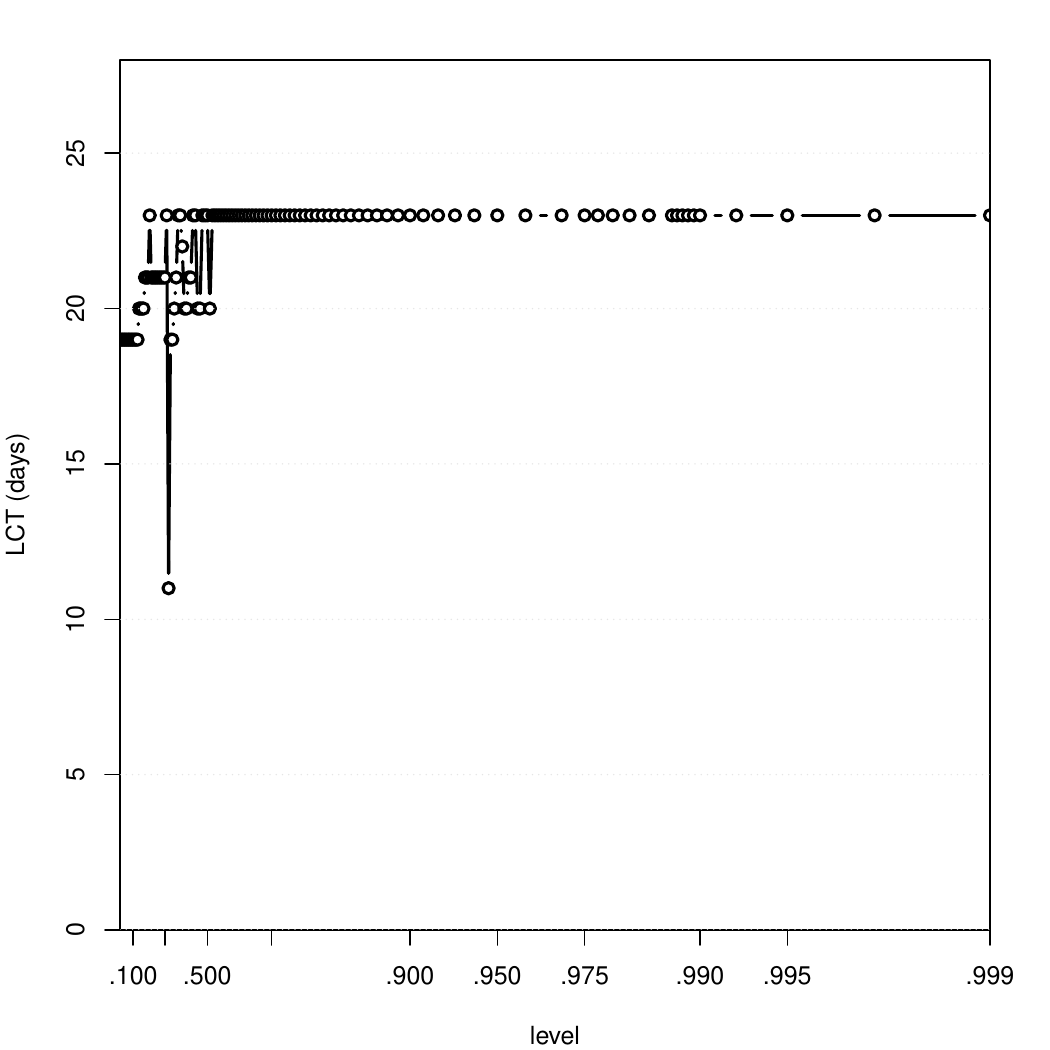}
	\includegraphics[width=0.32\linewidth]{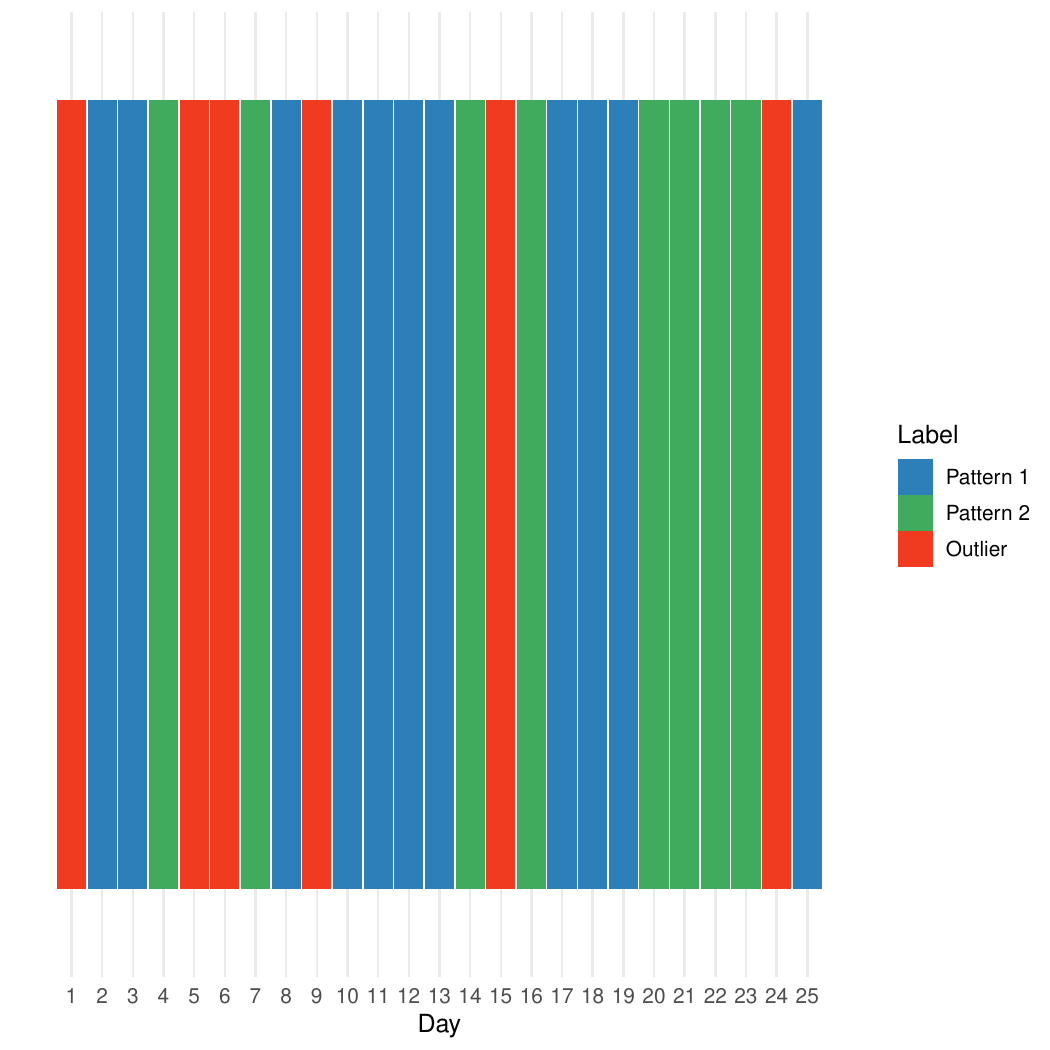}
	\caption{
		Left panel: Last--crossing times (LCTs) for polygon--based activity spaces across different levels~$c$. Dashed vertical lines mark the inclusion of new polygons: as the level of activity space increases
		polygon for workplace at $c=0.78$ and polygons  and other occasional sites at $c=0.99$. Middle panel: LCTs for road--based activity spaces, which show a generally higher and more stable trend, reflecting persistent variability in route choices. Right panel: Daily pattern labels classified into three states (Pattern~1, Pattern~2, and \emph{Outlier}) illustrate alternating patterns of daily mobility and occasional irregular days.
	}
	\label{fig:LCT-all}
\end{figure}

Figure~\ref{fig:LCT-all} provides an overview of the temporal stability of both polygon- and road-based activity spaces, together with the day-to-day mobility patterns observed over the 25-day study period. The left panel of Figure~\ref{fig:LCT-all} presents the LCTs for the polygon-based activity space across various coverage thresholds. At low coverage levels ($c \le 0.7$), LCT values remain near zero, suggesting that the activity space stabilizes quickly around a single predominant polygon. This polygon represents the person’s home, indicating that time spent at this main stationary location becomes consistently dominant early in the observation period.

As the coverage level approaches $c=0.78$, the workplace polygon becomes contained within the activity space, causing a sharp rise in the LCT. As outlined in Section~\ref{subsec:LCT and pattern}, this jump illustrates the sensitivity of the LCT to the threshold values at which additional polygons enter the activity space. At this coverage level, the workplace polygon is near the boundary of the cumulative probability, so minor differences between the cumulative distributions from the first $D$ days and the entire observation window can determine whether the polygon is included in $\mathcal{AS}^{\mathcal{A}}_{c}(D)$. Consequently, the provisional activity spaces based on early days differ from the complete activity space, generating high LCT values despite the established home–work pattern.

In line with this threshold-sensitive behavior, the LCT declines steeply at slightly higher coverage levels once the workplace polygon is consistently included. Beyond this stage, additional increases in $c$ no longer introduce new dominant stationary sites, and the cumulative activity distribution converges quickly. An analogous pattern arises as the coverage level nears $c \approx 0.99$, when polygons associated with commercial areas and other rarely visited places are incorporated into the activity space. These sites receive small but nonzero time allocations and therefore only appear near the upper tail of the cumulative distribution. Their incremental addition again induces a localized rise in the LCT, followed by a rapid drop once these polygons are fully integrated.
Taken together, the polygon-based LCT profile exhibits a distinct hierarchical organization: a very stable core of routine locations (home), then a regularly frequented secondary location (work), and finally, an extended tail of occasional destinations.

The middle panel of Figure~\ref{fig:LCT-all} illustrates the road-based LCTs. Unlike the polygon-based findings, the road-based LCT values remain comparatively high and stable across almost all coverage thresholds. This pattern suggests a slower convergence of the road-based activity space and points to sustained day-to-day variation in route selection. Although the individual frequently uses a core subset of roads for routine travel, sporadic detours, alternative paths, and exploratory trips continually introduce additional road segments, thereby postponing the point at which the road network representation stabilizes. The right panel of Figure~\ref{fig:LCT-all} shows the daily mobility pattern labels for all 25 days, assigning each day to Pattern~1, Pattern~2, or an \emph{Outlier}. The switching between Pattern~1 and Pattern~2 reflects recurring mobility regimes, such as workdays versus non-workdays, whereas isolated outlier days capture atypical schedules or unusual travel behavior.

The three panels highlight complementary aspects of participants' mobility dynamics. Stationary activity spaces stabilize quickly but exhibit localized rises in LCT with new polygons; road-based activity spaces achieve stability more gradually due to variable movement paths; and daily pattern classification reveals short-term behavior shifts and rare anomalies within an otherwise stable long-term spatial routine.

\section{Conclusion and future research directions}\label{sec:conclusion}

This article introduces a polygon-network framework to characterize human activity spaces from GPS data collected in built environments. Our framework represents daily mobility through time spent in spatial polygons and along road segments, focusing on entity-specific time-use proportions and their associated level-$\gamma$ activity spaces as the main inferential targets. It includes a time-weighted estimator for irregularly sampled GPS data, a map-enhanced depiction of daily activity patterns, a dwell-time-weighted distance metric for clustering trajectories, and stability measures for polygon-based and road-based activity spaces. Simulation studies and an empirical application demonstrate that appropriate weighting is crucial under irregular sampling, as stationary and movement components can exhibit markedly different empirical patterns. Our methodology provides a foundation for statistical inference on mobility summaries in GIS-based environments and outlines a promising direction for integrating activity-space methods with object-oriented and network-oriented spatial statistics \citep{menafoglio-secchi-2017,NKDEBook,baddeley-et-2021}.

Several extensions warrant further investigation. One avenue for future research is integrating more complex stochastic models for GPS uncertainty and formal map-matching, ensuring preprocessing uncertainty is preserved in subsequent estimation and clustering steps \citep{ranacheret-2016,newson2009hidden,lou2009map-matching}. A second avenue is designing population-level variants of the framework to facilitate systematic comparisons across individuals and subgroups, linking polygon-network activity spaces to contextual exposure assessments in health and social science applications \citep{RN149,RN144,RN146}. 

\section{Acknowledgments}

H.~W. and Y.-C.~C. gratefully acknowledge the support of NSF DMS-2310578, NSF DMS-2141808 and NIH U24-AG072122. A.~D. gratefully acknowledges the support of NIH 1R01MH133488 and NIH 1R01MH131480. The funders had no role in the study design, data collection and analysis, decision to publish, or preparation of the manuscript.


\pagebreak
\bigskip
\begin{center}
	{\large\bf SUPPLEMENTARY MATERIAL}
\end{center}

\appendix
\section{Handling key issues in the analysis of GPS data}
\subsection{Privacy-preserving rendering of application maps}
\label{app:privacy_rendering}

The application maps involve fine-scale spatial information, including local road geometry, building footprints, and GPS observations. Such information can be visually distinctive even when addresses and place names are removed. We therefore used a privacy-preserving rendering procedure for all figures displaying road segments or polygonal spatial entities. The procedure was applied only to the graphical layers for producing figures. It did not modify the data used for estimation, the definition of polygon and network entities, the assignment of observations to entities, or the resulting time-use and activity-space estimates.

\subsubsection{Road-network thinning}
Let $\mathcal{R}$ denote the primary road-segment layer and $\mathcal{X}=\{X_1,\ldots,X_n\}$ the set of observed GPS locations. For each road segment $r\in\mathcal{R}$, we computed its minimum distance to the observed GPS locations,
\[
d(r,\mathcal{X})=\min_{1\le i\le n} d(r,X_i),
\]
where $d(r,X_i)$ is the geometric distance from the road segment to the observation. Road segments that meet $d(r,\mathcal{X})>r_0$ were considered to have no nearby observations and were eligible for graphical thinning. From this eligible set, a fixed fraction $q$ was randomly selected for removal from the displayed road layer. The remaining segments formed the primary displayed road layer. In the figures of this paper, we used $r_0=50$ meters and $q=0.85$.

The primary road layer defined consistent keep/remove decisions for other road-network layers shown in the figures. If a segment in another road layer overlapped a segment removed from the primary layer, it was also removed from display. If it overlapped a segment retained in the primary layer, it was kept, except in cases of overlap with both retained and removed primary segments, where removal took priority for privacy. Remaining segments in additional road layers with no nearby observations were thinned using the same random-removal fraction. This consistency rule prevents segments removed from one layer from reappearing in another.

\subsubsection{Polygon reshaping}
The original polygonal entities, such as building or parcel footprints, were not displayed directly in the privacy-preserving maps. Instead, each displayed polygon was replaced by a square with a fixed side length and the same centroid as the original polygon. The attribute table and entity identifiers were preserved, allowing the squares to be linked to the estimated time-use proportions and activity-space membership. This transformation removes potentially identifying footprint characteristics, including exact shape, area, boundary geometry, and orientation, while retaining the approximate spatial arrangement needed to interpret the activity-space figures.

\subsubsection{Scope of the transformation}
Both road thinning and polygon reshaping were used only for rendering maps. All statistical summaries, nearest-entity assignments, estimated time-use proportions, activity-space thresholds, clustering results, and LCR in stability analyzes were computed using the original spatial data. The privacy-preserving layers affect only the visual presentation of the results. The figures also omit street names, addresses, basemap labels, and study-location identifiers.
\subsection{Pre-processing GIS data}\label{subsec:GIS_info_select}

To concentrate the analysis on the geographic context that matters for individual mobility patterns, we must delineate a compact study area from the GPS data and then extract the road network and building polygons that intersect with, or lie close to, the observed spatiotemporal trajectories.

\subsubsection{Determining the boundary of the study area}

Let $(x_i,y_i)$ denote the longitude and latitude of observation $i$ with a non-negative dwell-time weight $w_i$. We sought an axis-aligned rectangle with edge length $< \theta$ (in degrees) that maximizes the total weight it contains. Since the full set may not fit within a single $\theta \times \theta$ box, we approximated the maximizer using a grid search with weighted counts on a fine grid.

In the examples presented in this writing, we set $\theta=0.05^{\circ}$ and $r=0.001^{\circ}$, yielding a bounding rectangle of approximately $0.05^{\circ}\!\times\!0.05^{\circ}$ ($\approx 5.5 \times 5.5$ km at the study latitude). This area includes about 97.7\% of the total weighted observations while excluding distant outliers. We consider this area the effective study boundary for further spatial selection.

\subsubsection{Selecting the road network}

OpenStreetMap (OSM) road geometries were queried within the bounding box of several road types to evaluate the trade-off between point coverage and network sparsity:

\begin{itemize}
	\item \textbf{Primary/secondary only}: \{\texttt{primary}, \texttt{secondary}\}
	\item \textbf{Arterial set}: \{\texttt{motorway}, \texttt{trunk}, \texttt{primary}, \texttt{secondary}, \texttt{tertiary}\}
	\item \textbf{Arterial--local mix}: \{\texttt{motorway}, \texttt{trunk}, \texttt{primary}, \texttt{secondary}, \texttt{tertiary}, \texttt{residential}\}
	\item \textbf{All classes}: \{\texttt{highway} = *\}
\end{itemize}

For each network, we calculated the minimum distance from every GPS point to the nearest road segment and summarized the coverage proportion, i.e., the fraction of observations within $d_0=150\,\text{m}$ of the network. Distances were evaluated in geographic degrees, with $0.00135^{\circ} \approx 150\,\text{m}$ at the study latitude. The coverage proportions are shown in Table \ref{tab:coverageproportions}. The arterial--local mix achieves nearly complete coverage without excessive density from including every minor path. We adopt this network for subsequent analysis.

\begin{table}
\begin{center}
	\begin{tabular}{c c c c}
		\hline
		Primary/secondary only & Arterial set & Arterial--local mix & All classes \\
		\hline
		0.647 & 0.655 & 0.988 & 1.000 \\
		\hline
	\end{tabular}
    \caption{Coverage proportions of road types.}
    \label{tab:coverageproportions}
\end{center}
\end{table}

\subsubsection{Selecting the spatial polygons}

To reduce irrelevant features, we retained only the building polygons closest to at least one GPS observation within $d_0=150\,\text{m}$ of the polygon set. We computed the point–polygon distance matrix, identified the nearest polygon for each point, and retained the unique polygons whose closest points were within $0.00135^{\circ}$. This yields a parsimonious set of structures most likely visited by the individual whose mobility pattern is recorded.

The selection workflow is summarized in Figure~\ref{fig:gis_selection_triptych}. The left panel displays all polygons, road segments, and GPS points within the bounding box. The middle panel shows the selected arterial--local road subset, achieving high coverage while remaining interpretable. The right panel shows the final polygon selection overlaid with roads and GPS points for spatial context.

\begin{figure}[htbp]
	\centering
	\includegraphics[width=0.32\linewidth]{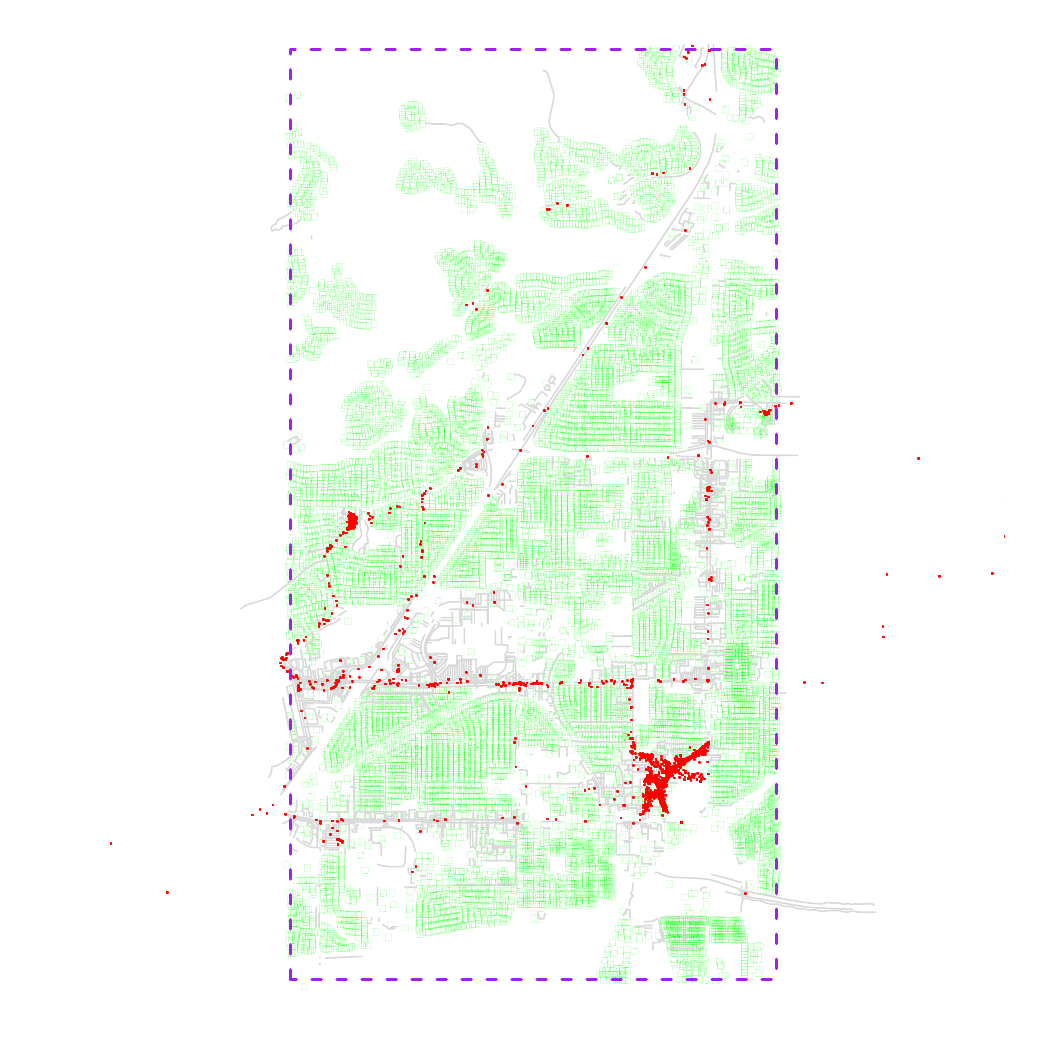}
	\includegraphics[width=0.32\linewidth]{privacy_r1-2-dense.pdf}
	\includegraphics[width=0.32\linewidth]{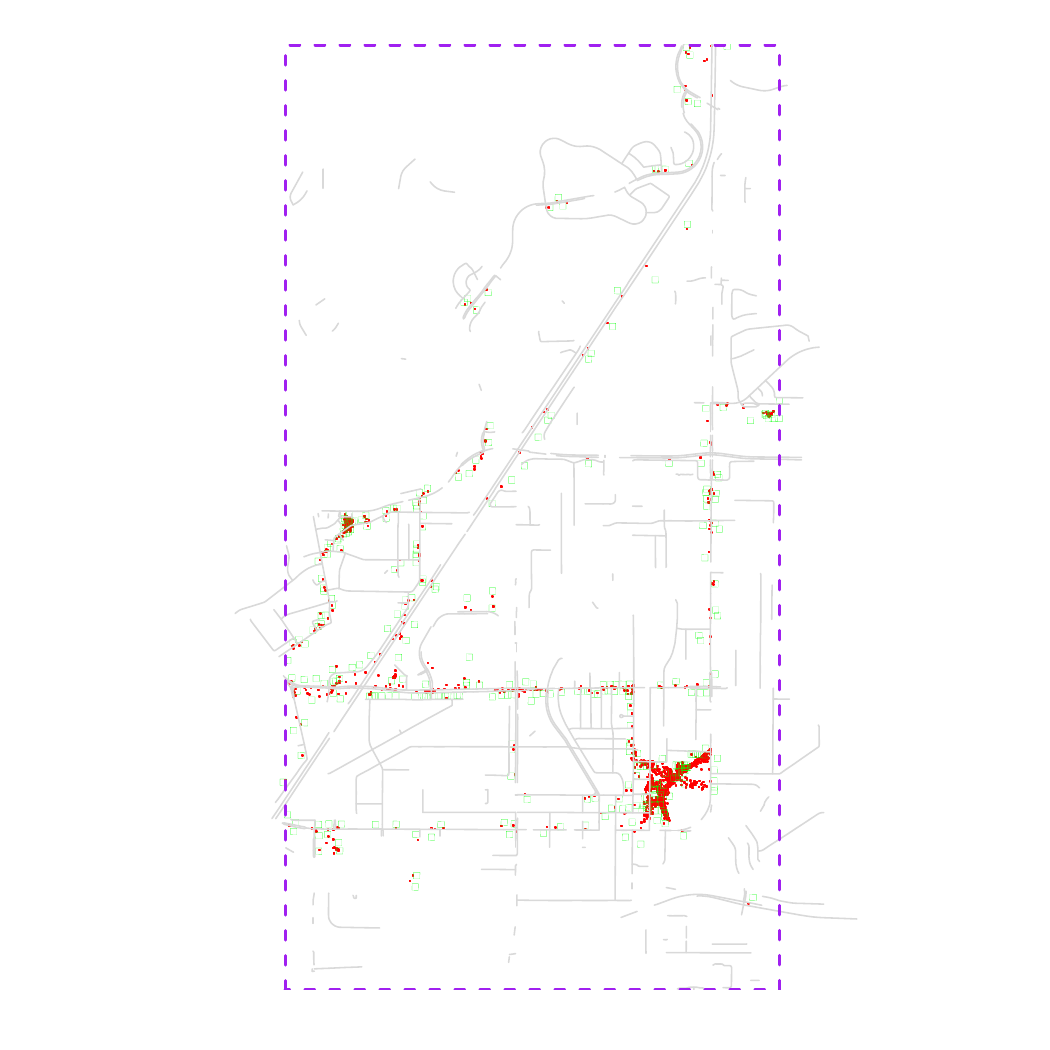}
	\caption{GIS selection workflow displayed in WGS84 (EPSG:4326). Polygons and road segments are shown in green, while GPS observations are shown in red. Left panel: full infrastructure inside the data-driven bounding box. Middle: road subset chosen to maximize coverage of points within $d_0=150$\,m while limiting density. Right panel: polygons retained based on proximity to observations (nearest within $0.00135^{\circ}\!\approx\!150$\,m).}
	\label{fig:gis_selection_triptych}
\end{figure}

\subsubsection{Aggregating the spatial polygons}\label{subsec:polygon_aggregation}

In dense residential areas, polygons representing buildings are often small and packed closely together. When GPS observations map to individual polygons, spatial fragmentation can disperse dwell-time estimates. Although individuals may spend significant time in a single functional location (e.g., a residential complex), measurement noise and positioning errors can lead to observations being intermittently assigned to neighboring polygons. 
As a result, time spent in a location may be distributed across many small polygons, each receiving only a minor estimated dwell time. This complicates the interpretation of activity patterns and increases the apparent complexity of the activity space.

To mitigate this issue, we aggregate nearby building polygons into composite entities before estimating activity and clustering with map-SMM. The aggregation relies on spatial proximity; we compute the centroids of the polygons and conduct single-linkage hierarchical clustering, merging those whose centroids are within a specified distance threshold. 
This procedure groups contiguous structures that likely represent a single functional location, such as adjacent residential buildings or interconnected units in the same housing complex.

In our data examples, we use a distance cutoff of $D=0.001^{\circ}$ (approximately 100 meters at the study latitude) to balance reducing over-fragmentation while preserving meaningful spatial distinctions between locations. 
Each resulting cluster is treated as a single aggregated polygon in subsequent analyzes, and GPS observations assigned to any polygon within the cluster contribute to the same activity entity. This aggregation step enhances the stability and interpretability of dwell-time estimation by minimizing sensitivity to GPS measurement error and polygon granularity while preserving the essential spatial structure for clustering daily activity patterns.

\subsection{Adjusting entity assignments in the presence of borderline GPS observations}
\label{subsec:assignment-adjustment}

A practical challenge in applying our method to real GPS data is the presence of borderline observations at the boundaries between polygons (e.g., residential buildings) and adjacent road segments. Entity assignment relies on the nearest geodesic distance; thus, a GPS point  close to a polygon may be assigned to a road segment due to GPS noise, slight geometric irregularities, or map–matching artifacts. In high–frequency GPS traces, this issue often causes unrealistic oscillations between a polygon and a road segment—for example, a sequence of polygon $\rightarrow$ road $\rightarrow$ polygon, even when the individual is clearly inside or adjacent to the polygon.

These oscillatory assignments do not accurately reflect the underlying human mobility processes. If two consecutive observations fall within the same polygon (or road segment) but the timestamp between them is assigned to a different entity due to noise, the estimated time allocation may inflate road time while deflating time spent in polygons. Such distortions accumulate over many days, particularly for entities with small spatial footprints, and can bias both naive and weighted estimators.

Humans rarely transition directly between an indoor space and a nearby road. Thus, if the midpoint in a triple $(E_{i,j-1},E_{i,j},E_{i,j+1})$ is inconsistent with its neighbors' movement patterns, the discrepancy is likely due to measurement noise rather than true movement. Correcting these inconsistencies enhances robustness in applications, especially when polygons represent indoor locations or areas with degraded GPS accuracy (e.g., high-density urban areas).

Before computing the weighted estimator, we modify the entity assignments as follows. Let $X_{i,j}$ be the GPS location at time $t_{i,j}$ and let $E_{i,j}$ denote its nearest-entity assignment.  
For each observation $(i,j)$, take the following actions.

\begin{itemize}
	\item \textbf{Polygon override for road segments.}  
	If $E_{i,j}$ corresponds to a road segment, $E_{i,j-1}$ and $E_{i,j+1}$ represent the same polygon $e$, and $d(X_{i,j},e)$ is less than a fixed distance threshold, then $E_{i,j}$ is reassigned to polygon $e$.
	
	\item \textbf{Road segment override for polygons.}  
	If $E_{i,j}$ corresponds to a polygon, $E_{i,j-1}$ and $E_{i,j+1}$ are the same road segment $e$, and 
	$d(X_{i,j},e)$ is below the same threshold,
	then $E_{i,j}$ is reassigned to road segment $e$.
\end{itemize}

Let $\tilde{E}_{i,j}$ denote the adjusted entity labels.  
After this adjustment, we apply the weighted estimator
\[
\widehat{T}_{e}^{\text{adj}}
\;=\;
\frac{1}{n}\sum_{i=1}^{n}\sum_{j=1}^{m_i}
W_{i,j}\,
1\!\left(\tilde{E}_{i,j}=e\right),
\]
where the weights $W_{i,j}$ represent time spent at $X_{i,j}$. This adjustment removes spurious entity switching from GPS noise or boundary artifacts, enhancing the stability and interpretability of estimated time allocation in real-data applications.

\subsection{Pre-processing loops in pattern clustering}

Measurement noise may create short jitter loops of the form
\(
e,x_{1}\dots x_{r},\,e
\)
with a very short duration for $x_{1},\dots ,x_{r}$. Assume an individual is at their home designated as entity $\mathrm{P1}$ from midnight to 11 a.m. Table \ref{tab:threerecords} provides three consecutive GPS observations. The second observation was collected at 9:00 a.m. and is projected incorrectly. A momentary positioning error projects the observation at 09:03 onto road segment $\mathrm{S2}$, producing the matched sequence
\(\mathrm{P1}\rightarrow\mathrm{S2}\rightarrow\mathrm{P1}\).
After weighting the observations by their timestamps, we obtain the estimated action vector fragment and the time spent estimation
\[
(\mathrm{P1},\mathrm{S2},\mathrm{P1}),
\qquad
(9/24,\,1/480,\,13/160),
\]
where the numbers represent the proportions of the day accounted
for by the three observations.  
Because the middle entity $\mathrm{S2}$ lasts only $1/480$ (a relatively tiny proportion) of the
day and is flanked by identical entities $\mathrm{P1}$, the sequence
forms a jitter loop.  

\begin{table}
\begin{center}
	\begin{tabular}{@{}cccc@{}}
		\toprule
		Clock time & True activity & Raw coordinates $(x,y)$ & Matched entity \\ \midrule
		08:57      & home          & $(0.42,\,-0.08)$      & $\mathrm{P1}$ \\
		09:00      & home          & $(0.43,\,-0.07)$      & $\mathrm{S2}$ \\
		09:03      & home          & $(0.42,\,-0.08)$      & $\mathrm{P1}$ \\ \bottomrule
	\end{tabular}
\end{center}
\caption{Three consecutive GPS records collected from an individual.}
\label{tab:threerecords}
\end{table}

Given a tolerance $\tau>0$, we detect such loops and replace
them with the entity $e$ whenever
$\sum_{s=1}^{r} x_{s}\le\tau$.  
This compression maintains the total dwell time of the surrounding entity
$e$ while eliminating short loops caused by one or two mismatched GPS
observations due to
measurement error—spurious, one-off detours that are behaviorally
irrelevant.



\begin{figure}[htbp]
	\centering
	\includegraphics[width=0.32\linewidth]{privacy_cluster2_2_2_gps.pdf}
	\includegraphics[width=0.32\linewidth]{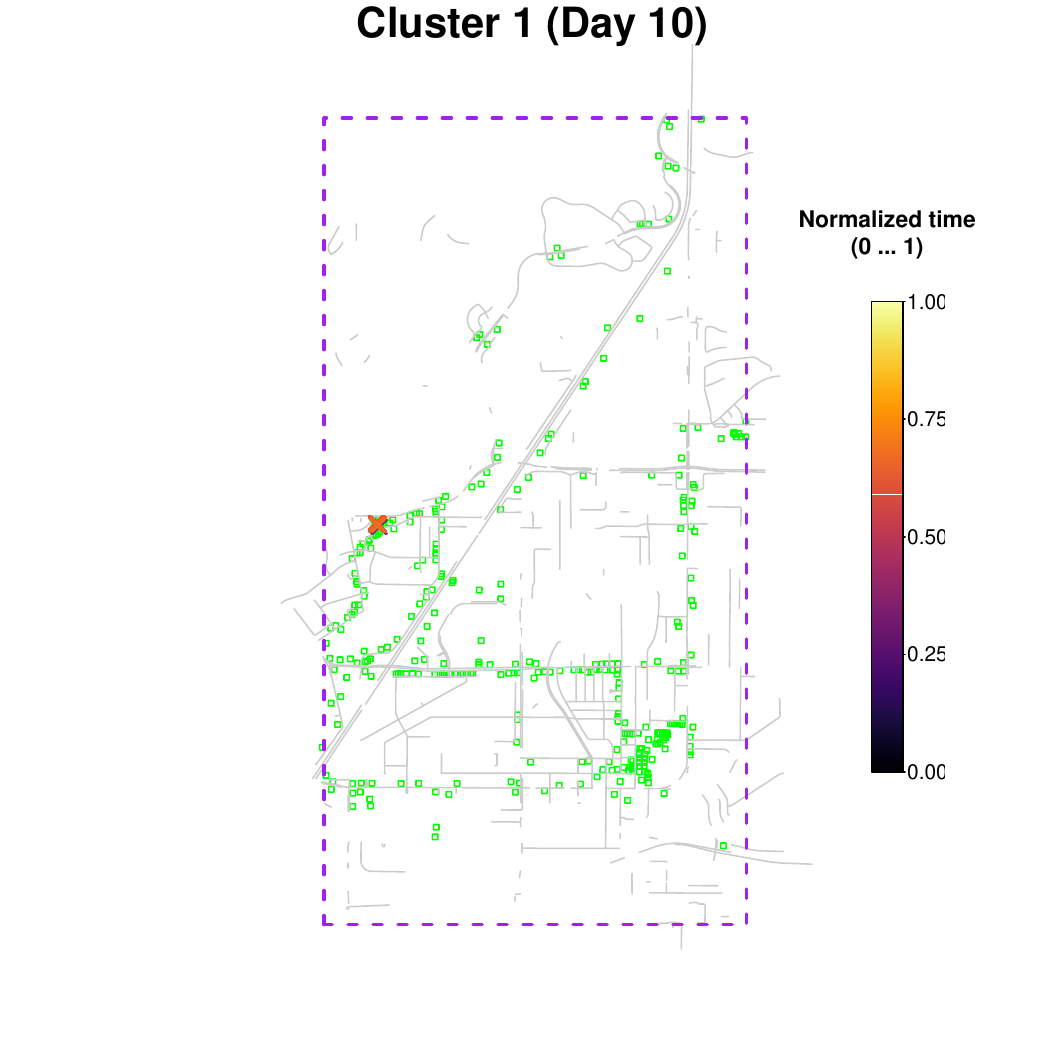}
	\includegraphics[width=0.32\linewidth]{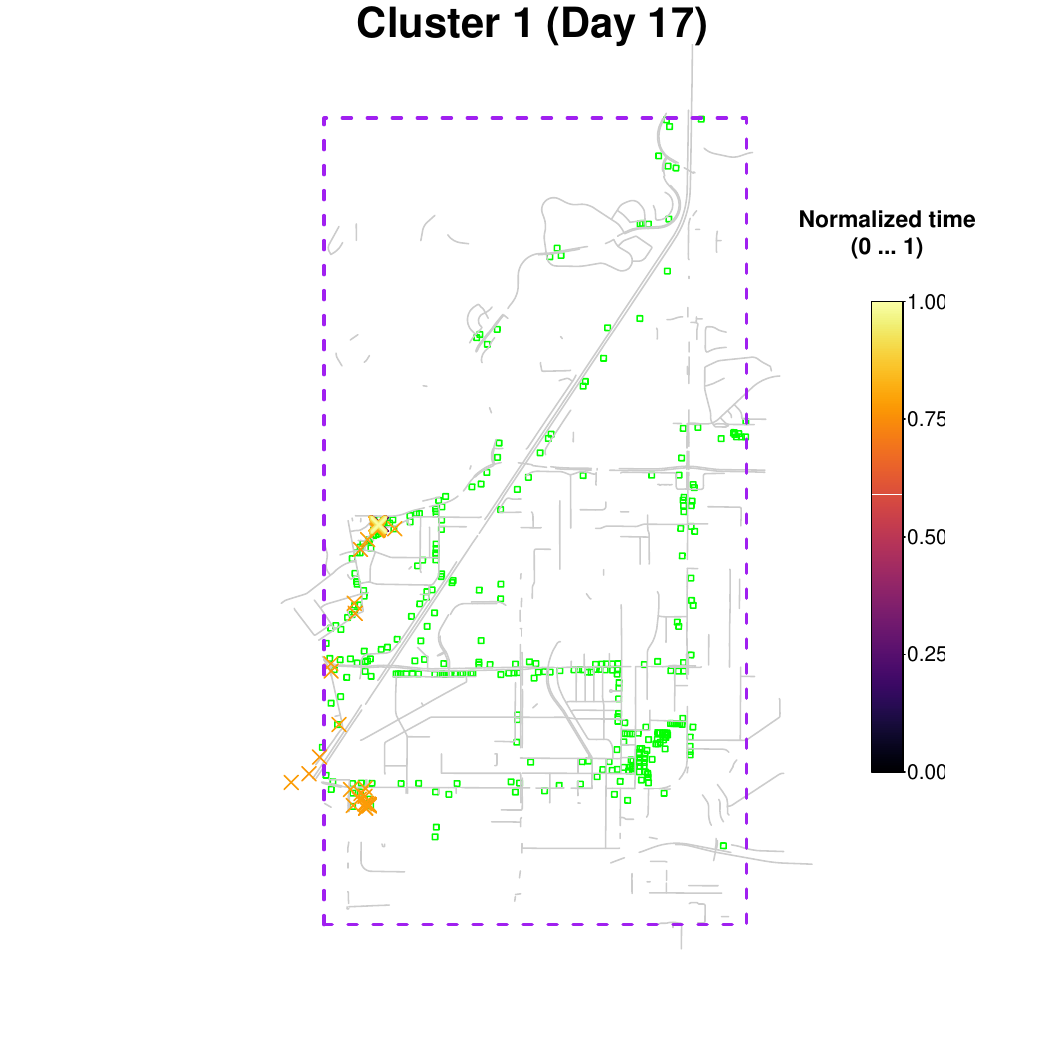}
	\caption{Representative days from Cluster~1 (days 2, 10, 17). These days are
		dominated by stationary activity at the living place, indicating limited
		mobility.}
	\label{fig:cluster1days}
\end{figure}

\begin{figure}[htbp]
	\centering
	\includegraphics[width=0.32\linewidth]{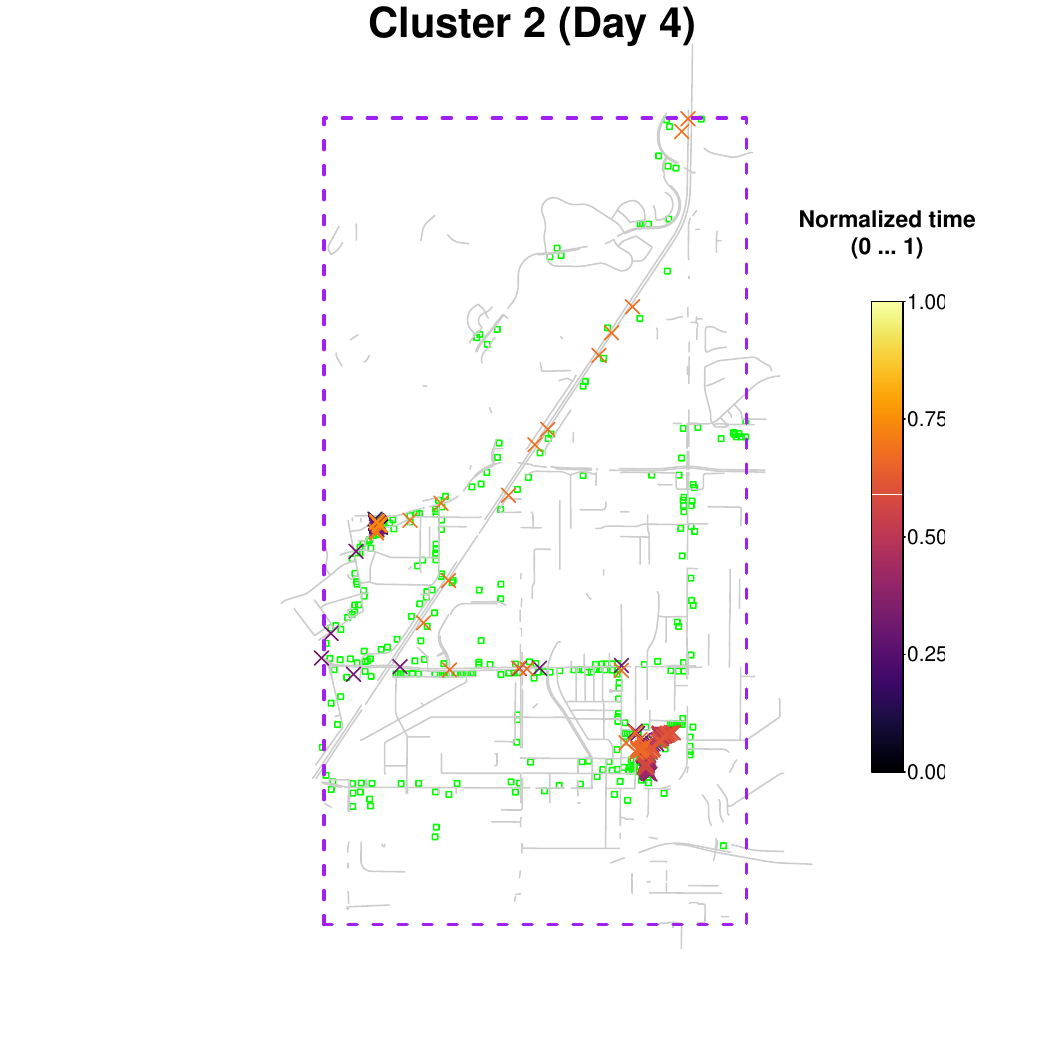}
	\includegraphics[width=0.32\linewidth]{privacy_cluster3_3_16_gps.pdf}
	\includegraphics[width=0.32\linewidth]{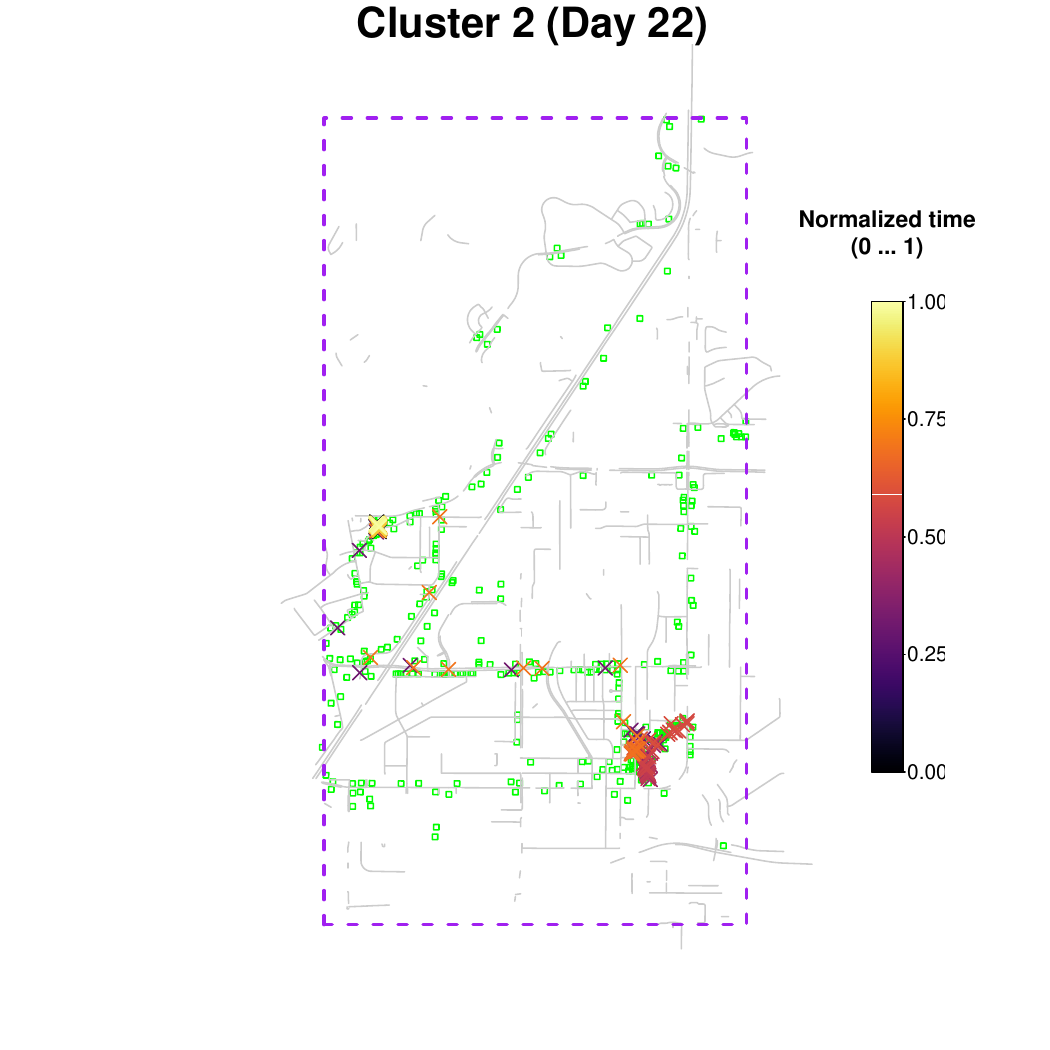}
	\caption{Representative days from Cluster~2 (days 4, 16, 22). These days show
		mobility between the living place and the main daytime hub, with minor
		variations such as detours or additional short-duration activities.}
	\label{fig:cluster2days}
\end{figure}

\section{Simulation Experiments} \label{simulationexperiments}
\subsection{Design of the experiments}\label{sec:timestamp-generation}

\begin{table}
	\centering
    {\tiny
	\begin{tabular}{|c|c|c|c|c|c|c|}
		\hline
		Day                       & Pattern                     & Probability             & \begin{tabular}[c]{@{}c@{}}Expected duration\\ (Hours)\end{tabular} & $\eta$ & $q$  & Actions               \\ \hline
		\multirow{24}{*}{Weekday} & \multirow{11}{*}{Pattern 1} & \multirow{11}{*}{15/28} & 9                                                                   & 0.15   & 0.5  & At home               \\ \cline{4-7} 
		&                             &                         & 0.5                                                                 & 0.08   & 0.25 & Home to office        \\ \cline{4-7} 
		&                             &                         & 2.5                                                                 & 0.15   & 0.5  & At office             \\ \cline{4-7} 
		&                             &                         & 0.1                                                                 & 0.03   & 0.05 & Office to office cafe \\ \cline{4-7} 
		&                             &                         & 0.5                                                                 & 0.15   & 0.05 & At office cafe        \\ \cline{4-7} 
		&                             &                         & 0.1                                                                 & 0.03   & 0.05 & Office cafe to office \\ \cline{4-7} 
		&                             &                         & 4.8                                                                 & 0.15   & 0.5  & At office             \\ \cline{4-7} 
		&                             &                         & 0.6                                                                 & 0.08   & 0.25 & Office to home        \\ \cline{4-7} 
		&                             &                         & 2                                                                   & 0.2    & 0.6  & At home               \\ \cline{4-7} 
		&                             &                         & 1                                                                   & 0.06   & 0.15 & Home to park to home  \\ \cline{4-7} 
		&                             &                         & 2.9                                                                 &        &      & Home                  \\ \cline{2-7} 
		& \multirow{13}{*}{Pattern 2} & \multirow{13}{*}{5/28}  & 8.5                                                                 & 0.15   & 0.5  & Home                  \\ \cline{4-7} 
		&                             &                         & 0.5                                                                 & 0.08   & 0.25 & Home to office        \\ \cline{4-7} 
		&                             &                         & 3                                                                   & 0.15   & 0.5  & At office             \\ \cline{4-7} 
		&                             &                         & 0.1                                                                 & 0.03   & 0.05 & Office to office cafe \\ \cline{4-7} 
		&                             &                         & 0.5                                                                 & 0.15   & 0.05 & At office cafe        \\ \cline{4-7} 
		&                             &                         & 0.1                                                                 & 0.03   & 0.05 & Office cafe to office \\ \cline{4-7} 
		&                             &                         & 4.3                                                                 & 0.15   & 0.5  & At office             \\ \cline{4-7} 
		&                             &                         & 0.75                                                                & 0.08   & 0.25 & Office to restaurant  \\ \cline{4-7} 
		&                             &                         & 1                                                                   & 0.08   & 0.25 & At restaurant         \\ \cline{4-7} 
		&                             &                         & 0.4                                                                 & 0.08   & 0.25 & Restaurant to home    \\ \cline{4-7} 
		&                             &                         & 0.95                                                                & 0.1    & 0.3  & At home               \\ \cline{4-7} 
		&                             &                         & 1                                                                   & 0.06   & 0.15 & Home to park to home  \\ \cline{4-7} 
		&                             &                         & 2.9                                                                 &        &      & At home               \\ \hline
		\multirow{11}{*}{Weekend} & \multirow{5}{*}{Pattern 3}  & \multirow{5}{*}{4/28}   & 11                                                                  & 0.3    & 1    & At home               \\ \cline{4-7} 
		&                             &                         & 0.75                                                                & 0.15   & 0.45 & Home to supermarket   \\ \cline{4-7} 
		&                             &                         & 2.5                                                                 & 0.3    & 1    & At supermarket        \\ \cline{4-7} 
		&                             &                         & 0.75                                                                & 0.15   & 0.45 & Supermarket to home   \\ \cline{4-7} 
		&                             &                         & 9                                                                   &        &      & At home               \\ \cline{2-7} 
		& \multirow{5}{*}{Pattern 4}  & \multirow{5}{*}{1/28}   & 10                                                                  & 0.3    & 1    & At home               \\ \cline{4-7} 
		&                             &                         & 0.8                                                                 & 0.3    & 0.7  & Home to beach         \\ \cline{4-7} 
		&                             &                         & 5.7                                                                 & 0.35   & 1    & Beach                 \\ \cline{4-7} 
		&                             &                         & 0.8                                                                 & 0.3    & 0.7  & Beach to home         \\ \cline{4-7} 
		&                             &                         & 6.7                                                                 &        &      & At home               \\ \cline{2-7} 
		& Pattern 5                   & 3/28                    & 24                                                                  &        &      & At home               \\ \hline
	\end{tabular}
    }
	\caption{The map-SMM model used for the simulation experiments. There are a total of five daily activity patterns. The time the individual spends at each location follows a truncated normal distribution with mean 
		(Expected Duration), standard deviation \(\eta\), and half-width \(q\). 
		Times shown are in hours. The last row associated with each pattern (where \(\eta\) and \(q\) are blank) absorbs any remaining time up to 24 hours.}
	\label{Table:pattern_detail}
\end{table}

To evaluate the performance of the proposed activity density estimators, we create a simplified scenario featuring a single fictitious individual. This individual's daily activities involve five spatial polygons: home, office, restaurant, beach, and supermarket, connected by several routes as shown in the top left panel of Figure \ref{fig:map_pattern}. The key aim of the simulation experiment is to estimate the time spent in each polygon and road segment.

We create five distinct daily activity patterns for this individual, as shown in Figure \ref{fig:map_pattern}. Key details about this individual's movements, including visit order, expected time at each location, and the probability of selecting each pattern, are summarized in Table \ref{Table:pattern_detail}. The simulated timeline corresponds to a regular week of five weekdays and two weekend days. Each day, the fictitious individual randomly selects one of the activity patterns using the probabilities specific to weekdays or weekends provided in Table \ref{Table:pattern_detail}.

At each selected pattern location, time spent follows a truncated normal distribution \(z_i \mid b \sim \mathcal{N}\bigl(\mu_{b,i},\,\eta_{b,i}^2\bigr)\) truncated to \(\bigl(\mu_{b,i}-q_{b,i},\;\mu_{b,i}+q_{b,i}\bigr)\), where \(\mu_{b,i}\) and \(\eta_{b,i}\) denote the mean and standard deviation, respectively, and \(q_{b,i}\in [0,1]\) specifies a reasonable range for the duration of stay. The length of stay at the last location in each pattern is calculated as the remaining time in the 24-hour day after visiting all previous locations.

We utilize a truncated Gaussian distribution to determine the length of stay at each location for two reasons. First, we need to control variance: by adjusting \(\eta_{b,i}\), we can account for realistic scenarios; for example, weekday departures tend to be less variable than weekend departures for a beach visit. Thus, we set the variance of the duration for the first position higher on weekends. The second reason is reasonable time distribution: truncation prevents unrealistic visit durations; for example, a restaurant stay should not exceed typical operating hours.

We generate timestamps of fictitious GPS observations based on the assumption that each day $i=1,2,\ldots,n$ has an equal number of observations, i.e., $m_i = m$. We set \(n \in \{7, 30, 90\}\) and \( m \in \{159, 479, 1439\}\). The chosen $n$ values represent one week, one month, and three months of GPS data collection. The chosen \(m\) values represent GPS observations recorded, on average, every 9 minutes, every 3 minutes, or once per minute. To investigate estimator performance under regular and irregular observation frequencies, we generate timestamps in two regimes:
\begin{enumerate}
	\item \textbf{Even-spaced timestamps:} In each simulated day, the \(m\) timestamps are spaced at \(\tfrac{1440}{m+1}\) minute intervals. The GPS observations are recorded every \(\frac{1440}{m+1}\) minutes.
	
	\item \textbf{Realistic timestamps:} To generate non-uniform timestamps similar to real GPS data, we use the GPS dataset described in Section~\ref{sec:real-data} as follows:
	\begin{enumerate}
		\item Randomly select an individual from the GPS dataset.
		\item For each observation day of the fictitious individual, randomly sample a day from the individual's observation period selected in (a). If the sampled day has fewer than \(m\) real timestamps, generate additional timestamps by sampling from a Gaussian kernel density estimate of that day's timestamp distribution using Silverman's rule of thumb for bandwidth. If the sampled day has more than \(m\) real timestamps, randomly remove timestamps until \(m\) remain, ensuring each observation has the same probability of removal.	
   \end{enumerate}
	This process yields \(m\) timestamps per simulated day, preserving the original non-uniformity of the individual's data. 
	As discussed in Section~\ref{sec:real-data}, the timestamp distribution of the real data is potentially skewed. 
	This procedure preserves the day-to-day skew and clustering patterns of timestamps found in real data. 
\end{enumerate}

We introduce measurement errors in simulated GPS observations by sampling from a Gaussian noise distribution with \(\sigma = 0.1\) or \(\sigma = 0.05\).

\subsection{Comparison of the mean proportion of time spent estimators}

We compare three estimators of $T_e$, the mean proportion of time spent in entity $e \in \mathcal{E}$ under even-spaced and realistic timestamp designs.

{\bf Method 1: Naive (equal--weight) estimator.}
The naive estimator assigns equal weight to each GPS timestamp, ignoring the time 
intervals between observations.  
For day $i$, the estimated proportion of time spent in $e$ is the fraction of that 
day's timestamps assigned to $e$:
\[
\widehat{T}_e^{\text{naive}}
=
\frac{1}{n}\sum_{i=1}^{n}
\left(
\frac{1}{m_i} \sum_{j=1}^{m_i}
1\!\left(E_{i,j}=e\right)
\right).
\]

{\bf Method 2: Weighted estimator.}
Our weighted estimator accounts for varying time gaps between successive GPS locations using the weights $W_{i,j}$ defined in Section~\ref{sec:estimation}.  
This yields
\[
\widehat{T}_e
=
\frac{1}{n}\sum_{i=1}^{n}\sum_{j=1}^{m_i}
W_{i,j} \,
1\!\left(E_{i,j}=e\right).
\]

{\bf Method 3: Weighted estimator with assignment adjustment.}
Our third estimator uses the same weights $W_{i,j}$ as the weighted estimator (Method~2) but first applies the boundary-adjustment procedure in Section~\ref{subsec:assignment-adjustment} to correct unrealistic oscillations in 
entity assignment caused by GPS noise.  
Let $\tilde{E}_{i,j}$ denote the adjusted entity labels.  
The adjusted weighted estimator is
\[
\widehat{T}_e^{\text{adj}}
=
\frac{1}{n}\sum_{i=1}^{n}\sum_{j=1}^{m_i}
W_{i,j} \,
1\!\left(\tilde{E}_{i,j}=e\right).
\]

We evaluate the performance of these three estimators by calculating the mean integrated 
squared error (MISE) for all entities and simulation replicates.

\subsubsection{Results on the MISE}
\begin{table}[]
	\centering
	\begin{tabular}{cccccccc}
		\hline
		\multicolumn{2}{c}{Even-spacing}                    & \multicolumn{3}{c}{$\epsilon=0.1$}                                       & \multicolumn{3}{c}{$\epsilon=0.05$}                     \\ \hline
		$n$                     & $m$                       & $\hat{T}_e^{\text{naive}}$ & $\hat{T}_{e}$ & $\hat{T}_{e}^{\text{adj}}$                  & $\hat{T}_e^{\text{naive}}$ & $\hat{T}_{e}$ & $\hat{T}_{e}^{\text{adj}}$\\ \hline
		\multicolumn{1}{c|}{7}  & \multicolumn{1}{c|}{159}  & 0.0402            & 0.0392        & \multicolumn{1}{c|}{\textbf{0.0335}} & 0.0339            & 0.0335        & \textbf{0.0335}     \\
		\multicolumn{1}{c|}{7}  & \multicolumn{1}{c|}{479}  & 0.0399            & 0.0396        & \multicolumn{1}{c|}{\textbf{0.0333}} & 0.0337            & 0.0335        & \textbf{0.0336}     \\
		\multicolumn{1}{c|}{7}  & \multicolumn{1}{c|}{1439} & 0.0398            & 0.0397        & \multicolumn{1}{c|}{\textbf{0.0334}} & 0.0336            & 0.0335        & \textbf{0.0335}     \\
		\multicolumn{1}{c|}{30} & \multicolumn{1}{c|}{159}  & 0.0171            & 0.0164        & \multicolumn{1}{c|}{\textbf{0.0091}} & 0.0090            & 0.0090        & \textbf{0.0090}     \\
		\multicolumn{1}{c|}{30} & \multicolumn{1}{c|}{479}  & 0.0169            & 0.0167        & \multicolumn{1}{c|}{\textbf{0.0090}} & 0.0090            & 0.0089        & \textbf{0.0089}     \\
		\multicolumn{1}{c|}{30} & \multicolumn{1}{c|}{1439} & 0.0166            & 0.0165        & \multicolumn{1}{c|}{\textbf{0.0090}} & 0.0090            & 0.0089        & \textbf{0.0089}     \\
		\multicolumn{1}{c|}{90} & \multicolumn{1}{c|}{159}  & 0.0122            & 0.0115        & \multicolumn{1}{c|}{\textbf{0.0034}} & 0.0032            & 0.0031        & \textbf{0.0031}     \\
		\multicolumn{1}{c|}{90} & \multicolumn{1}{c|}{479}  & 0.0117            & 0.0115        & \multicolumn{1}{c|}{\textbf{0.0033}} & 0.0031            & 0.0031        & \textbf{0.0031}     \\
		\multicolumn{1}{c|}{90} & \multicolumn{1}{c|}{1439} & 0.0115            & 0.0114        & \multicolumn{1}{c|}{\textbf{0.0033}} & 0.0031            & 0.0031        & \textbf{0.0031}     \\ \hline
	\end{tabular}
	\caption{%
		Square root of the mean integrated squared error (MISE) for estimating the
		time-use proportions $T_e$ under the even–spacing timestamp design.  
		Results are shown for three estimators—the naive equal–weight estimator
		$\widehat{T}_e^{\text{naive}}$, the weighted estimator
		$\widehat{T}_e^{\text{wgt}}$, and the adjusted weighted estimator
		$\widehat{T}_e^{\text{adj}}$—across combinations of the number of days $n$,
		the number of within-day observations $m$, and the reassignment threshold
		$\epsilon$.  
		Smaller values indicate more accurate estimation of the time spent in each
		entity.}
	\label{tab:even-mise}
\end{table}
\begin{table}[]
	\centering
	\begin{tabular}{cccccccc}
		\hline
		\multicolumn{2}{c}{Realistic}                    & \multicolumn{3}{c}{$\epsilon=0.1$}                                       & \multicolumn{3}{c}{$\epsilon=0.05$}                     \\ \hline
		$n$                     & $m$                       & $\hat{T}_e^{\text{naive}}$ & $\hat{T}_{e}$ & $\hat{T}_{e}^{\text{adj}}$                  & $\hat{T}_e^{\text{naive}}$ & $\hat{T}_{e}$ & $\hat{T}_{e}^{\text{adj}}$ \\ \hline
		\multicolumn{1}{c|}{7}  & \multicolumn{1}{c|}{159}  & 0.1691            & 0.0442        & \multicolumn{1}{c|}{\textbf{0.0353}} & 0.1393            & 0.0341        & \textbf{0.0341}     \\
		\multicolumn{1}{c|}{7}  & \multicolumn{1}{c|}{479}  & 0.1708            & 0.0398        & \multicolumn{1}{c|}{\textbf{0.0334}} & 0.1429            & 0.0336        & \textbf{0.0336}     \\
		\multicolumn{1}{c|}{7}  & \multicolumn{1}{c|}{1439} & 0.1837            & 0.0398        & \multicolumn{1}{c|}{\textbf{0.0332}} & 0.1539            & 0.0336        & \textbf{0.0336}     \\
		\multicolumn{1}{c|}{30} & \multicolumn{1}{c|}{159}  & 0.1315            & 0.0172        & \multicolumn{1}{c|}{\textbf{0.0097}} & 0.0995            & 0.0091        & \textbf{0.0092}     \\
		\multicolumn{1}{c|}{30} & \multicolumn{1}{c|}{479}  & 0.1255            & 0.0169        & \multicolumn{1}{c|}{\textbf{0.0091}} & 0.0937            & 0.0090        & \textbf{0.0090}     \\
		\multicolumn{1}{c|}{30} & \multicolumn{1}{c|}{1439} & 0.1291            & 0.0165        & \multicolumn{1}{c|}{\textbf{0.0090}} & 0.0976            & 0.0090        & \textbf{0.0090}     \\
		\multicolumn{1}{c|}{90} & \multicolumn{1}{c|}{159}  & 0.1232            & 0.0123        & \multicolumn{1}{c|}{\textbf{0.0036}} & 0.0901            & 0.0032        & \textbf{0.0032}     \\
		\multicolumn{1}{c|}{90} & \multicolumn{1}{c|}{479}  & 0.1179            & 0.0114        & \multicolumn{1}{c|}{\textbf{0.0033}} & 0.0856            & 0.0031        & \textbf{0.0031}     \\
		\multicolumn{1}{c|}{90} & \multicolumn{1}{c|}{1439} & 0.1222            & 0.0114        & \multicolumn{1}{c|}{\textbf{0.0033}} & 0.0895            & 0.0031        & \textbf{0.0031}     \\ \hline
	\end{tabular}
	\caption{%
		Square root of the mean integrated squared error (MISE) for estimating the time-use proportions $T_e$ under the realistic (non-uniform) timestamp design. The table compares the naive equal–weight estimator, the weighted estimator, and the adjusted weighted estimator across sample sizes $(n,m)$ and reassignment thresholds $\epsilon$. Irregular sampling substantially increases the error of the naive estimator, while the weighted and adjusted weighted estimators remain stable, with the adjusted weighted estimator achieving the lowest error in all the experiments.}
		\label{tab:realistic-mise}
\end{table}
Tables~\ref{tab:even-mise} and~\ref{tab:realistic-mise} report the square root of the mean integrated squared error (MISE) for estimating the time-use proportions $\widehat{T}_e$ across all entities $e \in \mathcal{E}$ under both timestamp designs. Several systematic patterns emerge from these results.

Across all estimators, the error magnitude decreases with the number of observed days $n$, reflecting the expected $\mathcal{O}(n^{-1/2})$ behavior of averaged estimators. In contrast, changes in within-day sampling density $m$ have only a modest effect once $m$ exceeds a moderate level, indicating that day-to-day variability is the main source of estimation uncertainty.

In the even–spacing design, both the naive and weighted estimators exhibit similar performance. This similarity is expected because uniformly spaced timestamps allow equal weights to approximate the time-allocation weights $W_{i,j}$. The adjusted weighted estimator achieves the smallest errors in this context, although its advantage over the weighted estimator is limited. Since evenly spaced sampling leads to stable entity assignments and infrequent boundary oscillations, the adjustment step has few opportunities to change the labels; therefore, the improvement is moderate.

Substantial differences emerge under the realistic (non-uniform) timestamp design. The naive estimator has errors an order of magnitude larger than other methods. This deterioration highlights the sensitivity of equal–weight approaches to irregular sampling, where long gaps between observations can distort the estimated time proportions. The weighted estimator significantly reduces this bias by accounting for the underlying sampling intervals, yielding errors comparable to those under even spacing.

In the realistic design, the adjusted weighted estimator consistently achieves the lowest errors across all sample sizes $(n,m)$ and threshold values~$\epsilon$. The improvement over the weighted estimator is most pronounced when $n$ is small and threshold $\epsilon$ is larger. Reassigning borderline points aids cases with noisy entity labels and limited days. As the number of days increases, the performance of the weighted and adjusted weighted estimators becomes nearly indistinguishable. This indicates that the adjustment primarily reduces variability from random misclassifications at the boundary rather than systematic bias.

Comparing the two thresholds, $\epsilon = 0.1$ and $\epsilon = 0.05$, smaller values yield fewer errors for both weighted and adjusted estimators. This pattern shows that moderate reassignment levels are sufficient to correct most label fluctuations caused by boundaries, while overly permissive adjustments may lead to unnecessary smoothing. These results demonstrate that time-based weighting is essential for accurately estimating daily time-use proportions with irregular sampling. The boundary-adjustment procedure improves situations where GPS noise or map–matching artifacts create unstable entity assignments. The adjusted estimator is especially useful in finite samples and realistic data-collection scenarios.

\subsection{Simulation experiments for the map-SMM}

We illustrate the performance of the map-augmented simple movement
model (map-SMM) employing the simulation setting in Section \ref{sec:timestamp-generation} with $n=90$ days and $m=1439$ observations per day. This corresponds to about one GPS observation per minute. Timestamps are generated using a realistic (non-uniform) scheme, as described in Section \ref{sec:timestamp-generation}, ensuring the temporal sampling pattern mimics natural smartphone-based GPS logging rather than idealized even spacing. The choice of $n=90$ reflects roughly three months of daily activity, providing a sufficient observation window to stabilize the estimated time–use distribution and reveal persistent differences between weekday and weekend behavior encoded in the activity-pattern model in Table~\ref{Table:pattern_detail}. These daily trajectories serve as input for the clustering analysis described  below.

\subsubsection{True activity-time distribution in the simulation setup}

Figure \ref{fig:1} illustrates the spatial setting of the simulation study and the true distribution of time across entities as indicated by the activity-pattern model in Table \ref{Table:pattern_detail}.
The left panel of Figure \ref{fig:1} displays a spatial map with six polygons representing key locations—home, restaurant, office, supermarket, park, and beach—along with adjacent road segments. These entities create the environment where the fictitious individual follows daily activity patterns. The middle panel of Figure \ref{fig:1} displays the proportion of time the fictitious individual spends in each polygon, as calculated from the generative model. These proportions reflect the five activity patterns. Home accounts for the largest share of time because the individual spends more time there during early mornings, evenings, and weekends. Home is the only entity involved in Pattern 5. The office is the second most visited location due to long weekday stays. The park, restaurant, supermarket, and beach receive smaller proportions, consistent with their roles as optional or infrequent destinations in the activity model.

The right panel of Figure \ref{fig:1} shows the true time proportions for each road segment. The map-SMM model treats movement between polygons as instantaneous; however, the simulated GPS data records discrete transitions. Road segments primarily capture travel episodes. Thus, the true time spent on roads is generally small but not uniform. Segments connecting frequently paired polygons—especially the home–office route—exhibit the highest proportions, while segments leading to rarely visited leisure destinations have minimal values. These panels summarize the spatial environment and the temporal structure underlying subsequent simulation experiments.

\begin{figure}[]
	\centering
	\includegraphics[width=0.32\linewidth]{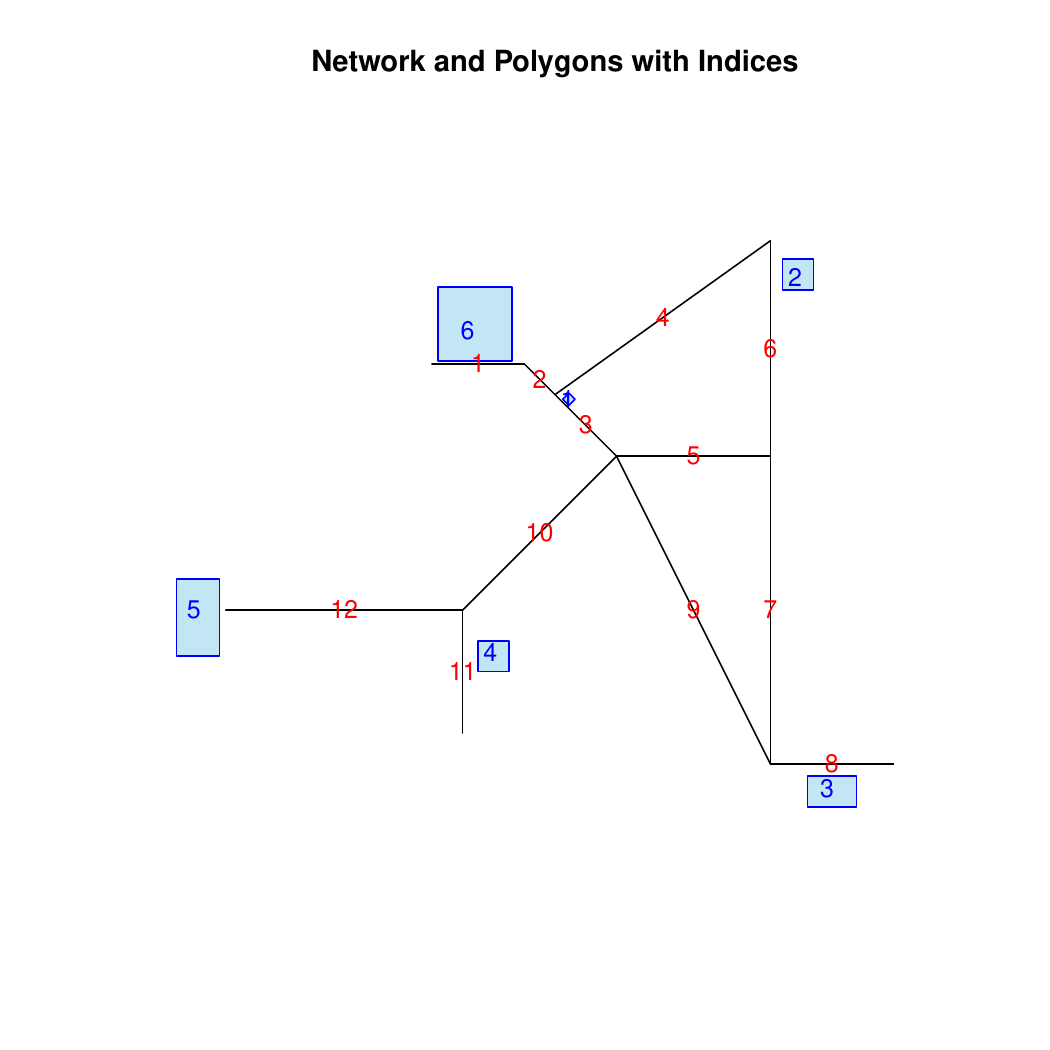}
	\includegraphics[width=0.32\linewidth]{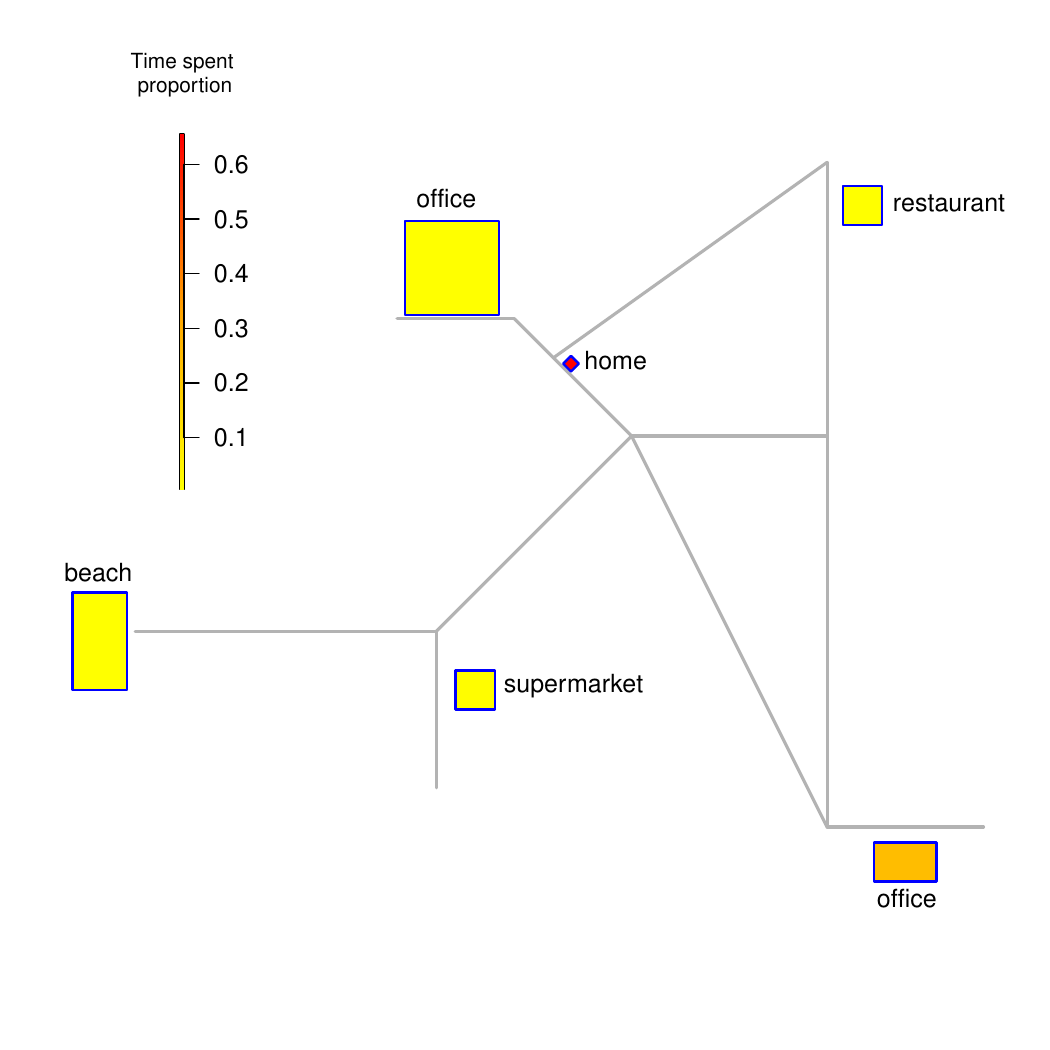}
	\includegraphics[width=0.32\linewidth]{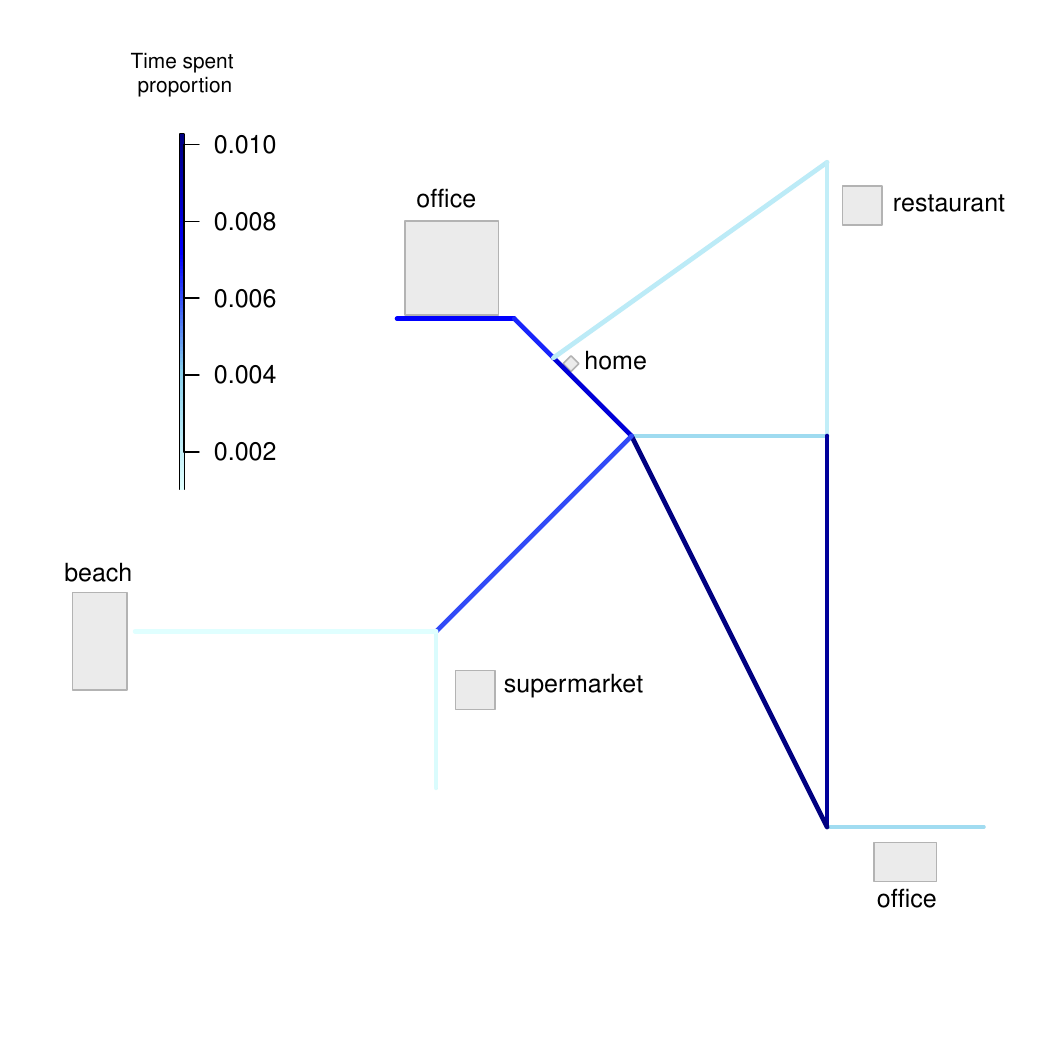}
	\caption{
		Spatial environment and true activity-time distribution used in the simulation study. Left panel: Map of the simulation layout, which consists of six polygons (home, restaurant,
		office, supermarket, park, and beach) and the connecting road network. Middle panel: True proportion of time spent in each polygon, computed from the activity-pattern model in Table~\ref{Table:pattern_detail}. The distribution reflects the 
		dominant role of home and office, as well as the infrequent visits to other locations. Right panel: True proportion of time spent on each road segment, capturing the 
		travel structure implied by the model, with higher values on frequently traversed routes such as the home--office connection. These panels summarize the spatial and temporal structure used to generate the simulated GPS data.}
	\label{fig:1}
\end{figure}

\subsubsection{Time spent proportion estimation}

The estimated overall time–use proportions for polygons and road segments are shown in panels (a) and (b) of Figure~\ref{fig:2}. These estimates summarize the individual’s long-run time distribution across all entities and provide the empirical counterpart to the theoretical quantity $\{T_e\}_{e \in \mathcal{E}}$ defined in Eq.~\eqref{eq:meantimespententity}. The resulting spatial pattern closely reflects the activity-pattern model structure in Table~\ref{Table:pattern_detail}. The home polygon accounts for the largest share of time, consistent with extended morning, evening, and weekend durations in both weekday and weekend patterns, as well as the full-day stay in Pattern~5. The office polygon receives the second largest share, driven by the fictitious individual's long, continuous weekday work periods. By contrast, the restaurant, supermarket, park, and beach contribute relatively small proportions, reflecting their roles as occasional destinations visited only in a minority of daily patterns.

The estimated time-use proportions for the road segments in panel (b) of Figure~\ref{fig:2} are much smaller, as travel between polygons occupies a limited fraction of the simulated day. Nevertheless, the proportions are not spatially uniform. Segments connecting frequently paired polygons—most notably the home–office path—exhibit the highest values, while roads to infrequent leisure locations, such as the supermarket and beach, show near-zero proportions. These patterns illustrate the proposed weighted estimator's ability, with assignment adjustment, to produce coherent representations of both high-frequency and rare transitional behaviors.

Panels (c) and (d) of Figure~\ref{fig:2} present the PN activity spaces based on estimated time-use proportions. The PN activity space includes all entities with positive time use, and the level-$\gamma$ activity space $\mathcal{AS}_\gamma$ is the smallest subset of entities that collectively accumulate at least $100\gamma\%$ of the individual’s total time. For polygons, the level-$0.50$, $0.70$, and $0.90$ activity spaces consist exclusively of home and office, indicating these locations represent the majority of the individual’s daily routine. As $\gamma$ increases to $0.95$ and $0.99$, the activity space expands to include the park, supermarket, restaurant, and beach, reflecting small but significant contributions from occasional leisure or errand trips. This nested structure highlights the core–periphery nature of the simulated mobility process.

The road-segment activity spaces exhibit a similar hierarchical pattern. For lower coverage levels ($\gamma = 0.50$ and $0.70$), the activity space mainly consists of the most frequently traversed commuting path. As the threshold rises to $0.90$ and $0.95$, segments associated with commercial spaces or restaurant transitions appear, and at $\gamma = 0.99$, the space includes seldom-used routes to the supermarket and beach. This progression distinguishes routine travel routes from infrequent paths linked to weekend or leisure activities.

The estimated proportions of time-use and resultant PN activity spaces demonstrate that the proposed estimator effectively captures both dominant and peripheral elements of the simulated mobility pattern, providing a coherent and interpretable representation of individual space–time behavior.

\begin{figure}[]
	\centering
	\includegraphics[width=0.48\linewidth]{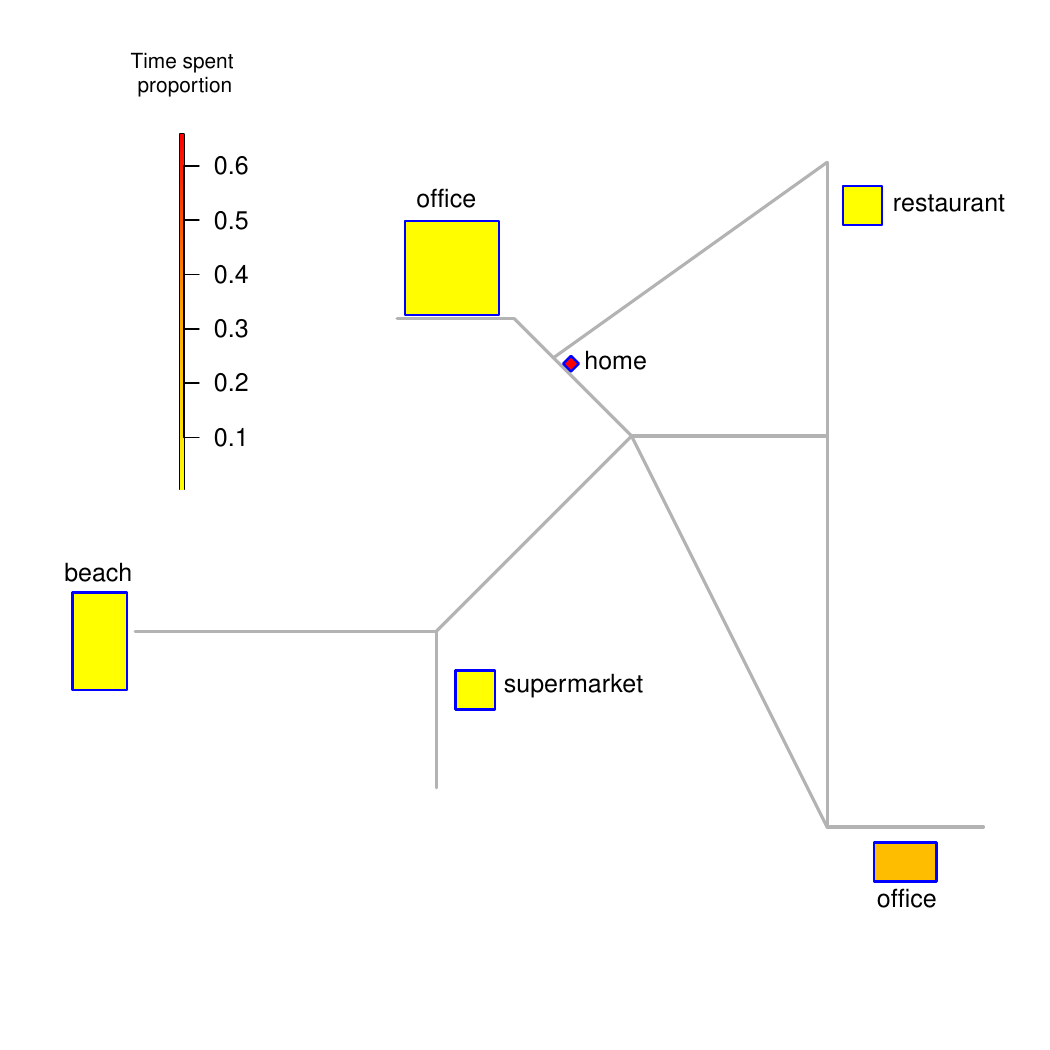}
	\includegraphics[width=0.48\linewidth]{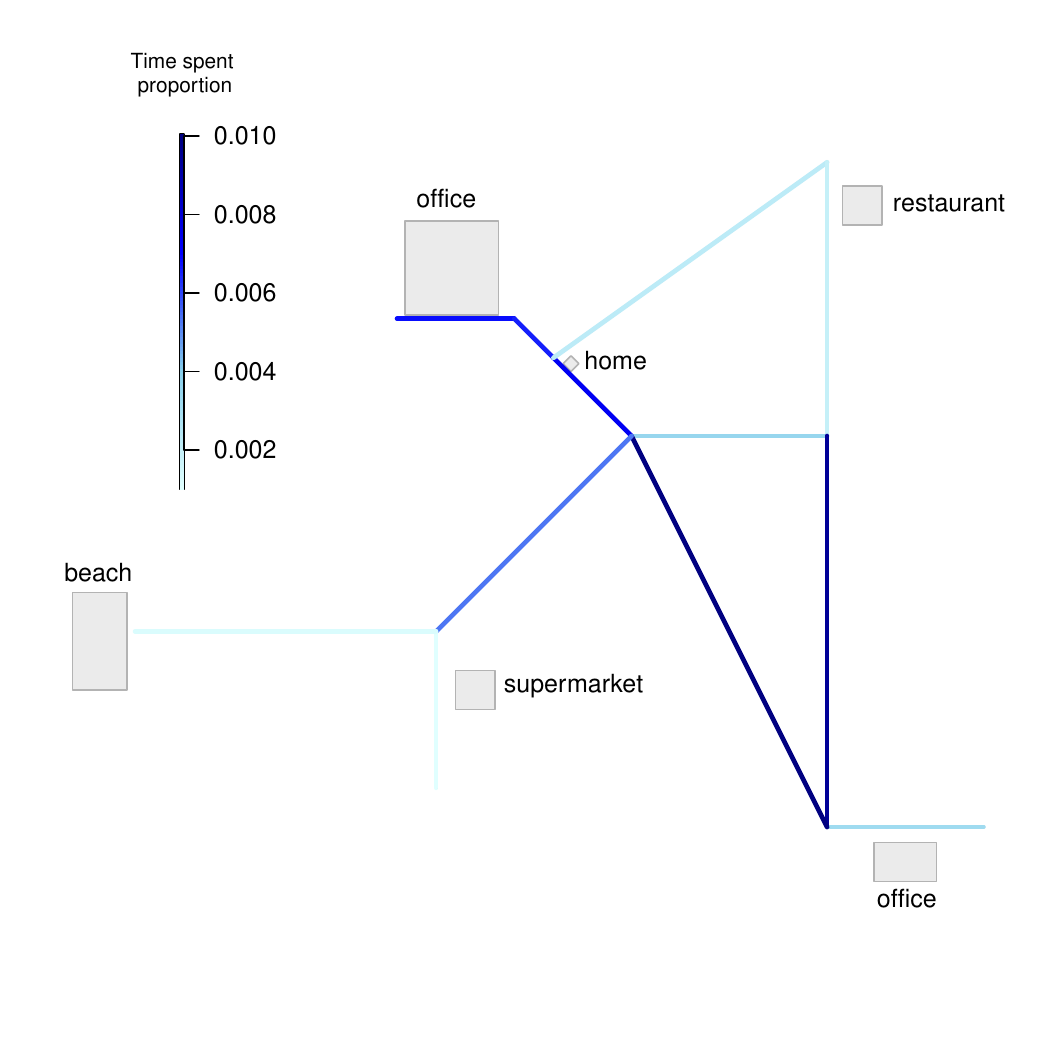}
	\includegraphics[width=0.48\linewidth]{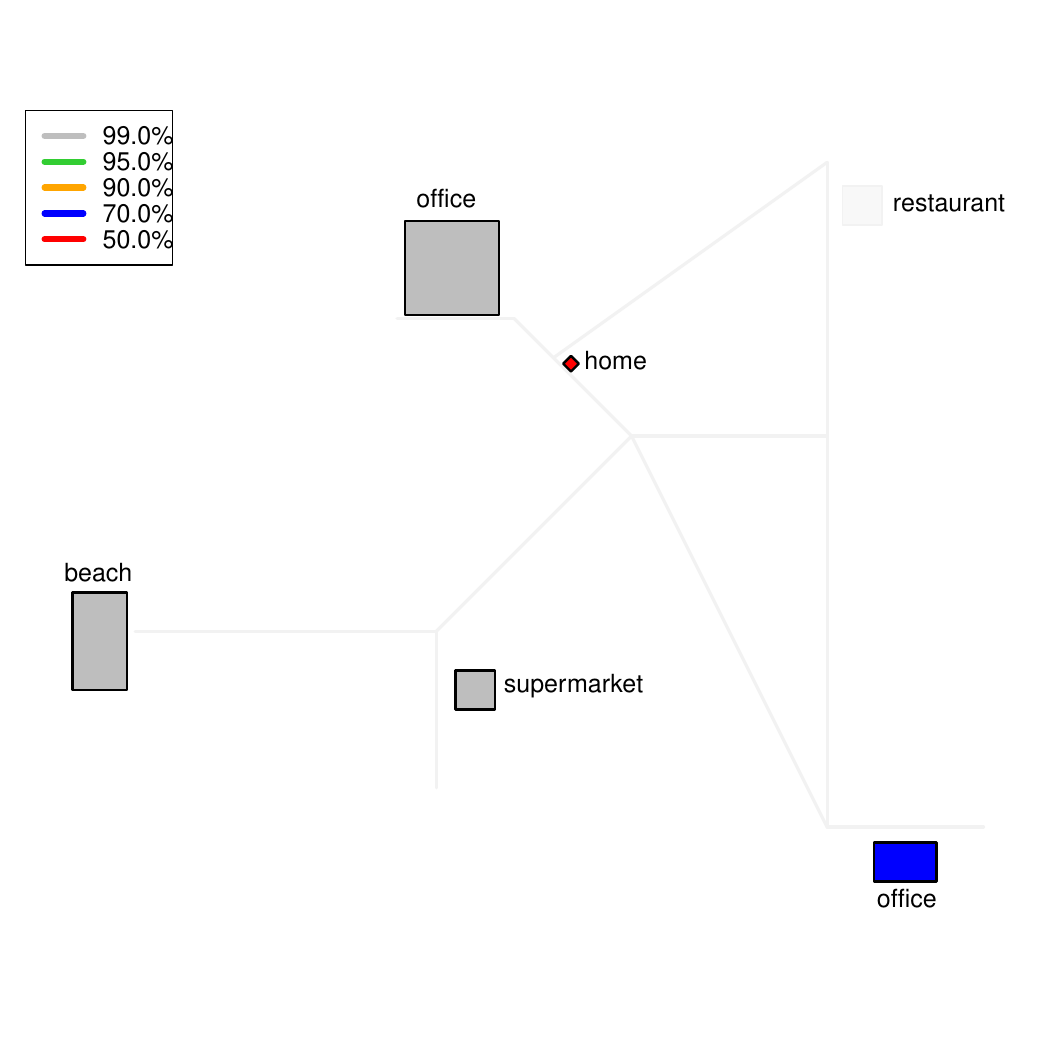}
	\includegraphics[width=0.48\linewidth]{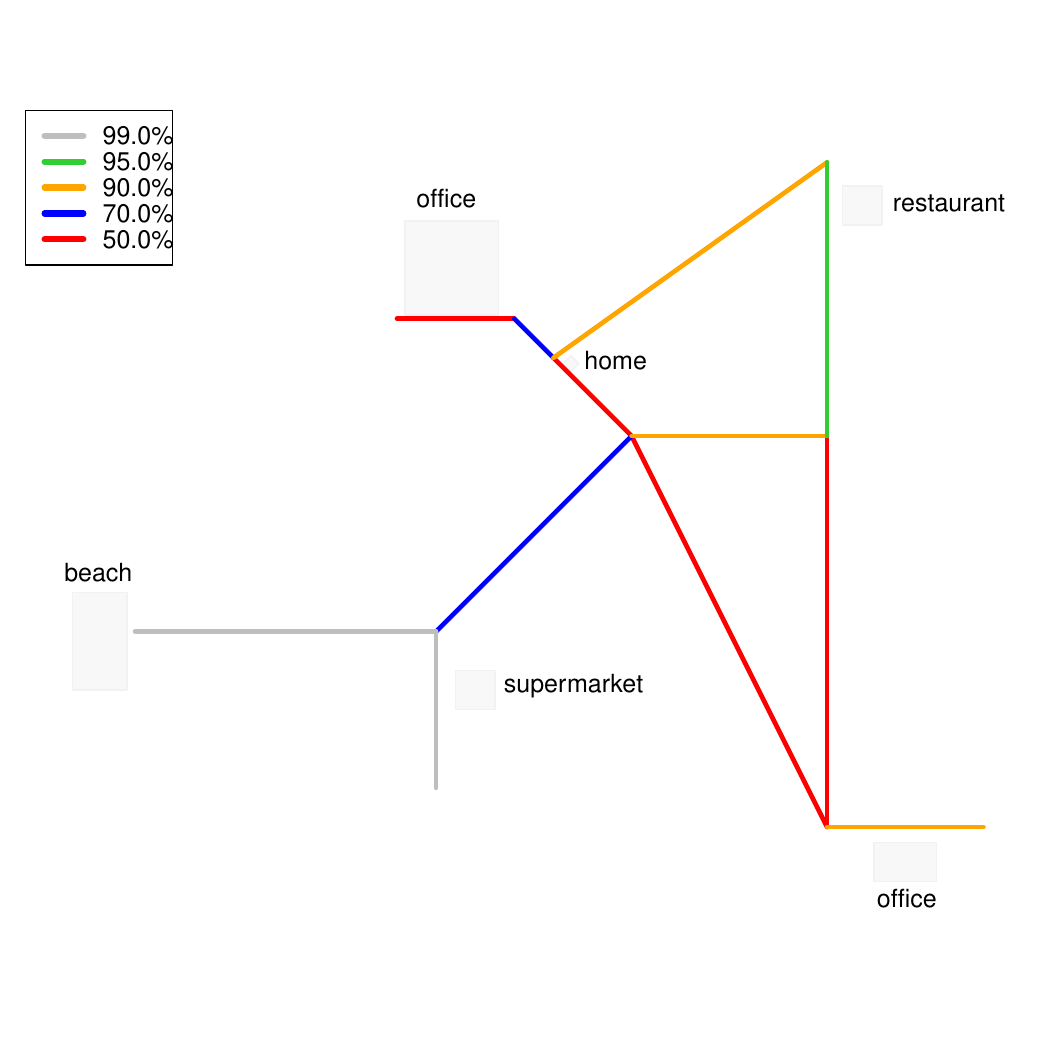}
	\caption{
		Overall time–use proportions and PN activity spaces in the simulation study. Top left panel: Estimated long-run proportion of time spent in each polygon. Top right panel: Estimated long-run proportion of time spent on each road segment. Bottom left panel: Level-$\gamma$ PN activity spaces for polygons, showing the smallest sets of polygons that together account for $\gamma \in \{0.50,0.70,0.90,0.95,0.99\}$
		of total time. Bottom right panel: Corresponding level-$\gamma$ PN activity spaces for road segments. The results highlight a core activity region dominated by home and office, with progressively larger spaces capturing occasional visits to additional polygons and rarely used travel paths.
	}
	\label{fig:2}
\end{figure}

\subsubsection{Clustering results}

Figure~\ref{fig:3} displays clustering results from the time weighted edit distance defined in Section~\ref{subsec:clustering}. The affinity matrix (left panel in Figure~\ref{fig:3}) presents the $90\times 90$ distance matrix reordered by hierarchical clustering. Blocks of low distances (shown in darker colors) appear along the diagonal, indicating sets of days with similar activity sequences. In contrast, off-diagonal regions with consistently high distances (shown in brighter colors) reveal groups of days that differ markedly. This block structure reflects the simulation design: days generated from the same activity pattern (e.g., Pattern~1 or Pattern~3) exhibit long contiguous sequences of visits to the same polygons, resulting in small edit distances. Conversely, days from distinct patterns require substantial insertions, deletions, or substitutions of high-dwell-time entities, leading to larger distances. The overall structure in the left panel of Figure~\ref{fig:3} displays five clear blocks corresponding almost exactly to the five daily activity patterns in Table~\ref{Table:pattern_detail}.

The dendrogram in the right panel of Figure~\ref{fig:3} presents a complementary view of hierarchical relationships among daily 
trajectories. Using single-linkage agglomeration, days with the same pattern merge at low heights, forming tight clusters. Merge heights between different clusters are significantly larger, indicating that altering a day's pattern requires reallocating substantial dwell time across multiple entities. This aligns with the time–weighted edit distance structure, where substitutions between long-duration home or office blocks are particularly costly. The dendrogram reveals five primary branches, corresponding to the number of activity patterns in the simulation. It shows that weekday patterns (Patterns~1 and~2) are more similar to each other than to weekend patterns (Patterns~3 and~4), as reflected in their merge heights. Pattern~5, representing a full day at home, forms a well-isolated cluster and merges at the greatest distance, consistent with its distinct temporal signature. The affinity matrix and dendrogram demonstrate that the proposed time–weighted edit distance captures structural differences in daily activity patterns, while clustering reveals the latent pattern structure in the simulation design.

\begin{figure}[]
	\centering
	\includegraphics[width=0.45\linewidth]{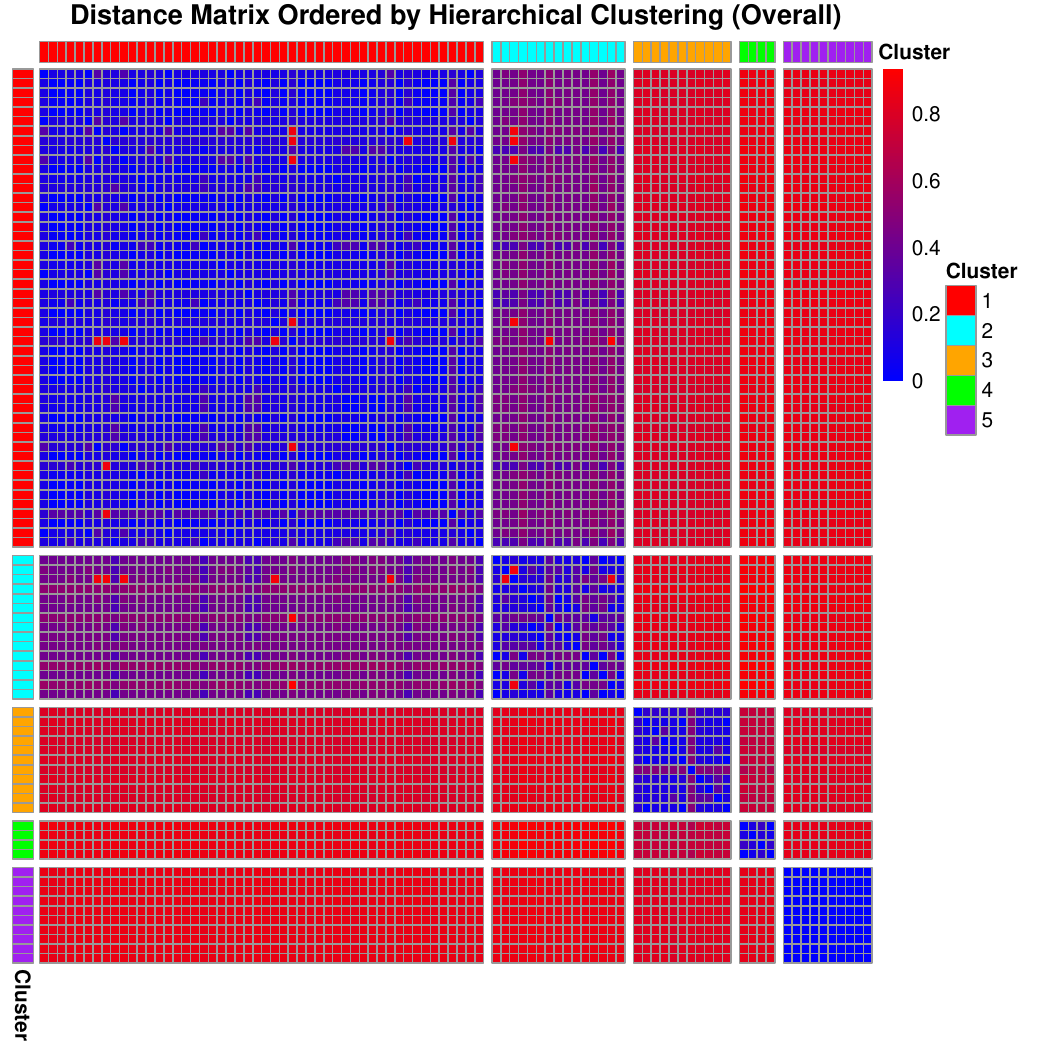}
\includegraphics[width=0.45\linewidth]{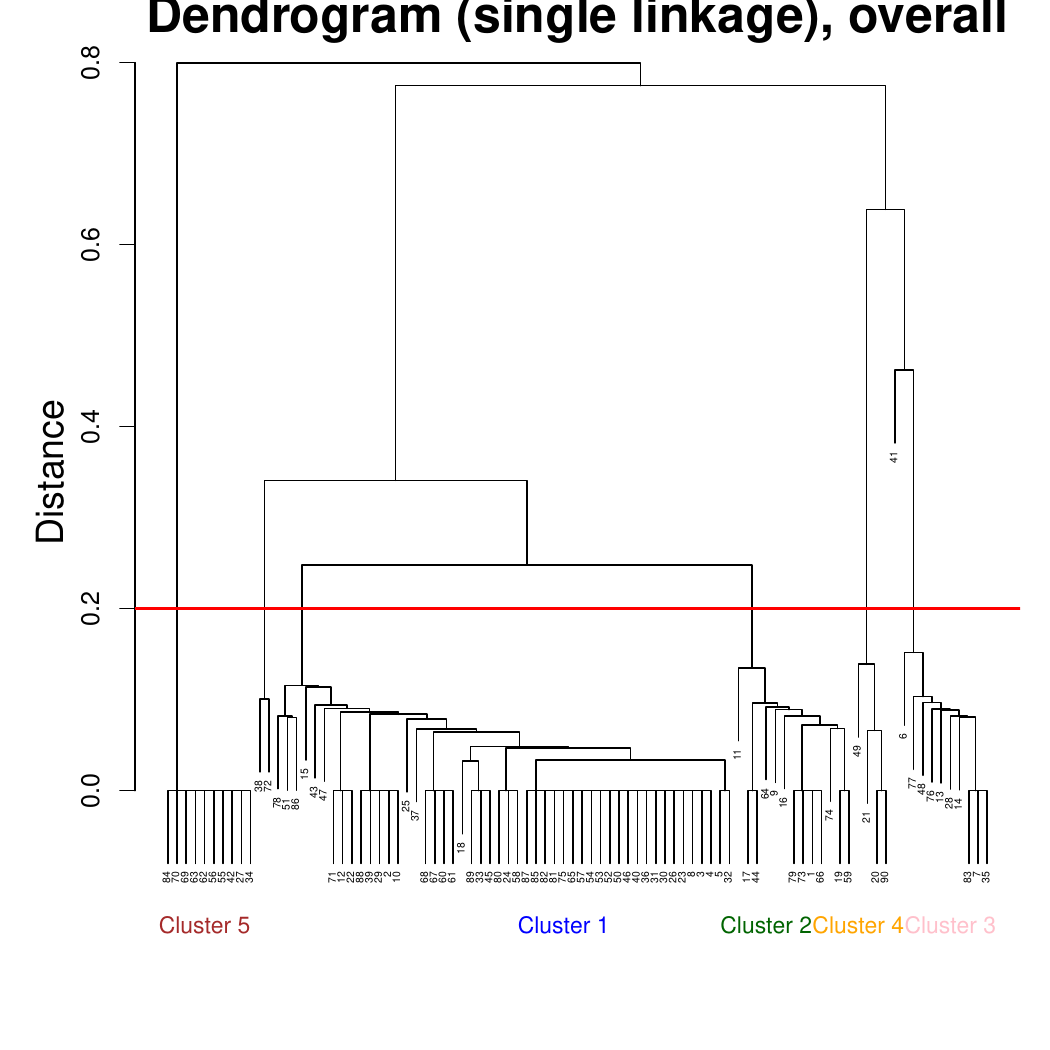}
\caption{Clustering of the 90 simulated daily trajectories based on the time--weighted edit distance. Left panel: Affinity (distance) matrix ordered by hierarchical clustering. Distinct low-distance blocks along the diagonal indicate groups of days with highly similar activity sequences, corresponding closely to the five underlying daily activity patterns in the simulation design. Right panel: Single-linkage dendrogram showing the hierarchical relationships among days. Days generated from the same activity pattern merge at small distances, while patterns with distinct temporal structures merge only at larger heights. Together, the two panels confirm that the clustering procedure successfully recovers the latent pattern structure encoded in the simulation model.}
	\label{fig:3}
\end{figure}

\subsubsection{Temporal stability of the estimated activity spaces}\label{subsec: appendix-stablity}

The left panel of Figure~\ref{fig:stability} displays daily activity-pattern sequences under two scenarios: an alternating-pattern regime and a weekend–weekday regime. These regimes differ only in the ordering of the same set of simulated daily patterns. In the alternating regime, five patterns are interwoven throughout the observation window, while in the weekend–weekday regime, weekend-like behavior (Patterns~3, 4, and 5) occurs before weekday-like patterns (Patterns~1 and 2). Since the patterns are identical in both cases, the overall activity-space distribution, obtained by aggregating all 90 days, is the same for both regimes. What differs is the trajectory by which the cumulative activity space approaches this overall activity space; it is this ordering effect that the LCT curves are designed to detect.

The middle panel of Figure~\ref{fig:stability} shows polygon-based LCT values across coverage levels $c$. The LCT curve structure reflects the ordered inclusion of major and minor polygons. The first major rise occurs when $c$ exceeds the time spent at the home polygon, after which the office polygon enters the activity space. In the alternating regime, early-day averages include contributions from both home and office patterns, allowing the cumulative distribution to quickly approach its full-period form, yielding relatively small LCT values. In contrast, the weekend–weekday regime consists almost entirely of weekend-type patterns with little or no office time in the early days. Consequently, the cumulative proportion spent at the office remains far from the full-period proportion for many days, resulting in a much larger LCT around this transition point. A second rise near $c \approx 0.95$ corresponds to the addition of low-frequency polygons (e.g., park and other discretionary sites). These polygons appear sporadically across the 90 days, requiring a long accumulation period under both regimes for stabilization. However, the alternating regime mixes rare visit days throughout the sample, converging slightly faster than the weekend–weekday regime, where rare visits are delayed until the latter part of the time series.

The right panel of Figure~\ref{fig:stability} presents the analogous LCT values for road segments. In both regimes, early stabilization of commute-related segments is visible at moderate $c$ levels. However, patterns diverge at higher levels: in the alternating regime, infrequent road segments (linked to sporadic travel to peripheral destinations) accumulate gradually over the period, while in the weekend–weekday regime, these segments are delayed until the second half of the observation window. Consequently, the weekend–weekday regime yields larger LCT values at higher $c$ levels, reflecting the delayed inclusion of rare road segments and the slower approach of the cumulative distribution to its long-run counterpart.

The three panels of Figure~\ref{fig:stability} illustrate LCT sensitivity to the temporal order of daily activity patterns. Although the two regimes share the same activity-space structure, their differing orderings create notably different stabilization trajectories. The alternating regime disperses all patterns—including rare ones—throughout the period, allowing cumulative distributions to converge rapidly and resulting in low LCT values. The weekend–weekday regime prioritizes low-frequency weekend patterns and delays high-frequency office patterns, postponing the appearance of key polygons and road segments in the cumulative distribution, leading to considerably larger LCT values. This demonstrates how LCT captures both the long-term structure of mobility and the temporal dynamics through which that structure emerges from daily observations.

\begin{figure}[htbp]
	\centering
	\includegraphics[width=0.32\linewidth]{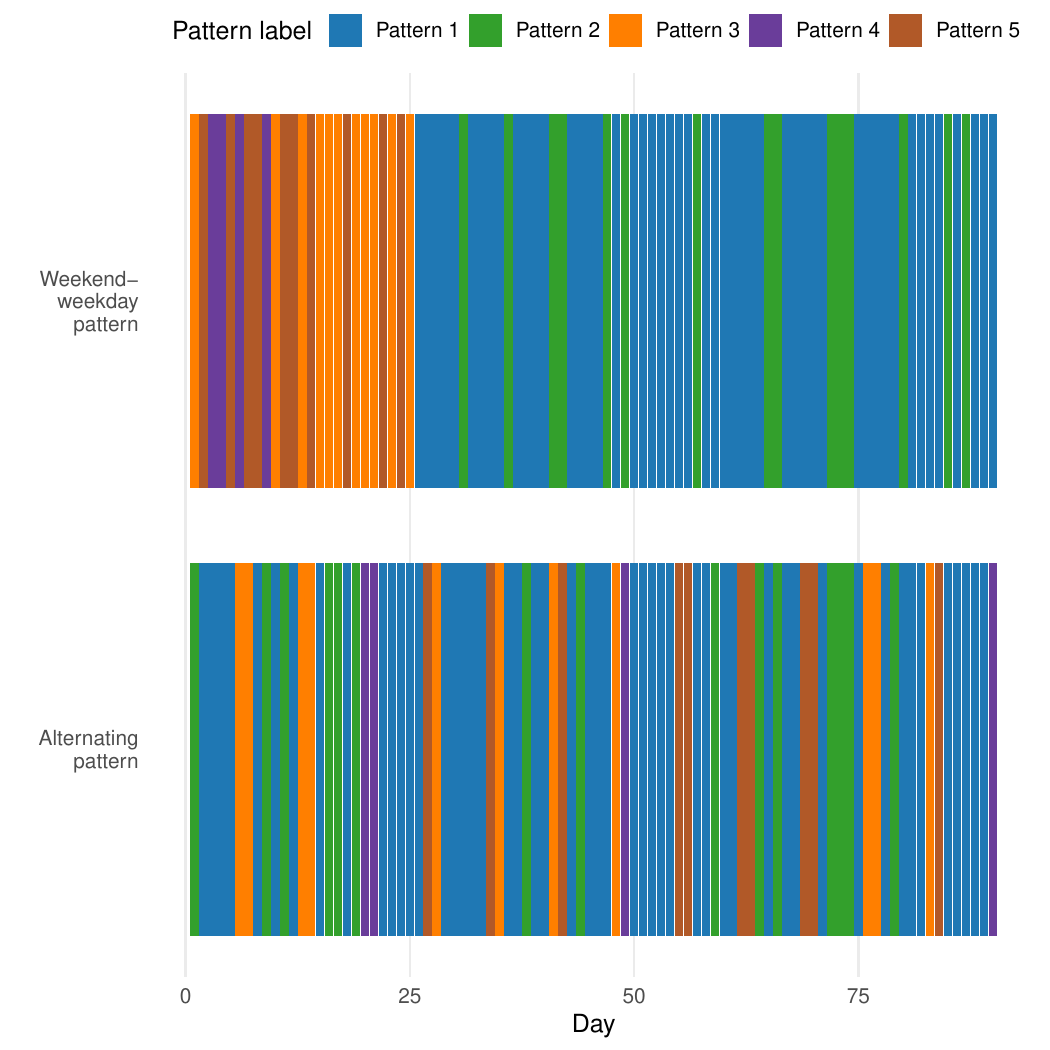}
	\includegraphics[width=0.32\linewidth]{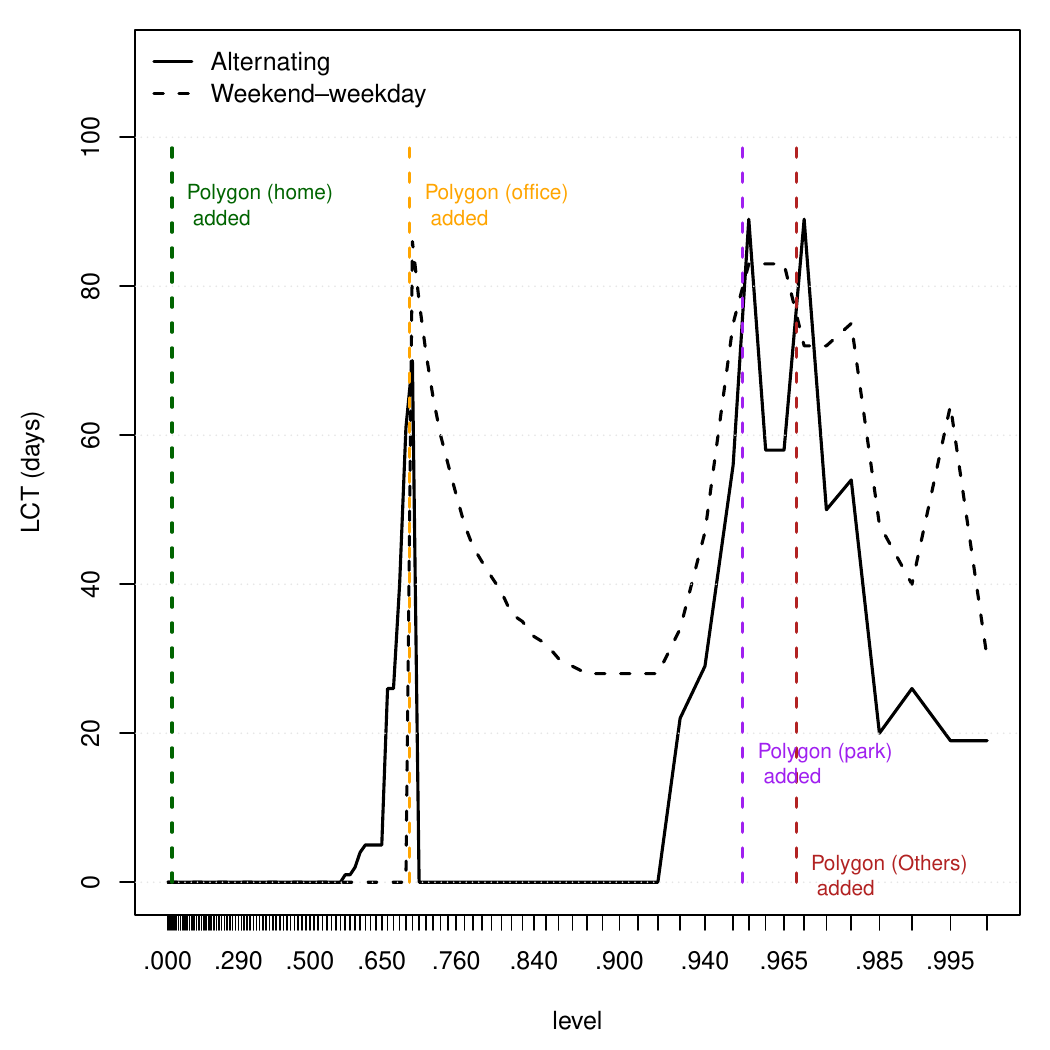}
	\includegraphics[width=0.32\linewidth]{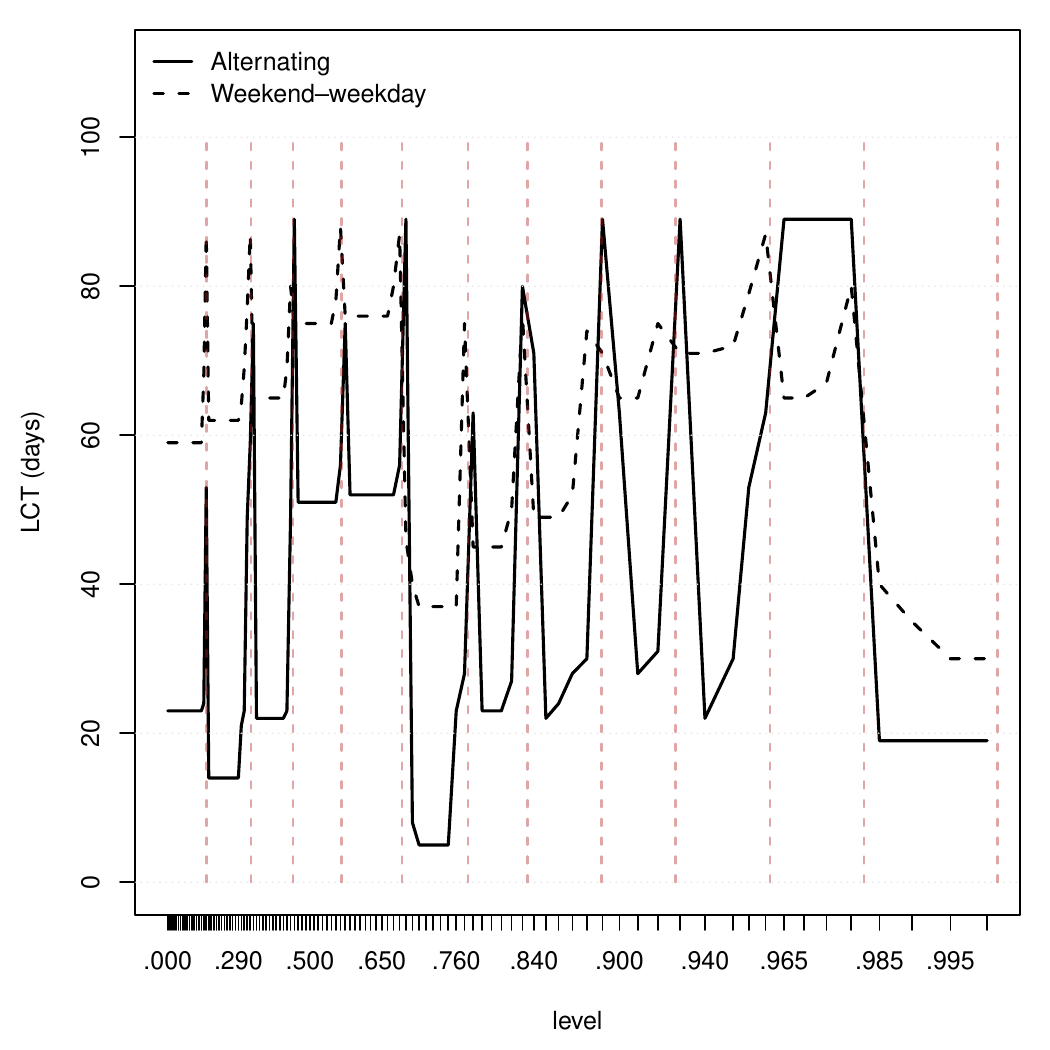}
	\caption{
		Stability analysis of the simulated mobility patterns over 90 days. Left panel: Daily activity-pattern labels, showing frequent alternation among the five patterns due to the weekday–weekend structure of the simulation design. Middle panel:  Last-crossing time (LCT) for polygon-based activity spaces across
		coverage levels $c$. The sharp increases at $c=0.90$ and $c=0.95$ correspond to the addition of the office and secondary polygons, respectively. Right panel: LCT for road-segment activity spaces, showing a smoother progression and slower stabilization at high coverage levels due to the rarity of long-distance weekend travel routes.}
	\label{fig:stability}
\end{figure}

\section{Proof of Theorem \ref{thm:1}}
\begin{theorem}{thm:1}
Under Assumption \ref{ass:visits}, for each entity $e\in\mathcal{E}$ the asymptotic absolute error \( |\widehat{T}_{e}-T_e|\) between the mean proportion of time spent by the individual in $e$ and its estimator \eqref{esti:Tk} is
	\begin{align*}
		O_p\left(\frac{1}{n}\sqrt{\sum_{i=1}^{n} \sum_{j=1}^{m_i} W_{i,j}^2 \pi_{i,j}}\right)+O\left(\frac{1}{n} \sum_{i=1}^n\sum_{j=1}^{m_i}W_{i,j} \pi_{i,j}\right)\\
        +O\left(\frac{1}{n}\sum_{i=1}^n \sum_{j:(i,j)\in  \mathcal{I}_e} t_{i,j}-t_{i,j-1} \right)+
		O_p\left(\sqrt{\frac{1}{n}}\right),
	\end{align*}
	as $\max_{(i,j)}W_{i,j}\rightarrow 0$ and $n\rightarrow \infty$.
\end{theorem}
\begin{proof}
For any $e\in \mathcal{E}$, we bound the error of the estimated time spent proportion $|\hat T_{e}-T_{e}|$ by decomposing it into four differences. Recall $W_{i,j} = (t_{i,j+1}-t_{i,j-1})/2$ and write the Voronoi cell of $e$ by $r_e$:
\begin{align*}
	r_e = \{x\in W: d(x,e)\leq d(x,e') \text{ for any }e' \in \mathcal{E}\setminus\{e\}\},
\end{align*}
the set of points in the spatial window $\mathcal{W}$ with the closest entity $e$. In the shortest-distance method, $X_{i,j}\in r_e$ is equivalent to $X_{i,j}$ being assigned to $e$., i.e. the closest entity of $X_{i,j}$ is $e$. Then, 
\begin{align*}
	\widehat{T}_{e} =& \frac{1}{n}\sum_{i=1}^{n}\sum_{j=1}^{m_i}
	W_{i,j}\mathbf{1}\left(
	d_E(X_{i,j},e)\le d_E(X_{i,j},e')\text{ for any } e' \in \mathcal{E}\setminus\{e\}
	\right)\\
    &=\frac{1}{n}\sum_{i=1}^n\sum_{j:X_{i,j}\in r_e} W_{i,j}. \end{align*}
For the $i-$th day, we denote the actual proportion of time spent in $e$ as $T_{e,i}$. Hence, by definition of $T_e$, $\mathbb{E}(T_{e,i}) = T_e$.  Using these notations, we could further decompose the error $|\hat T_{e}-T_e|$ into the following four terms:
\begin{align}
    &|\hat T_{e}-T_e| \nonumber\\
    = & \left|\frac{1}{n}\sum_{i=1}^n\sum_{j:X_{i,j}\in r_e} W_{i,j} - T_e\right| \nonumber \\
    \leq & \underbrace{\left|\frac{1}{n}\sum_{i=1}^n\sum_{j:X_{i,j}\in r_e} W_{i,j} - \mathbb{E}\left(\frac{1}{n}\sum_{i=1}^n\sum_{j:X_{i,j}\in r_e} W_{i,j}\right)\right|}_{\text{Term (I)}} +\nonumber \\ & \underbrace{\left|\mathbb{E} \left(\frac{1}{n}\sum_{i=1}^n\sum_{j:X_{i,j}\in r_e} W_{i,j}\right)-\frac{1}{n}\sum_{i=1}^n\sum_{j:X^{\mathrm{true}}_{i,j}\in e} W_{i,j}\right|}_{\text{Term (II)}} + \nonumber\\
    &\underbrace{\left|\frac{1}{n}\sum_{i=1}^n\sum_{j:X^{\mathrm{true}}_{i,j}\in e} W_{i,j} - \frac{1}{n}\sum_{i=1}^n T_{e,i}\right|}_{\text{Term (III)}}+\underbrace{\left|\frac{1}{n}\sum_{i=1}^n T_{e,i}-T_e\right|}_{\text{Term (IV)}} \label{eq:thm1-pf-1}
\end{align}
In Eq.~\eqref{eq:thm1-pf-1}, $\frac{1}{n}\sum_{i=1}^n\sum_{j:X^{\mathrm{true}}_{i,j}\in e} W_{i,j} $ represents the estimated time spent proportion with no measurement error, meaning each timestamp $t_{i,j}$ is correctly classified to its entity. 

In the decomposition from Eq.~\eqref{eq:thm1-pf-1}, term (I) refers to the variance of $\hat T_{e}$; term (II) is the difference between the expected and actual time spent based on the true location of each observation in the entity $e$; and term (III) is the error in the estimated time proportion for each observation's true location. 

\textbf{Term (I)}: By independence among days, the variance of the first term on the right side of Eq.~\eqref{eq:thm1-pf-1} is 
\begin{align}
	\text{Var}\left(\frac{1}{n}\sum_{i=1}^{n}\sum_{j:X_{i,j}\in r_e} W_{i,j}\right)=\frac{1}{n^2}\text{Var}\left(\sum_{i=1}^{n}\sum_{j=1} ^{m_i}W_{i,j}1(X_{i,j}\in r_e)\right).\label{eq:thm1-Var-term-I}
\end{align}
Note $X_{i,j} = X^{\mathrm{true}}_{i,j}+\epsilon_{i,j}$. From the independence of $\{\epsilon_{1,1},\ldots,\epsilon_{n,m_n}\}$, we obtain
\begin{eqnarray}
\text{Var}\left(\frac{1}{n}\sum_{i=1}^{n}\sum_{j:X_{i,j}\in r_e} W_{i,j}\right) & = & \frac{1}{n^2} \text{Var}\left(\sum_{i=1}^{n}\sum_{j=1}^{m_i}W_{i,j}1(X^{\mathrm{true}}_{i,j}+\epsilon_i\in r_e)\right),\nonumber \\
& = & \frac{1}{n^2} \sum_{i=1}^{n}\sum_{j=1}^{m_i}W_{i,j}^2\text{Var}(1(X^{\mathrm{true}}_{i,j}+\epsilon_i\in r_e)).\label{eq:thm1-Var-term-I-de}    
\end{eqnarray}

 Since $\mathbb{E}(1(X_{i,j}\in r_e)) = P(X_{i,j}\in r_e)$, we have
\begin{align*}
 \text{Var}(1(X_{i,j}\in r_e))=P(X_{i,j}\in r_e)-P(X_{i,j}\in r_e)^2,
\end{align*}
where
\begin{align*}
	P(X_{i,j}\in r_e)-P(X_{i,j}\in r_e)^2= P(X_{i,j}\in r_e)(1-P(X_{i,j}\in r_e)) 
\leq \\
\min\{P(X_{i,j}\in r_e),P(X_{i,j}\notin r_e)\}.
\end{align*}
Recall that \(
\pi_{i,j}
\;=\;
\Pr\bigl\{
\text{entity assigned to }X_{i,j}
\neq
\text{entity containing }X^{\mathrm{true}}_{i,j}
\bigr\}
\). We decompose the summation based on whether the true location is in $e_k$:
\begin{align}
	&\frac{1}{n^2} \sum_{i=1}^{n}\sum_{j=1}^{m_i}W_{i,j}^2\text{Var}(1(X^{\mathrm{true}}_{i,j}+\epsilon_i\in r_e)),\nonumber\\
	\leq &\frac{1}{n^2} \sum_{i=1}^{n} \sum_{i:X^{\mathrm{true}}_{i,j}\in e}W_{i,j}^2P(X_{i,j}\notin r_e)+ \frac{1}{n^2}\sum_{i=1}^{n} \sum_{i:X^{\mathrm{true}}_{i,j}\notin e}W_{i,j}^2P(X_{i,j}\in r_e),\nonumber\\
	\leq &\frac{1}{n^2} \sum_{i=1}^{n} \sum_{j=1}^{m_i} W_{i,j}^2 \pi_{i,j}.\label{eq:thm1-term-I-final}
\end{align}
\textbf{Term (II)}: For entity $e$, we have:
\begin{align*}
    &\mathbb{E} \left(\frac{1}{n}\sum_{i=1}^n\sum_{j:X_{i,j}\in r_e} W_{i,j}\right)=\frac{1}{n} \mathbb{E} \left(\sum_{i=1}^{n} \sum_{j=1}^{m_i} 1\{X_{i,j}\in r_e\} W_{i,j}\right)\\
    & = \frac{1}{n}\sum_{i=1}^{n} \sum_{j=1}^{m_i} W_{i,j}P\left( X_{i,j}\in r_e \right).
\end{align*}
We decompose the summation based on the true location being in entity $e_k$:
    \begin{align*}
   \frac{1}{n}\sum_{i=1}^{n}  \sum_{j=1}^{m_i} W_{i,j}P\left( X_{i,j}\in r_e \right)
    = & \frac{1}{n}\sum_{i=1}^{n} \sum_{j:X^{\mathrm{true}}_{i,j}\in e}  W_{i,j} P(X_{i,j}\in r_e)\\
    & +\frac{1}{n}\sum_{i=1}^{n} \sum_{j:X^{\mathrm{true}}_{i,j}\notin e}  W_{i,j}P(X_{i,j}\in r_e).
\end{align*}
Hence
\begin{align}
    & \mathbb{E} \left(\frac{1}{n}\sum_{i=1}^n\sum_{j:X_{i,j}\in r_e} W_{i,j}\right)-\frac{1}{n}\sum_{i=1}^n\sum_{j:X^{\mathrm{true}}_{i,j}\in e} W_{i,j}\nonumber\\
=& \frac{1}{n}\sum_{i=1}^{n} \sum_{j:X^{\mathrm{true}}_{i,j}\in e}  W_{i,j} P(X_{i,j}\in r_e)+\frac{1}{n}\sum_{i=1}^{n} \sum_{j:X^{\mathrm{true}}_{i,j}\notin e}  W_{i,j}P(X_{i,j}\in e)\nonumber \\
& - \frac{1}{n}\sum_{i=1}^{n} \sum_{j:X^{\mathrm{true}}_{i,j}\in e}  W_{i,j},\nonumber\\
=& \frac{1}{n}\sum_{i=1}^{n} \sum_{j:X^{\mathrm{true}}_{i,j}\notin e}  W_{i,j}P(X_{i,j}\in r_e)-\frac{1}{n}\sum_{i=1}^{n} \sum_{j:X^{\mathrm{true}}_{i,j}\in e}  W_{i,j}P(X_{i,j}\notin r_e),\nonumber\\
=&\frac{1}{n} \sum_{i=1}^n\sum_{j=1}^{m_i}W_{i,j} \pi_{i,j}.\label{eq:term2 in  thm1}
\end{align}

\textbf{Term (III)}: We seek a bound for $\frac{1}{n}\sum_{i=1}^n\sum_{j:X^{\mathrm{true}}_{i,j}\in e}W_{i,j}$ and $\frac{1}{n}\sum_{i=1}^n T_{e,i}$. For any entity $e$, on the $i-$th day, the corresponding index set satisfies the relationship: $\{j:X^{\mathrm{true}}_{i,j},X^{\mathrm{true}}_{i,j-1}\in e_k\}\subseteq \{j:X^{\mathrm{true}}_{i,j}\in e\}$. For convenience, we write $t_{i,0}=0, t_{i,m_i+1}=1$. Based on this relationship, we determine a lower bound for $\frac{1}{n}\sum_{i=1}^n\sum_{j:X^{\mathrm{true}}_{i,j}\in e}W_{i,j}$ as follows. 
\begin{align}
    &\frac{1}{n}\sum_{i=1}^n\sum_{j:X^{\mathrm{true}}_{i,j}\in e} W_{i,j} \nonumber\\
    = & \frac{1}{n}\sum_{i=1}^n\sum_{j=1}^{m_i} 1\{X^{\mathrm{true}}_{i,j}\in e\}\left(\frac{t_{i,j+1}-t_{i,j}}{2}+\frac{t_{i,j}-t_{i,j-1}}{2}\right) + \frac{1}{n}\sum_{i=1}^n 1\{X^{\mathrm{true}}_{i,1}\in e\} \frac{t_{i,1}}{2}\nonumber\\
    &+\frac{1}{n}\sum_{i=1}^n 1\{X^{\mathrm{true}}_{i,n}\in e\} \frac{1-t_{i,n}}{2} \nonumber\\
    \geq & \frac{1}{n}\sum_{i=1}^n\sum_{j=1}^{m_i-1} 1\{X^{\mathrm{true}}_{i,j}\in e,X^{\mathrm{true}}_{i,j+1}\in e\}\frac{t_{i,j+1}-t_{i,j}}{2}\nonumber \\
    & +\frac{1}{n}\sum_{i=1}^n\sum_{j=2}^{m_i} 1\{X^{\mathrm{true}}_{i,j-1}\in e, X^{\mathrm{true}}_{i,j}\in e\}\frac{t_{i,j}-t_{i,j-1}}{2}+\nonumber\\
    & \frac{1}{n}\sum_{i=1}^n 1\{X^{\mathrm{true}}_{i,0}\in e, X^{\mathrm{true}}_{i,1}\in e\} {t_{i,1}}+\frac{1}{n}\sum_{i=1}^n 1\{X^{\mathrm{true}}_{i,m_i}\in e, X^{\mathrm{true}}_{i,m_i+1}\in e\} (1-t_{i,n}) \nonumber\\
    =& \frac{1}{n}\sum_{i=1}^n\sum_{j=1}^{m_i+1} 1\{X^{\mathrm{true}}_{i,j}\in e, X^{\mathrm{true}}_{i,j-1}\in e\}({t_{i,j}-t_{i,j-1}})\nonumber\\
     =& \frac{1}{n}\sum_{i=1}^n\sum_{\substack{j=1,...,m_i+1:\\X^{\mathrm{true}}_{i,j}\in e \text{ and}\\X^{\mathrm{true}}_{i,j+1}\in e}} ({t_{i,j}-t_{i,j-1}}).\label{eq:sum of ti, left}
\end{align}
Similarly, because $\{(i,j):X^{\mathrm{true}}_{i,j}\in e\}\subseteq \{(i,j):X^{\mathrm{true}}_{i,j} \text{ or }X^{\mathrm{true}}_{i,j-1}\in e\}$, we determine an upper bound of $\frac{1}{n}\sum_{i=1}^n\sum_{j:X^{\mathrm{true}}_{i,j}\in e} W_{i,j}$ as follows:
\begin{align}
    &\frac{1}{n}\sum_{i=1}^n\sum_{j:X^{\mathrm{true}}_{i,j}\in e} W_{i,j} \nonumber\\
    = & \frac{1}{n}\sum_{i=1}^n\sum_{j=1}^{m_i} 1\{X^{\mathrm{true}}_{i,j}\in e\}\left(\frac{t_{i,j+1}-t_{i,j}}{2}+\frac{t_{i,j}-t_{i,j-1}}{2}\right)\nonumber\\
    & + \frac{1}{n}\sum_{i=1}^n 1\{X^{\mathrm{true}}_{i,1}\in e\} \frac{t_{i,1}}{2}\nonumber\\
    & +\frac{1}{n}\sum_{i=1}^n 1\{X^{\mathrm{true}}_{i,n}\in e\} \frac{1-t_{i,n}}{2} \nonumber\\
    \leq & \frac{1}{n}\sum_{i=1}^n\sum_{j=1}^{m_i-1} 1\{X^{\mathrm{true}}_{i,j}\in e \text{ or }X^{\mathrm{true}}_{i,j+1}\in e\}\frac{t_{i,j+1}-t_{i,j}}{2}\nonumber \\
    & +\frac{1}{n}\sum_{i=1}^n\sum_{j=2}^{m_i} 1\{X^{\mathrm{true}}_{i,j-1}\in e \text{ or } X^{\mathrm{true}}_{i,j}\in e\}\frac{t_{i,j}-t_{i,j-1}}{2}\nonumber\\
    & +\frac{1}{n}\sum_{i=1}^n 1\{X^{\mathrm{true}}_{i,0}\in e \text{ or }X^{\mathrm{true}}_{i,1}\in e\} {t_{i,1}}\nonumber\\
    &+\frac{1}{n}\sum_{i=1}^n 1\{X^{\mathrm{true}}_{i,m_i}\in e \text{ or } X^{\mathrm{true}}_{i,m_i+1}\in e\} (1-t_{i,n}) \nonumber\\
    =& \frac{1}{n}\sum_{i=1}^n\sum_{j=1}^{m_i+1} 1\{X^{\mathrm{true}}_{i,j}\in e \text{ or }X^{\mathrm{true}}_{i,j-1}\in e\}({t_{i,j}-t_{i,j-1}})\nonumber\\
     =& \frac{1}{n}\sum_{i=1}^n\sum_{\substack{j=1,...,m_i+1:\\X^{\mathrm{true}}_{i,j}\in e \text{ or}\\X^{\mathrm{true}}_{i,j+1}\in e}} ({t_{i,j}-t_{i,j-1}}).\label{eq:sum of ti, right}
\end{align}
Combining Eq.~\eqref{eq:sum of ti, left} and Eq.~\eqref{eq:sum of ti, right} yields
\begin{align}
    & \frac{1}{n}\sum_{i=1}^n\sum_{\substack{j=1,...,m_i+1:\\X^{\mathrm{true}}_{i,j}\in e \text{ and}\\X^{\mathrm{true}}_{i,j+1}\in e}} ({t_{i,j}-t_{i,j-1}})\nonumber\\
    & \leq\frac{1}{n}\sum_{i=1}^n\sum_{j:X^{\mathrm{true}}_{i,j}\in e} W_{i,j}\nonumber\\
    & \leq\frac{1}{n}\sum_{i=1}^n\sum_{\substack{j=1,...,m_i+1:\\X^{\mathrm{true}}_{i,j}\in e \text{ or}\\X^{\mathrm{true}}_{i,j+1}\in e}} ({t_{i,j}-t_{i,j-1}}).\label{eq:realtime-eq1}
\end{align}
Next, we find a bound for $\frac{1}{n}\sum_{i=1}^n T_{e,i}$. We use $V_{e,i}$ to denote the number of visits to $e$ on the $i-$th day, i.e., there are $V_{e,i}$ disjoint time intervals during which the individual is in entity $e$ on day $i$. We write $T_{e,i}=\sum_{b=1}^{V_{e,i}}T_{e,i,b}$, where $T_{e,i,b}$ is the duration of the $b-$th visit. For any $b\in\{1,\ldots,V_{e,i}\}$, during the $b-$th visit to $e$ on the $i-$day, WLOG, we use $t_{i,j_1}$ and $t_{i,j_2}$ to denote the first and last recorded timestamps of this visit. For any $j\in\{j_1,\ldots,j_2\}$, the person is in entity $e$ at $t_{i,j}$. At $t_{i,j_1-1}$ and $t_{i,j_2+1}$, the person is not in entity $e$. It follows that
\begin{align*}
   \sum_{j=j_1+1}^{j_2} t_{i,j}-t_{i,j-1} =  t_{i,j_2}-t_{i,j_1} \leq T_{e,i,b}\leq t_{i,j_2+1}-t_{i,j_1-1} =  \sum_{j=j_1}^{j_2+1} t_{i,j}-t_{i,j-1}. 
\end{align*}
Since $j\in \{j_1,\ldots,j_2\} \Leftrightarrow X^{\mathrm{true}}_{i,j}\in e$, we obtain 
\begin{align}
\sum_{\substack{j=1,...,m_i+1:\\X^{\mathrm{true}}_{i,j}\in e \text{ and}\\X^{\mathrm{true}}_{i,j+1}\in e}}  {t_{i,j}-t_{i,j-1}}\leq \sum_{b=1}^{V_{e,i}} T_{e,i,b} \leq \sum_{\substack{j=1,...,m_i+1:\\X^{\mathrm{true}}_{i,j}\in e \text{ or}\\X^{\mathrm{true}}_{i,j+1}\in e}} {t_{i,j}-t_{i,j-1}}.\label{eq:realtime-eq2-Tb}
\end{align}
Note that have $T_{e,i}=\sum_{b=1}^{V_{e,i}} T_{e,i,b}$, thus
\begin{align}
	\frac{1}{n}\sum_{i=1}^n\sum_{\substack{j=1,...,m_i+1:\\X^{\mathrm{true}}_{i,j}\in e \text{ and}\\X^{\mathrm{true}}_{i,j+1}\in e}}  {t_{i,j}-t_{i,j-1}}\leq \frac{1}{n}\sum_{i=1}^nT_{e,i} \leq \frac{1}{n}\sum_{i=1}^n\sum_{\substack{j=1,...,m_i+1:\\X^{\mathrm{true}}_{i,j}\in e \text{ or}\\X^{\mathrm{true}}_{i,j+1}\in e}}  {t_{i,j}-t_{i,j-1}} .\label{eq:realtime-eq2}
\end{align}
Combining Eq.~\eqref{eq:realtime-eq1} and Eq.~\eqref{eq:realtime-eq2} leads to
\begin{align}
	&\left|\frac{1}{n}\sum_{i=1}^n\sum_{j:X^{\mathrm{true}}_{i,j}\in e} W_{i,j} - \frac{1}{n}\sum_{i=1}^n T_{e,i}\right|\leq  \frac{1}{n}\sum_{i=1}^n\sum_{\substack{j=1,...,m_i+1:\\X^{\mathrm{true}}_{i,j}\in e \text{ or}\\X^{\mathrm{true}}_{i,j+1}\in e}}  t_{i,j}-t_{i,j-1}\nonumber\\
    & -  \frac{1}{n}\sum_{i=1}^n\sum_{\substack{j=1,...,m_i+1:\\X^{\mathrm{true}}_{i,j}\in e \text{ and}\\X^{\mathrm{true}}_{i,j+1}\in e}}  t_{i,j}-t_{i,j-1}.\label{eq:thm1-term(III)-before}
\end{align}
Recall that
\begin{align}
	\mathcal{I}_{e}=&\left\{(i,j):\mathbf{1}\{X^{\mathrm{true}}_{i,j}\in e\}\neq\mathbf{1}\{X^{\mathrm{true}}_{i,j-1}\in e\}\right\},\nonumber\\
	=& \{(i,j): X^{\mathrm{true}}_{i,j}\notin e, X^{\mathrm{true}}_{i,j-1}\in e \text{ or }X^{\mathrm{true}}_{i,j}\in e, X^{\mathrm{true}}_{i,j-1}\notin e\}\nonumber.
\end{align}
Eq.~\eqref{eq:thm1-term(III)-before} can be expressed as
\begin{align}
    \left|\frac{1}{n}\sum_{i=1}^n\sum_{j:X^{\mathrm{true}}_{i,j}\in e} W_{i,j} - \frac{1}{n}\sum_{i=1}^n T_{e,i}\right| \leq  \frac{1}{n}\sum_{i=1}^n \sum_{j:(i,j)\in \mathcal{I}_{e}} t_{i,j}-t_{i,j-1}.\label{eq:term3 in thm1}
\end{align}

\textbf{Term (IV)}: Given the boundedness of $T_{e,i}$ for each $i\in\{1,...,n\}$ and the independence of $\{T_{e,i}\}_{i=1}^n$, we obtain
\begin{align*}
	 \text{Var}\left(\frac{1}{n}\sum_{i=1}^n T_{e,i}\right) = \frac{1}{n}\text{Var}(T_{e,1}) \leq \frac{1}{n}\mathbb{E}(T_{e,1}^2) = O\left(\frac{1}{n}\right).
\end{align*}
Therefore, we have
\begin{align*}
&|\hat T_{e}-T_e| \\
=& O_p\left(\frac{1}{n}\sqrt{\sum_{i=1}^{n} \sum_{j=1}^{m_i} W_{i,j}^2 \pi_{i,j}}\right)+O\left(\frac{1}{n} \sum_{i=1}^n\sum_{j=1}^{m_i}W_{i,j} \pi_{i,j}\right)\nonumber\\
& +O\left(\frac{1}{n}\sum_{i=1}^n \sum_{j:(i,j)\in  \mathcal{I}_{e}} t_{i,j}-t_{i,j-1} \right)+
 O_p\left(\sqrt{\frac{1}{n}}\right)
\end{align*}
\end{proof}

\end{document}